\RequirePackage[l2tabu,orthodox]{nag}
\documentclass
[11pt,letterpaper]
{article} 

\usepackage[dvipsnames]{xcolor}
\usepackage[notes=false,later=false,camera=true,draft=false]{dtrt}
\usepackage[utf8]{inputenc}
\usepackage{etex}
\usepackage{ stmaryrd }
\usepackage{xspace,enumerate}
\usepackage[T1]{fontenc}
\usepackage[full]{textcomp}
\usepackage[american]{babel}
\usepackage{mathtools}

\usepackage{amsthm}
\usepackage{empheq}

   \usepackage{hyperref}
   
\hypersetup{
colorlinks=true,
urlcolor=Cerulean,
linkcolor=RoyalBlue,
citecolor=OliveGreen,
linktocpage=true,
}
\usepackage[capitalise,nameinlink]{cleveref}
\crefname{lemma}{Lemma}{Lemmas}
\crefname{fact}{Fact}{Facts}
\newcommand{\colorconstraints}{\text{Color Constraints}}
\crefname{colorconstraints}{(color constraints)}{Color Constraints}
\crefformat{colorconstraints}{#2\colorconstraints#3}
\crefname{indsetconstraints}{(indset constraints)}{IndSet Constraints}
\crefformat{indsetconstraints}{#2$\mathsf{IndSet\ Axioms}$#3}
\crefname{theorem}{Theorem}{Theorems}
\crefname{mtheorem}{Theorem}{Theorems}
\crefname{corollary}{Corollary}{Corollaries}
\crefname{claim}{Claim}{Claims}
\crefname{example}{Example}{Examples}
\crefname{algorithm}{Algorithm}{Algorithms}
\crefname{problem}{Problem}{Problems}
\crefname{definition}{Definition}{Definitions}
\usepackage{paralist}
\usepackage{turnstile}
\usepackage{mdframed}
\usepackage{tikz}
\usepackage{caption}
\DeclareCaptionType{Algorithm}
\usepackage{newfloat}
\newtheorem{theorem}{Theorem}[section]
\newtheorem{mtheorem}{Theorem}
\newtheorem*{theorem*}{Theorem}

\newtheorem*{proposition*}{Proposition}
\newtheorem{lemma}[theorem]{Lemma}
\newtheorem*{lemma*}{Lemma}

\newtheorem*{conjecture*}{Conjecture}
\newtheorem{fact}[theorem]{Fact}
\newtheorem*{fact*}{Fact}

\newtheorem*{hypothesis*}{Hypothesis}

\theoremstyle{definition}
\newtheorem{definition}[theorem]{Definition}
\newtheorem*{definition*}{Definition}

\theoremstyle{remark}
\newtheorem{claim}[theorem]{Claim}
\newtheorem*{claim*}{Claim}
\newtheorem{remark}[theorem]{Remark}
\newtheorem*{remark*}{Remark}
\newtheorem{observation}[theorem]{Observation}
\newtheorem*{observation*}{Observation}

\usepackage[
letterpaper,
top=1.2in,
bottom=1.2in,
left=1in,
right=1in]{geometry}

\usepackage{mathpazo}
\usepackage{textcomp} 
\usepackage[varg,bigdelims]{newpxmath}
\usepackage{bm} 
\let\mathbb\varmathbb
\usepackage{microtype}
\definecolor{petergreen}{rgb}{0, 0.75, 0}

\newcommand{\peter}[1]{\dtcolornote[Peter]{petergreen}{#1}}

\usepackage{thm-restate}

\allowdisplaybreaks
\newcommand{\FormatAuthor}[3]{
\begin{tabular}{c}
#1 \\ {\small\texttt{#2}} \\ {\small #3}
\end{tabular}
}

\newcommand{\keywords}[1]{\bigskip\par\noindent{\footnotesize\textbf{Keywords\/}: #1}}

\newcommand{\R}{{\mathbb R}}
\newcommand{\N}{{\mathbb N}}
\newcommand{\norm}[1]{\lVert #1 \rVert}
\newcommand{\boolnorm}[1]{{\lVert #1 \rVert}_{\infty \to 1}}
\newcommand{\abs}[1]{\lvert #1 \rvert}

\newcommand{\floor}[1]{\lfloor #1 \rfloor}

\newcommand{\eps}{\varepsilon}
\newcommand{\F}{{\mathbb F}}

\newcommand{\E}{{\mathbb E}}

\newcommand{\Tr}{\mathrm{Tr}}

\newcommand{\Bits}{\{0,1\}}

\newcommand{\Fits}{\{-1,1\}}

\newcommand{\cH}{\mathcal H}

\newcommand{\cB}{\mathcal B}
\newcommand{\cG}{\mathcal G}

\newcommand{\val}{\mathrm{val}}

\newcommand{\ceil}[1]{\lceil #1 \rceil}

\newcommand{\mper}{\,.}
\newcommand{\mcom}{\,,}
\newcommand{\Paren}[1]{\left(#1\right)}

\newcommand{\Norm}[1]{\left\lVert#1\right\rVert}

\newcommand{\polylog}{\mathrm{polylog}}
\newcommand{\cD}{\mathcal D}

\newcommand{\wt}{\mathrm{wt}}
\newcommand{\cV}{\mathcal{V}}

\newcommand{\Lincode}{\mathcal{L}}
\newcommand{\Code}{\mathcal{C}}

\newcommand{\cT}{\mathcal{T}}
\newcommand{\Dec}{\mathrm{Dec}}

\newcommand{\defeq}{:=}

\newcommand{\AND}{\mathrm{AND}}

\newcommand{\Enc}{\mathrm{Enc}}
\newcommand{\rand}{\mathrm{rand}}

\begin{document}


\title{Exponential Lower Bounds for Smooth 3-LCCs and Sharp Bounds for Designs}
\author{
\begin{tabular}[h!]{ccc}
      \FormatAuthor{Pravesh K.\ Kothari}{kothari@cs.princeton.edu}{IAS \& Princeton University}
      \FormatAuthor{Peter Manohar}{pmanohar@cs.cmu.edu}{Carnegie Mellon University}
\end{tabular}
} %
\date{\today}

\maketitle

\begin{abstract}
We give improved lower bounds for binary $3$-query locally correctable codes (3-LCCs) $\Code \colon \{0,1\}^k \rightarrow \{0,1\}^n$. Specifically, we prove: 
\begin{enumerate}[(1)]
    \item If $\Code$ is a linear \emph{design} 3-LCC, then $n \geq 2^{(1 - o(1))\sqrt{k} }$. A design 3-LCC has the additional property that the correcting sets for every codeword bit form a perfect matching and every pair of codeword bits is queried an equal number of times across all matchings. Our bound is tight up to a factor $\sqrt{8}$ in the exponent of $2$, as the best construction of binary $3$-LCCs (obtained by taking Reed--Muller codes on $\F_4$ and applying a natural projection map) is a design $3$-LCC with $n \leq 2^{\sqrt{8 k}}$. Up to a factor of $8$, this resolves the Hamada conjecture on the maximum $\F_2$-codimension of a $4$-design.
    \item If $\Code$ is a smooth, non-linear, adaptive $3$-LCC with perfect completeness, then, $n \geq 2^{\Omega(k^{1/5})}$. 
    \item If $\Code$ is a smooth, non-linear, adaptive $3$-LCC with completeness $1 - \eps$, then $n \geq \tilde{\Omega}(k^{\frac{1}{2\eps}})$. 
    In particular, when $\eps$ is a small constant, this implies a lower bound for general non-linear LCCs that beats the prior best $n \geq \tilde{\Omega}(k^3)$ lower bound of~\cite{AlrabiahGKM23} by a polynomial factor.\peter{here}
\end{enumerate}

Our design LCC lower bound is obtained via a fine-grained analysis of the Kikuchi matrix method applied to a variant of the matrix used in~\cite{KothariM23}. Our lower bounds for non-linear codes are obtained by designing a from-scratch reduction from nonlinear 3-LCCs to a system of ``chain XOR equations'' --- polynomial equations with similar structure to the \emph{long chain derivations} that arise in the lower bounds for linear $3$-LCCs~\cite{KothariM23}.

\end{abstract}

\clearpage
 \microtypesetup{protrusion=false}
  \tableofcontents{}
  \microtypesetup{protrusion=true}

\clearpage

\pagestyle{plain}
\setcounter{page}{1}

\section{Introduction}
A \emph{locally correctable code} (LCC) is an error correcting code that admits, in addition, a \emph{local} correction (a.k.a. \emph{self-correction}) algorithm that can recover any symbol of the original codeword by querying only a small number of randomly chosen symbols from the received corrupted codeword. More formally, we say that a code $\Code \colon \Bits^k \to \Bits^n$ is $q$-locally correctable if for any codeword $x$, a corruption $y$ of $x$, and input $u \in [n]$, the local correction algorithm reads at most $q$ symbols (typically a small constant such as $2$ or $3$) of $y$ and recovers the bit $x_u$ with probability $1 -  \eps$ whenever $\Delta(x, y) \coloneqq \abs{\{v \in [n] : x_v \ne y_v\}} \leq \delta n$, where $\delta$, the ``distance'' of the code, and $\eps$, the decoding error, are constants. Such codes 
have had myriad applications (see the surveys~\cite{Trevisan04,Yekhanin12,Dvir12}) starting with \emph{program checking}~\cite{BlumK95},   sublinear algorithms and property testing~\cite{RubinfeldS96,BlumLR93}, probabilistically checkable proofs~\cite{AroraLMSS98,AroraS98}, \textsc{IP}=\textsc{PSPACE}~\cite{LundFKN90,Shamir90}, worst-case to average-case reductions~\cite{BabaiFNW93}, constructions of explicit rigid matrices~\cite{Dvir10}, and $q$-server private information retrieval protocols~\cite{IshaiK99,BarkolIW10}. 

Reed--Muller codes, or codes based on the evaluations of multivariate degree $q-1$ polynomials over a finite field, provide a natural class of $q$-query locally correctable codes. For any constant $q\geq 2$, they imply binary $q$-LCCs with block length $n \leq 2^{ O(k^{1/(q-1)})}$. These codes are in fact $\F_2$-linear if we view the code $\Code$ as mapping $\F_2^k$ into $\F_2^n$. Despite significant effort over the past three decades, we do not know of a binary $q$-LCC with a smaller block length than Reed--Muller codes. This has motivated the conjecture that Reed--Muller codes are optimal $q$-LCCs for any constant $q$. 

For $q=2$, classical works~\cite{KerenidisW04, GoldreichKST06} on lower bounds on local codes confirm that the $2$-LCC based on binary Hadamard codes (the special case of Reed--Muller codes when $q=2$) achieves the smallest possible block length of $n = 2^k$ up to absolute constants in the exponent. For $q = 3$, a recent work of~\cite{KothariM23} improved on the best prior lower bound of $n \geq \tilde{\Omega}(k^3)$~\cite{AlrabiahGKM23}, and showed that for any binary\footnote{Their results extends to codes over any field $\F$ of size $\abs{\F} \leq k^{1 - \eta}$ for a constant $\eta > 0$, more generally.} linear code, $n \geq 2^{\Omega(k^{1/8})}$. As a corollary, they obtained a strong separation between $3$-LCCs and the weaker notion of $3$-query locally \emph{decodable} codes --- codes where the local correction algorithm only needs to succeeds for the $k$ message bits and for which sub-exponential length constructions, i.e., $n = 2^{k^{o(1)}}$, are known~\cite{Yekhanin08, Efremenko09}. The lower bound of~\cite{KothariM23} was very recently improved to $n \geq 2^{\Omega(k^{1/4})}$ in a follow-up work by Yankovitz~\cite{Yankovitz24}.

Despite this substantial progress, these results strongly exploit the linearity of the codes and do not yield \emph{any} improvement over the prior cubic lower bound for non-linear codes of~\cite{AlrabiahGKM23}. Even for the case of linear codes, these bounds, while exponential, still do not asymptotically match the block length of Reed--Muller codes. In this work, we make progress on both these fronts. Before discussing our results, we will take a brief detour to discuss a connection between $3$-LCCs and a foundational question about the algebraic rank of combinatorial designs.  

\parhead{Connections to the Hamada Conjecture.} Locally correctable codes have a deep connection --- first formalized by Barkol, Ishai and Weinreb~\cite{BarkolIW10} --- to the widely open Hamada conjecture from the 1970s in combinatorial design theory (with deep connections to coding theory, see~\cite{AssmusKey1992}
 for a classical reference). For positive integers $m,s,\lambda$, a $2$-$(m,s,\lambda)$-design is a collection $\cB \subseteq [m]$ of subsets (called \emph{blocks}) of size $s$, such that every pair of elements in $[m]$ appears in exactly $\lambda$ subsets in $\cB$. For any prime $p$, the $p$-rank of a design $\cB$ is the rank, over $\F_p$, of the \emph{incidence} matrix of $\cB$: the $0$-$1$ matrix with rows labeled by elements of $[m]$, columns labeled by elements of $\cB$ and an entry $(i,B)$ is $1$ iff $B$ contains $i$. A central question in algebraic design theory is understanding the smallest possible $p$-rank of a $2$-$(m,s,\lambda)$-design. 

In~\cite{BarkolIW10}, the authors showed that given any $2$-$(m,s,\lambda)$-design $\cD$ of $p$-rank $m - k$, the dual subspace to the column space of the incidence matrix of $\cD$ yields a linear $(s-1)$-query locally correctable code on $\F_p^m$ of dimension $k$. In particular, applying this transformation to the well-studied \emph{geometric designs} yields the folklore construction of Reed--Muller locally correctable codes discussed earlier. Specifically, the $3$-query locally correctable \emph{binary} code obtained from Reed--Muller codes on $\F_4$ corresponds to a $2$-$(n,4,1)$-design over $\F_2$ (see \cref{app:reedmuller}).

In 1973, Hamada~\cite{Hamada73} made a foundational conjecture (see~\cite{MR2860656} for a recent survey) in the area that states\footnote{Hamada's original conjecture is that affine geometric designs, or, dual codes to Reed--Muller codes, are the unique optimal designs with the same parameters. This strong form has since then been disproved -- there are non-affine geometric designs that achieve the \emph{same} (but not better!) parameters~\cite{Jungnickel84,Kantor94,LamLT00,LamLT01,LamT02,MR2480694}. The version of the problem we study here is called the \emph{weak version} of Hamada conjecture. 
} that affine geometric designs (i.e., duals to the Reed--Muller LCCs) minimize the $p$-rank among all algebraic designs of the same parameters. Over the past few decades, the conjecture has been confirmed in various special cases~\cite{HamadaO75,DoyenHV78,Teirlinck80,Tonchev99} that all correspond to $s\leq 3$ or $s= n-1$. In particular, the case of $s= 4$ (the setting of $3$-LCCs) was widely open until the recent result of~\cite{KothariM23} for $3$-LCC lower bounds. The connection between Hamada's conjecture and LCC lower bounds was suggested in~\cite{BarkolIW10} as evidence for \emph{the difficulty} of proving LCC lower bounds.

\subsection{\cref{mthm:designs}: Sharp lower bounds for design $3$-LCCs.} 
In our first result, we obtain a bound that is sharp up to a $\sqrt{8}$ factor in the exponent on the blocklength of any binary linear $3$-LCC where the local correction query sets form a $2$-$(n,4,1)$-design. This is equivalent to asking for the local correction sets for correcting any bit of the codeword to be a \emph{perfect} $3$-uniform hypergraph matching and that every pair of codewords bits appears in exactly $2$ triples across all matchings\footnote{The reason that this is $2$ instead of $\lambda = 1$ is because a $4$-tuple $(u,v,s,t)$ yields $2$ decoding triples, $(u,s,t)$ for $v$ and $(v,s,t)$ for $u$, that contain the pair $(s,t)$.}. Specifically, for such design $3$-LCCs, we prove:
\begin{mtheorem}
\label{mthm:designs}
Let $\Lincode \colon \Bits^k \to \Bits^n$ be a design $3$-LCC. Then, $n \geq 2^{(1-o(1))\sqrt{k}}$. Here, the $o(1)$-factor is $O(\log k/\sqrt{k})$.
\end{mtheorem}
\cref{mthm:designs} improves on the prior best lower bound of $n \geq 2^{\Omega(k^{1/3})}$ for designs recently obtained by Yankovitz~\cite{Yankovitz24} building on the $n \geq 2^{\Omega(k^{1/6})}$ bound of~\cite{KothariM23}.\footnote{The stated result of~\cite{KothariM23} is $n \geq 2^{\Omega(k^{1/8})}$ for (non-design) linear $3$-LCCs; however, their proof implicitly gives this slightly stronger bound for designs.} We note that there is a technical bottleneck that prevents the proof of~\cite{Yankovitz24} from beating a lower bound of $n \geq 2^{\Omega(k^{1/3}})$ even for the case of designs that \cref{mthm:designs} tackles (see \cref{rem:newmatrix} for a more detailed explanation). 

Reed--Muller codes, in particular, are design LCCs. In fact, in~\cref{app:reedmuller} we observe that the folklore best-known construction of binary $3$-query LCCs --- obtained by projecting Reed--Muller codes of degree-$2$ polynomials over $\F_4$ to $\F_2$ via the trace map --- is a design $3$-LCC with $n \leq 2^{\sqrt{8k}}$, or equivalently, a $2$-$(n,4,1)$ design of rank $n - k$.
Thus, the bound in \cref{mthm:designs} is \emph{tight up to a factor of $\sqrt{8}$ in the exponent.} As a direct corollary, we also confirm the Hamada conjecture for $2$-$(n,4,1)$-designs up to a factor of $8$ in the co-dimension.

\parhead{Towards obtaining a $k \leq O(\log^2 n)$ bound for all linear $3$-LCCs.} Given our almost sharp lower bound for design $3$-LCCs, we can use our proof to attribute the ``extra'' $\log^2 n$ factor loss in~\cite{KothariM23,Yankovitz24} to certain ``irregularities'' of general linear $3$-LCCs. Concretely, there are two places where the proof of~\cite{KothariM23,Yankovitz24} is lossy: \begin{inparaenum}[(1)]\item there is a ``hypergraph decomposition'' step to handle that pairs of codeword bits may appear in $\gg O(1)$ triples across all matchings (one $\log n$ loss), and \item there is a ``row pruning'' step to argue that a certain graph is approximately regular (one $\log n$ loss)\end{inparaenum}. For the case of designs,~\cite{Yankovitz24} proves a $k \leq O(\log^3 n)$ bound since there is no ``hypergraph decomposition'' needed for designs as each pair of codeword bits appears in a bounded number of triples. Our proof of \cref{mthm:designs} additionally shows that because the hypergraph matchings in the design are \emph{perfect}, we can (via this work's modified approach, see \cref{rem:newmatrix}) mitigate the $\log n$ factor loss in the ``row pruning'' step. 

The sharpness of our bound for design $3$-LCCs may be somewhat surprising because removing an analogous ``last" $\log n$ factor in the hypergraph Moore bound (also proved via the Kikuchi matrix method) and related problems remains a challenging problem~\cite{GuruswamiKM22,HsiehKM23,HsiehKMMS24}. In this setting, making extra structural assumptions about the hypergraph, analogous to the additional structure of design $3$-LCCs, does not seem to help.

\subsection{\cref{mthm:nonlin}: Exponential lower bounds for smooth $3$-LCCs.}
In our second result, we obtain improved lower bounds for smooth $3$-LCCs with high completeness. These codes may be non-linear and may have adaptive correction algorithms.\peter{here}

A $3$-LCC is said to be $\delta$-\emph{smooth} if no codeword bit is queried with probability more than $\frac{1}{\delta n}$ on any particular invocation of the decoder. Introduced by Katz and Trevisan~\cite{KatzT00}, smooth codes provide a clean formalization of general locally correctable/decodable codes. We say that such a code has completeness $1-\epsilon$, if, when running the $\delta$-smooth local correction algorithm on an \emph{uncorrupted} codeword, the algorithm succeeds with probability at least $1-\epsilon$. Recall that the usual notion of completeness (e.g., in~\cite{KatzT00}) for LCCs is for an input with a $\delta$-fraction of corruptions.

Our result shows that for any $(1-\epsilon)$-complete $\delta$-smooth code where $\delta$ is a constant, $n \geq k^{O(1/\epsilon)}$. In particular, when $\epsilon \leq 1/\log n$, we obtain a exponential lower bound on the block length, and as $\eps$ approaches $0$ the bound becomes $n \geq 2^{\Omega(k^{1/5})}$.

\begin{mtheorem} 
\label{mthm:nonlin}
There is an absolute constant $\gamma > 0$ such that the following holds. Let $\Code \colon \Bits^k \to \Bits^n$ be a $\delta$-smooth (possibly non-linear and adaptive) 3-LCC with completeness $1 - \eps$. Then for any $\eta \in (0,1)$, it holds that $k \leq \frac{\log(1/\delta)}{\eta^4 \delta^3} \cdot O(n^{\frac{1}{r}} \log^5 n)$, where $r = \min(\floor{\frac{1 - \eta}{2 \eps}}, \log_2 n)$.

In particular, if $\delta$ is a constant and $\eps = 0$, then $k \leq O(\log^5 n)$, i.e., $n \geq 2^{\Omega(k^{1/5})}$, and if $\eps > 0$ is a small constant and $1/(2\eps)$ is not an integer, then taking $\eta = 1/\log n$ implies that $k \leq \tilde{O}(n^{2 \eps})$, i.e., $k \geq \tilde{\Omega}(k^{\frac{1}{2\eps}})$.
\end{mtheorem}
As we shall discuss towards at the end of this section, \cref{mthm:nonlin} implies a lower bound for general $(3, \delta, \eps)$-LCCs that beats the prior best $n \geq \tilde{\Omega}(k^3)$ lower bound of~\cite{AlrabiahGKM23} by a polynomial factor when $\eps$ is a small constant. Moreover, in the case of near-perfect completeness, our result above obtains the first exponential lower bound for (possibly adaptive and non-linear) smooth $3$-LCCs. 

Our proof is based on the method of spectral refutation via Kikuchi matrices (first introduced in~\cite{WeinAM19} for an application to Gaussian tensor PCA and refutation of random $k$-XOR instances of even arity) developed in prior works~\cite{GuruswamiKM22,HsiehKM23,AlrabiahGKM23,KothariM23}. The key idea in this method is to associate the existence of a combinatorial object (e.g., a $3$-LCC) to the satisfiability of a family of XOR formulas and find a spectral \emph{refutation} (i.e., certificate of unsatisfiability) for a randomly chosen member of the family. 
Unlike the works of~\cite{KothariM23,Yankovitz24}, which only prove lower bounds for linear codes with an argument that can be reformulated to be entirely combinatorial, the proof of \cref{mthm:nonlin} crucially uses the power of spectral refutation.

The proof of \cref{mthm:nonlin} requires new conceptual ideas. The first immediate observation is that the standard reduction of~\cite{KatzT00} to nonadaptive, linear decoders that succeeds in expectation loses a large factor in the completeness parameter. So, to prove \cref{mthm:nonlin} we need to come up with a new reduction \emph{from scratch}. Our reduction is based on two new key ideas. First, we \emph{exactly} encode the behavior of the LCC decoder on a particular input index $u$ as a degree $\leq 3$ polynomial. Effectively, the constraints we uncover from the decoder are ``$\AND$'' constraints:\begin{equation*}
``x_{v_1} = a_1 \wedge x_{v_2} = a_2 \implies x_u = x_{v_3} \mcom''
\end{equation*}
rather than the linear constraints $x_{v_1} + x_{v_2} + x_{v_3} = x_u$ encountered for linear codes. 
Second, we execute the ``long chain derivation'' strategy of~\cite{KothariM23} by forming chains by (1) first splitting the degree $\leq 3$ polynomial that represents the decoder into the homogeneous degree $3$ component and the degree $\leq 2$ component, and then (2) forming chains by replacing the third query $v_3$ of the decoder in the ``homogeneous degree $3$ component'' by an invocation of the decoding algorithm for $v_3$ (and then iterating); we call such chains ``adaptive chains''. \peter{here}
Combining these two ideas allows us to write a ``chain polynomial'' that plays the role of the ``long chain derivations'' in~\cite{KothariM23}. Refuting this ``chain polynomial'' using spectral bounds on Kikuchi matrices yields \cref{mthm:nonlin}.

\parhead{Smooth vs.\ general LCCs} Smooth LCCs (\cref{def:smoothLCC}) were defined in the work of~\cite{KatzT00}, motivated by their connection to general LCCs (\cref{def:LCC}). A simple reduction in~\cite{KatzT00} shows that any $(3, \delta,\eps)$-LCC, i.e., an LCC with distance $\delta$ and completeness $1 - \eps$, can be turned into a $(3, \delta/3, \eps)$-smooth LCC, i.e., a $\delta/3$-smooth $3$-LCC with completeness $1 - \eps$. Conversely, any $(3, \delta, \eps)$-smooth LCC is a $(3, \eta \delta, \eps + \eta)$-LCC for any $\eta > 0$.

Thus, when $\eps$ is a small constant, \cref{mthm:nonlin} implies a lower bound for general $(3, \delta, \eps)$-LCCs that beats the prior best $n \geq \tilde{\Omega}(k^3)$ lower bound of~\cite{AlrabiahGKM23} by a polynomial factor.

However, in the setting of perfect completeness (and $\eps = o(1)$ more generally), the comparison between smooth LCCs and general LCCs begins to break down. This is because, for a general LCC, $\delta$ is the fraction of errors one can tolerate while still decoding correctly with probability $1 - \eps$; the parameters $\delta$ and $\eps$ are coupled! In particular, it is likely not possible to simultaneously have $\eps = o(1)$, $\delta = O(1)$ and $q = O(1)$. On the other hand, for a smooth LCC, $\delta$ is the smoothness parameter, and $1 - \eps$ is the probability that the decoder succeeds \emph{on an uncorrupted codeword}. Thus, for smooth codes, it is perfectly sensible to set $\delta = O(1)$, $\eps = 0$, and $q = O(1)$.

In retrospect, the definition of LCCs inherently couples $\delta$ and the completeness $\eps$, whereas for smooth codes these parameters become independent. In particular, a smooth code allows us to seamlessly trade off between the fraction of errors $\eta \delta$ tolerated and the success probability $1 - \eps - \eta$ of the decoder in the presence of this fraction of errors. For this reason, a smooth code is a stronger object, but also perhaps a more natural one. 

Indeed, in some important applications of LDCs/LCCs, smooth LDCs/LCCs are the right notion to consider. For example, a \emph{perfectly smooth} $(q, 1, \eps)$-smooth LDC gives a $q$-server information-theoretically secure private information retrieval scheme with completeness $1 - \eps$.

The subtle definitional issues above did not affect prior lower bound (or upper bound) techniques. Indeed, known constructions of $q$-LDCs and LCCs are \emph{perfectly} smooth and satisfy \emph{perfect} completeness, i.e., $(q, 1, 0)$-smooth LDCs/LCCs, and the lower bound techniques of~\cite{KatzT00, KerenidisW04, AlrabiahGKM23} (that is, the best known lower bounds before the work~\cite{KothariM23}) succeed for smooth LDCs/LCCs even with \emph{low} completeness.

\parhead{Concurrent work.} In concurrent and independent work, \cite{AlrabiahG24} proves an $n \geq 2^{\Omega(\sqrt{k/\log k})}$ lower bound for all linear $3$-LCCs over $\F_2$, improving on the $2^{\Omega(k^{1/4})}$ shown in \cite{Yankovitz24}. This is incomparable to \cref{mthm:designs}, as it proves a weaker (and possibly not tight) lower bound, as compared to the sharp statement in \cref{mthm:designs}, but it applies for all linear $3$-LCCs over $\F_2$, not just design $3$-LCCs. The work of~\cite{AlrabiahG24} does not prove any lower bound for nonlinear codes.



\section{Proof Overview}
\label{sec:techniques}
In this section, we summarize the key conceptual ideas that we use in the proofs. We start by recalling the approach of~\cite{KothariM23} for proving lower bounds for \emph{linear} $3$-LCCs. Then, we explain the technical barriers to proving \cref{mthm:designs} encountered in the works of~\cite{KothariM23,Yankovitz24}. Finally, we discuss our approach to handling the nonlinear case.

\subsection{The approach of~\cite{KothariM23}}
The proof of~\cite{KothariM23} gives a transformation that takes any linear $3$-LCC $\Lincode \colon \Bits^k \to \Bits^n$ and turns it into a $2$-LDC $\Lincode' \colon \Bits^k \to \Bits^N$, where $N = n^{O(\polylog(n))}$. By applying known $2$-LDC lower bounds (\cref{fact:2ldclb}), we can then conclude that $k \leq O(\polylog(n)) \cdot \log n$. Obtaining better lower bounds thus boils down to optimizing the $\polylog(n)$ factor here and/or removing the extra $\log n$ factor from \cref{fact:2ldclb}.

For intuition, let us think of $N$ as $N = {n \choose s}^r$ for some choice of parameters $s$ and $r$, and $\Lincode'$ as the very simple transformation: for sets $(S_1, \dots, S_r)$, each in ${[n] \choose s}$, we set $\Lincode'(b)_{(S_1, \dots, S_r)} = \sum_{h = 1}^r \sum_{v \in S_h} x_v$, where $x = \Lincode(b)$. Namely, we just take the XOR of the bits across all the sets. If we can show that $\Lincode'$ is indeed a $2$-LDC, then we can apply known $2$-LDC lower bounds (\cref{fact:2ldclb}) to conclude that $k \leq O(\log N) = O(r s \log n)$. If we can then take $rs = O(\log n)$, or $r = s = O(\log n)$ while removing the extra $\log n$ factor from \cref{fact:2ldclb}, we will get an $O(\log^2 n)$ bound, i.e, $n \geq 2^{\Omega(\sqrt{k})}$.

Recall that in a linear $3$-LCC, we are given $3$-uniform hypergraph matchings $\{H_u\}_{u \in [n]}$, such that for each $C \in H_u$, $\sum_{v \in C} x_v = x_u$ for all codewords $x \in \Lincode$. To show that $\Lincode'$ is a $2$-LDC, we need to find, for each $i \in [k]$, many pairs of vertices $((S_1, \dots, S_r), (T_1, \dots, T_r))$ such that $\sum_{h = 1}^r \sum_{v \in S_h} x_v + \sum_{h = 1}^r \sum_{v \in T_h} x_v = b_i$ for all $x = \Lincode(b)$. The key idea of \cite{KothariM23} is to build many such constraints by building \emph{long chain derivations} out of the original constraints $H_u$.

\begin{definition}[$r$-chains]
\label{def:chainintro}
Let $H_1, \dots, H_n$ be the $3$-uniform hypergraph matchings defined from the $3$-LCC $\Lincode$. An $r$-chain with \emph{head $u_0$} is an ordered sequence of vertices of length $3r + 1$, given by $C = (u_0, v_1, v_2, u_1, v_3, v_4, u_2, \dots, v_{2(r-1) + 1}, v_{2(r-1) + 2}, u_r)$ and for each $h = 0, \dots, r-1$, it holds that $\{v_{2h + 1}, v_{2h + 2}, u_{h+1}\} \in H_{u_h}$. We let $\cH^{(r)}_u$ denote the set of $r$-chains with head $u$.

We let $C_L = (v_1, v_3, v_5, \dots, v_{2(r-1) + 1})$ denote the ``left half'' of the chain, and $C_R = (v_2, v_4, v_6, \dots, v_{2(r-1) + 2})$ denote the ``right half''. We call $u_r$ the ``tail''.

The number of chains $\abs{\cH_u^{(r)}}$ is at most $(6 \delta n)^r$.
\end{definition}
The $2$-chains can be interpreted as deriving constraints by taking two constraints $x_{v_1} + x_{v_2} +  x_{u_1} = x_{u_0}$ in $H_{u_0}$ and $x_{v_3} + x_{v_4} +  x_{u_2} = x_{u_1}$ in $H_{u_1}$, and then adding them together to produce the constraint $x_{v_1} + x_{v_2} +  x_{v_3} + x_{v_4} +  x_{u_2} = x_{u_0}$; the set of $r$-chains is formed by repeating this operation. We have $(6 \delta n)^r$ chains in $\cH_u^{(r)}$ because there are $6 \delta n$ \emph{ordered} hyperedges in $H_u$.

The next idea of~\cite{KothariM23} is to use a Kikuchi graph to (1) define $N$ and the map $\Lincode'$, and (2) define the $2$-LDC decoding constraints for $\Lincode'$. For this overview, we will start with the following graph due to~\cite{Yankovitz24}, which is a bit simpler than the graphs used in~\cite{KothariM23} as it saves a use of the Cauchy-Schwarz inequality.

\begin{definition}[Imbalanced Kikuchi graph]
\label{def:kikuchiintro}
Let $s$ be a parameter, and let $G_u$ be the graph with left vertex set $L = {[n] \choose s}^r \times [n]$ and right vertex set $R = {[n] \choose s}^r$. For a chain $C \in \cH_u^{(r)}$, we add an edge between $((S_1, \dots, S_r),w)$ and $(T_1, \dots, T_r)$ in $G_u$ ``labeled'' by $C$ if $w = u_r$ and for each $h = 1, \dots, r$, we have $S_h = \{v_{2(h-1)+1}\} \cup U_h$ 
and $T_h = \{v_{2(h-1) + 2}\} \cup U_h$
for some $U_h \subseteq [n] \setminus \{v_{2(h-1)+1}, v_{2(h-1)+2}\}$
of size $s - 1$. Two distinct chains may produce the same edge --- we add edges with multiplicity.
\end{definition}
To show that $\Lincode'$, defined now as a map from $\Bits^k \to \Bits^{L \cup R}$ in an analogous way, is a $2$-LDC, we need to show that for each $u$, $G_u$ admits a large matching. An obvious barrier to this is that the graph is bipartite and imbalanced, and so the largest matching can only have size at most $\abs{R} = \abs{L}/n$. This be fixed with a simple trick: for each right vertex $(T_1, \dots, T_r)$, we can add $n$ copies of the vertex to the graph and then split its edges evenly across the copies, thereby decreasing the degree by a factor of $n$.\footnote{Technically, the degree might not be divisible exactly by $n$, so reduces the degree by a factor of $n(1-o(1))$, which is sufficient.}

\parhead{Extracting a large matching from $G_u$.} Let us now explain the approach of~\cite{KothariM23} to show that $G_u$ admits a large matching. We note that this is the \emph{key technical difficulty} in the proof of~\cite{KothariM23}, and in some sense, this has to be the difficult step because it will prove that $\Lincode'$ is a $2$-LDC! 

Let $d_{u,L}$ denote the average left degree of $G_u$, and let $d_{u,R}$ denote the average right degree. For $G_u$ to have a large matching, it should, at the very least, have at least $\abs{L} = n {n \choose t}^r$ edges! 

Some simple combinatorics shows that each chain $C \in \cH_u^{(r)}$ contributes exactly ${n - 2 \choose s - 1}^r$ edges to the graph $G_u$. Therefore, as long as we have
\begin{flalign*}
d_{u,L} = \abs{\cH_u^{(r)}} \frac{{n - 2 \choose s - 1}^r}{n {n \choose s}^r} = (1 \pm o(1)) (6 \delta n)^r \frac{1}{n} \left(\frac{s}{n}\right)^{r} = (1 \pm o(1)) (6 \delta s)^r \frac{1}{n} \gg 1 \mcom
\end{flalign*}
then we can hope to find a large matching. Note that for this to hold, we \emph{must} set $r = O(\log n/\log s)$.

One simple way to find a matching in $G_u$ is to argue that $G_u$ is approximately regular, meaning that most vertices have degree $\leq O(d_{u,L})$. If this were the case, then (after the ``vertex splitting trick'') we get a matching of size $\abs{E(G_u)}{O(d_{u,L})} \geq \Omega(\abs{L})$, which would finish the proof. Unfortunately, this is not true: there can be left (right) vertices in $G_u$ of degree $\gg d_{u,L}$ ($\gg d_{u,R}$). The ``row pruning'' strategy of~\cite{KothariM23} is to show that such vertices are rare so that by removing them we obtain a graph $G'_u$ with $\Omega(\abs{E(G_u)})$ edges that has bounded left degree $\leq O(d_{u,L})$ and bounded right degree $\leq O(d_{u,R})$.

More formally, the proof of~\cite{KothariM23} uses a form of a Kim--Vu concentration inequality~\cite{KimV00, SchudyS12} for polynomials to argue that, with high probability, a \emph{random} left (right) vertex has degree at most $O(d_{u,L})$ ($O(d_{u,R})$), which finishes the proof. This is the key technical ``row pruning'' step in the proof, and the proof uses the moment method with high moments. Unfortunately, to prove this,~\cite{KothariM23} requires $s = O(r^3 \log n)$, which loses several extra $\log n$ factors.

The clever trick of~\cite{Yankovitz24}, when phrased in the language of probability, can be interpreted as follows: we can achieve the same goals by only bounding the \emph{second} moments of the degrees instead of the higher moments. This is similar in spirit to the ``edge deletion'' technique of~\cite{HsiehKM23}, which used a similar argument to remove some of the $\log n$ factors from a different ``row pruning'' argument of~\cite{GuruswamiKM22} that used Kim--Vu concentration inequalities~\cite{KimV00, SchudyS12} in the context of CSP refutation.

Specifically, if $\deg_{u,L}(S_1, \dots, S_r, w)$ is the left degree of the vertex $(S_1, \dots, S_r, w)$ and $\deg_{u,R}(T_1, \dots, T_r)$ is the right degree of $(T_1, \dots, T_r)$, then~\cite{Yankovitz24} shows that 
\begin{flalign*}
&\E_{S_1, \dots, S_r, w}[\deg_{u,L}(S_1, \dots, S_r, w)^2] \leq (1 + o(1)) d_{u,L}^2 \mcom\\
&\E_{T_1, \dots, T_r}[\deg_{u,R}(T_1, \dots, T_r)^2] \leq (1 + o(1)) d_{u,R}^2 \mcom
\end{flalign*}
where $d_{u,L}$ and $d_{u,R}$ are the first moments, and we only need $s = \Gamma r$ and $r = O(\log n)$, where $\Gamma$ is a large enough constant, for this to hold.

Thus, applying Chebyshev's inequality, we can show that after removing a small number of vertices, we can find a subgraph $G'_u$ of $G_u$ where each vertex has left degree $\leq (1+o(1)) d_{u,L}$ and right degree $\leq (1+o(1))d_{u,R}$, which, after a few more straightforward steps, finishes the proof.

In total, the final lower bound is $k \leq O(rs \log n) = O(\log^3 n)$, as we have set $r, s = O(\log n)$.

\begin{remark}
This sketches the proof of Theorem 1.6 in~\cite{Yankovitz24} for design $3$-LCCs. We note that the reason~\cite{Yankovitz24} obtains a weaker $k \leq O(\log^4 n)$ bound for linear $3$-LCCs is because a general $3$-LCC can have ``heavy pairs'' --- pairs of variables $(v_1, v_2)$ that appear in many hyperedges across all the $H_u$'s --- and this loses the extra $\log n$ factor. Indeed, overcoming this issue to produce \emph{any} lower bound better than the $k \leq \tilde{O}(n^{1/3})$ of~\cite{AlrabiahGKM23} is one of two key technical difficulties in~\cite{KothariM23} (the above row pruning argument is the other one). However, as \cref{mthm:designs} is only for designs, this issue does not arise.
\end{remark}

\subsection{Tight bounds for designs: proof sketch of \cref{mthm:designs}}
\label{sec:techsdesigns}
To beat the $O(\log^3 n)$ of~\cite{Yankovitz24} for designs and get $O(\log^2 n)$, we need to find a $\log n$ factor to remove. At the very least, we know  we cannot hope to take $r$ much smaller than $\log n$, as we need $s^r \gg n$ for the entire approach to even make sense. So, there are two possibilities: either we can take $s = O(1)$, or the $O(r s \log n)$ bound coming from \cref{fact:2ldclb} is not tight for the $2$-LDC that we produce, and the truth is really $O(r s)$. Let us now investigate the first case, as if we could take $s = O(1)$ this would be the easiest route towards proving \cref{mthm:designs}.

\parhead{Second moments of degrees are large for $s \ll r$.} Unfortunately, as we shall show, we cannot take $s = O(1)$, or even $s$ much smaller than $r$, and still have the following moment bounds:
\begin{flalign*}
&\E_{S_1, \dots, S_r, w}[\deg_{u,L}(S_1, \dots, S_r, w)^2] \leq O(d_{u,L}^2) \mcom\\
&\E_{T_1, \dots, T_r}[\deg_{u,R}(T_1, \dots, T_r)^2] \leq O(d_{u,R}^2) \mper
\end{flalign*}
Indeed, let us compute $\E_{T_1, \dots, T_r}[\deg_{u,R}(T_1, \dots, T_r)^2]$. Recall that we will compare it to $d_{u,R}^2$, and we have already computed $d_{u,R} = n \cdot d_{u,L} = (1 \pm o(1)) (6 \delta s)^r$.
To do the computation, we will need to use the number of pairs of chains $C, C' \in \cH_u^{(r)}$ with $\abs{C_R \cap C'_R} = t$ is at most ${r \choose t} \cdot 2^{2r} (3 \delta n)^{2r - t}$. Let us denote $C_R = (v_2, v_4, v_6, \dots, v_{2(r-1) + 2})$ and $C'_R = (v'_2, v'_4, v'_6, \dots, v'_{2(r-1) + 2})$
\begin{flalign*}
&\E_{T_1, \dots, T_r}[\deg_{u,R}(T_1, \dots, T_r)^2] \leq \sum_{t = 0}^r \sum_{C,C' \in \cH_u^{(r)} : \abs{C_R \cap C'_R} = t} \Pr[v_{2h +2}, v'_{2h+2} \in T_{h+1} \ \forall h \in \{0, \dots, r-1\}] \\
&\leq \sum_{t = 0}^r \sum_{C,C' \in \cH_u^{(r)} : \abs{C_R \cap C'_R} = t} \frac{ {n \choose s - 1}^t {n \choose s - 2}^{r - t}}{ {n \choose s}^{t}} \ \text{(if $v_{2h+2} = v'_{2h+2}$, then $T_h$ has $\leq {n \choose s-1}$ choices, else $\leq {n \choose s - 2}$ choices)} \\
&\leq \sum_{t = 0}^r \sum_{C,C' \in \cH_u^{(r)} : \abs{C_R \cap C'_R} = t} (1+o(1)) \left(\frac{s}{n}\right)^t \left(\frac{s}{n}\right)^{2r - 2t} \ \ \text{(by binomial coefficient estimates)} \\
&\leq (1 + o(1)) \sum_{t = 0}^r {r \choose t} \cdot 2^{2r} (3 \delta n)^{2r - t} \left(\frac{s}{n}\right)^{2r - t} \ \ \text{(from the bound on number of pairs $(C,C')$)} \\
&\leq (1 + o(1)) (6 \delta s)^{2r} \sum_{t = 0}^r {r \choose t} (3 \delta s)^{- t} \\
&\leq (1 + o(1)) d_{u,R}^2 \sum_{t = 0}^r {r \choose t} (3 \delta s)^{- t} \mper
\end{flalign*}
The problem with the moment bound is now readily apparent. For small $t$, ${r \choose t}$ is roughly $r^t$, and so we need $s \gg r/3\delta$ so that $(3 \delta s)^{-t} r^t \ll 1$. Hence, if we take $s = o(r)$, the second moment is large. 

Still, one may wonder if our estimate on the second moment here is tight. Perhaps it is simply that our bound on the moment is large, but the true second moment is not. It turns out that when the $H_u$'s are near-perfect matchings (which is the case for design $3$-LCCs!), i.e., $3 \delta = 1 - o(1)$, then this bound is tight up to a $1 + o(1)$ factor. This implies that the right degrees truly have high variance (a similar calculation shows this for left degree also). So, it is unlikely that one can find a subgraph $G'_u$ with $\Omega(\abs{E(G_u)})$ edges and, say, right degree bounded by $O(d_{u,R})$.

\begin{remark}
\label{rem:newmatrix}
In \cref{mthm:designs}, our notion of designs requires the matchings $H_u$ to be perfect matchings, which is a stronger definition from the one used in~\cite{Yankovitz24}. One might be worried that our improvement in the lower bound for designs is therefore primarily due to the initial assumption being a bit stronger, rather than for any real technical improvements. The above observation that the second moment bound is tight for perfect matchings not only shows that this is not the case, but also that perfect matchings are \emph{the} case where the second moments are too large.
\end{remark}

It is still potentially possible that the graph $G_u$ admits a large matching, even though the second moments are large. While we have not formally ruled this out, the second moment calculation informally tells us that one probably needs a substantially different approach to prove this.

\parhead{Removing a $\log n$ factor from \cref{fact:2ldclb}?} We have argued that we cannot take $s = O(1)$ so that $O(r s \log^2 n)$. What about the other approach, where we try to shave off the extra $\log n$ coming from \cref{fact:2ldclb} to get a bound of $O(rs)$? This is not something that can be done \emph{generically} for all $2$-LDCs, as of course the Hadamard code is a $2$-LDC with $k = \log_2 n$.

Shaving this $\log n$ factor is closely related to removing the $\log n$ factor from Matrix Khintchine (\cref{fact:matrixkhintchine}), a task studied in many different contexts~\cite{BatsonSS14,MarcusSS15,BansalJM23}.

One such example is the hypergraph Moore bound: the task of showing that a $q$-uniform hypergraph on $n$ vertices with $(n/\ell)^{q/2} \ell$ must have a cycle (also called an even cover) of length $O(\ell \log(n/\ell))$. The best bound for this problem is due to the methods of~\cite{GuruswamiKM22,HsiehKM23}, which uses Kikuchi graphs similar to \cref{def:kikuchiintro} to show the existence of a length $O(\ell \log n)$ cycle when the hypergraph has $(n/\ell)^{q/2} \ell \cdot \log n$ hyperedges, a $\log n$ factor larger than the conjectured threshold. A more complicated argument manages to reduce the $\log n$ factor to $(\log n)^{\frac{1}{q+1}}$ when $q$ is odd~\cite{HsiehKMMS24}, and this proof requires a rather technical modification of the Kikuchi graph.

One might thus naturally surmise that a rather complicated modification of the graphs $G_u$, perhaps along the lines of~\cite{HsiehKMMS24}, is necessary for us to obtain a sharp lower bound via the Kikuchi matrix method. 

\parhead{Our new Kikuchi graph.} To our surprise, and perhaps contrary to the intuition developed above, it turns out that the following \emph{simple} modification of the graphs succeeds: instead of indexing the vertices by collections of sets $(S_1, \dots, S_r)$ or $(T_1, \dots, T_r)$, each of size $s$, we index by ``big'' sets $S$ and $T$ of size $\ell$ (\cref{def:designkikuchi}). This graph is essentially the same as the previous graph if we $\ell = r s$ --- the advantage of this new graph is that if we take $s = O(1)$, then the sets have size $\ell = O(r)$, which is still large, in contrast to the previous graph where the constituent sets $S_h$ would only have size $s = O(1)$. We shall call this new graph the ``uncolored graph'', as we think of the graphs in \cref{def:kikuchiintro} as consisting of a big set $S = S_1 \cup \dots \cup S_r$, where $S_h$ contains vertices from $[n]$ with \emph{color} $h$, which makes the $S_h$'s disjoint by fiat as they use different colors.

In some sense, the above graph is \emph{more} natural than even the ones appearing in~\cite{KothariM23,Yankovitz24}. Of course, the reason those works use the colored graph rather than the uncolored one is that the colored graph makes some of the combinatorial analysis more tame! Indeed, this is because the $r$ colors specify which vertices should appear in the $r$ ``links'' of the chain, namely, vertices that are in $S_h$ appear in the ``$h$-th link'', and correspond to choices only of $v_{2(h-1) + 1}$. One can note (as observed in \cref{def:designkikuchi}) that the uncolored graph does not have edges for all chains $C \in \cH_u^{(r)}$. Indeed, if the left half $C_L$ or the right half $C_R$ contains duplicate vertices, then $C$ does not contribute \emph{any} edges! And, if they share vertices, then they contribute many more edges than a chain where all vertices in $C_L$,$C_R$ are distinct.
Fortunately, as we observe in \cref{def:designchains}, there are at least $(1 - o(1))\abs{\cH_u^{(r)}}$ chains such that the vertices in $C_L$ and $C_R$ are distinct, so we can ignore these issues by working with this large subset of chains.

In spite of these technical challenges, we can make the analysis work for the uncolored graph. This allows us to take $\ell = 2r$ for $r = \frac{1}{2} \log_2 n + O(\log \log n)$, and gives us a final bound of $k \leq (1+o(1))2r \log_2 n$, i.e., $((1-o(1))k)^2 \leq \log_2 n$, which gives us \cref{mthm:designs}. We note that to get the sharp constant in \cref{mthm:designs}, we have to show that $G_u$ admits a \emph{near-perfect} matching.

The calculations involved here are sensitive and require sharp bounds on binomial coefficients. To give the reader a sense of how precise these bounds are, we observe that our proof shows that
\begin{flalign*}
&\E_{(S,w)}[\deg_{u,L}(S,w)^2] = (1 \pm o(1)) d_{u,L}^2\cdot  (3 \delta)^{-1} \sum_{t = 0}^{r-1}   (3 \delta)^{-t} \frac{{r \choose t} {\ell - r \choose r - t}}{{\ell \choose r}} \mcom \\
&\E_{T}[\deg_{u,R}(T)^2] = (1 \pm o(1))  d_R^2  \sum_{t = 0}^r (3 \delta)^{-t} \frac{{r \choose t}{\ell - r \choose r - t}}{{\ell \choose r}} \mper
\end{flalign*}
How do we bound the sum $\sum_{t = 0}^r (3 \delta)^{-t} \frac{{r \choose t}{\ell - r \choose r - t}}{{\ell \choose r}}$? First of all, $ \frac{{r \choose t}{\ell - r \choose r - t}}{{\ell \choose r}}$ is the probability mass function of a hypergeometric distribution, and so $\sum_{t = 0}^r \frac{{r \choose t}{\ell - r \choose r - t}}{{\ell \choose r}} = 1$. Note that the mean of this distribution is $r^2/\ell$ and it has good concentration, so we should expect to lose a factor of $(3 \delta)^{-t^*}$ where $t^* = r^2/\ell$.

In particular, if $3 \delta$ is bounded away from $1$, say, $1/2$, then we need to take $\ell = O(r^2)$ to mitigate this factor, and we so get \emph{no improvement}. But, if we have a \emph{design}, then $3 \delta = 1 - \frac{1}{n}$, and so $(3 \delta)^{-r}$ is only a $1 + o(1)$ factor and thus does not matter!

Finally, we note that since $\sum_{t = 0}^r \frac{{r \choose t}{\ell - r \choose r - t}}{{\ell \choose r}} = 1$, we do need to be precise in our estimates above. In particular, standard estimates on binomial coefficients such as ${\ell \choose r} \geq \left(\frac{\ell}{r} \right)^r$ are insufficient.

We give the full proof of \cref{mthm:designs} in \cref{sec:designs}.

\subsection{Exponential lower bounds for nonlinear smooth $3$-LCCs}
\label{sec:techsnonlin}
\peter{here}
We now explain the key ideas in the proof of \cref{mthm:nonlin}. In this section, we will let $\Code \colon \Fits^k \to \Fits^n$ be a nonlinear code, namely we will use $\Fits$-notation rather than $\Bits$-notation, as it is more convenient for the proof.

\parhead{Existing reductions do not work with the long chain derivation method of~\cite{KothariM23}.}
The standard starting point for lower bounds for nonlinear $q$-LDCs or $q$-LCCs is a reduction from the original work of Katz and Trevisan~\cite{KatzT00}. For $3$-LCCs, this reduction takes any $\delta$-smooth code $\Code$ with completeness even as low as $\frac{1}{2} + \eta$ and outputs $3$-uniform hypergraph matchings $H_u$ for $u \in [n]$ with the following property: for every $u \in [n]$ and hyperedge $C \in H_u$, $\E_{x}[x_u \prod_{v \in C} x_v] \geq \Omega(\eta)$, where the expectation is over a uniformly random codeword $x \in \Code$. That is to say, \emph{every} hyperedge decodes correctly with some constant advantage \emph{in expectation} over a random codeword.

We can now form chains (\cref{def:chainintro}) on the hypergraphs $H_u$. This gives us a ``chain XOR instance'' $\Psi_u(x)$ defined as $\sum_{C \in \cH_u^{(r)}} x_u x_w \prod_{v \in C_L} x_v \prod_{v \in C_R} x_v$, where $C_L$ and $C_R$ are the left and right halves of the chain $C$ and $w$ is the tail of $C$. Unlike in the linear case, it is not guaranteed that these equations are all satisifiable, and so the ``reduction-based approach of~\cite{KothariM23,Yankovitz24} to $2$-LDCs breaks down.\footnote{We note here that one of $\log n$ factors saved by~\cite{Yankovitz24} is in optimizing the ``hypergraph decomposition'' step in~\cite{KothariM23}, but this optimization is specific to linear codes. In this paper, we give a slightly tighter analysis of the decomposition in~\cite{KothariM23} that also saves this $\log n$ factor, and additionally generalizes to the nonlinear setting.} However, the ``spectral refutation approach'' of~\cite{KothariM23} is sufficiently general and resilient enough that, if we could show that $\E_{x \in \Code}[\Psi_u(x)] \geq \eta \abs{\cH_u^{(r)}}$ for $\eta$ even as small as $\eta \gg \frac{1}{n^{1/3}}$, then we could at the very least beat the cubic lower bound of~\cite{AlrabiahGKM23}. And, if we could take $\eta$ to be constant, we would get an exponential lower bound.

Unfortunately, we cannot show any lower bound on $\E_{x}[\Psi_u(x)]$. Indeed, let us even consider the simple case of length $2$-chains, and let us try to bound, for $C = \{v_1, v_2, u_1\}\in H_u$ and $C' = \{v_3, v_4, u_2\} \in H_{u_1}$, the quantity $\E_{x \in \Code}[x_{u} x_{v_1} x_{v_2} x_{v_3} x_{v_4} x_{u_2}]$. This is clearly $\E_{x \in \Code}[(x_{u} x_{v_1} x_{v_2} x_{u_1}) (x_{u_1} x_{v_3} x_{v_4} x_{u_2})]$, and while we know that the expectation of each term is $\geq \eta$, this does not imply anything on the expectation of the product.

This now suggests the following simple way to recover a lower bound: simply assume that the completeness is $1 - \eps$, with the intuition being that $\eta$ is related to $1 - \eps$, and if $\eta$ is close to $1$ then we can apply a union bound. However, the reduction of~\cite{KatzT00} is lossy with respect to the completeness parameter. Indeed, this is because the reduction first makes the decoder nonadaptive by simulating the adaptive decoder by giving it random answers, and this takes $1 - \eps$ completeness to $\frac{1}{2} + \eta$ for $\eta = \frac{1}{8}(1 - \eps)$, which is too small for the union bound strategy to succeed.

Below, we will now explain our proof strategy. We will first explain a simpler strategy that only achieves a superpolynomial lower bound of $n \geq k^{\Omega(\log k)}$, even in the case of perfect completeness. Then, we will explain how to modify that strategy to obtain an exponential lower bound.

\parhead{Key idea 1: adaptive chains and the adaptive chain decoder.} As explained above, the reduction of~\cite{KatzT00} is lossy, so we need to rethink the whole approach. Our intuition is that our reduction should try to remember as much information about the decoder as possible. In particular, this means that we need to remember information about the (possibly adaptive) decoder, and cannot use standard reductions to convert adaptive decoders into nonadaptive ones. At a high level, our new strategy is to form chains \emph{before} applying the reduction of~\cite{KatzT00}.

The first insight we have is that we can form chains \emph{adaptively} by invoking the adaptive decoder $\Dec^{x}(u)$ in a structured way. Namely, define the ``adaptive chain decoder'' $\Dec_r^{x}(u)$ to be the decoder that works as follows:
\begin{enumerate}[(1)]
\item Simulate $\Dec^{x}(u)$ to generate the first query $v_1$. Then, read $a_1 = x_{v_1}$ and respond with $a_1$ to the simulated $\Dec^{x}(u)$ instance.
\item The simulated $\Dec^{x}(u)$ generates a second query $v_2$. Then, read $a_2 = x_{v_2}$ and respond with $a_2$ to the simulated $\Dec^{x}(u)$ instance.
\item The simulated $\Dec^{x}(u)$ generates a third query $u_1$, but $\Dec_r{x}(u)$ \emph{does not} make this query. Instead, $\Dec_r{x}(u)$ invokes $\Dec^{x}(u_1)$, and then proceeds starting from Step (1).
\end{enumerate}
After $r$ iterations of the loop, $\Dec_r^{x}(u)$ makes the final query $u_r$ to receive $a_r$, and then ``feeds answers backwards''. Namely, the simulated $\Dec^{x}(u_{r-1})$ now outputs some guess $a_{r-1}$ for $x_{u_{r-1}}$, which $\Dec_r^{x}(u)$ uses to answer the query $u_{r-1}$ made by the simulated $\Dec^{x}(u_{r-2})$, etc. Finally, $\Dec_r^{x}(u)$ outputs the same answer as the first simulated $\Dec^{x}(u)$.

We can think of the decoder $\Dec_r^{x}(u)$ as generating \emph{adaptive chains}, which are sequences of the form $(u_0, (v_1, a_1), (v_2, a_2), u_1, (v_3, a_3), (v_4, a_4), \dots)$.

\parhead{Key idea 2: representing the decoder as a polynomial.} Let us assume that the decoder succeeds with probability $1$ for simplicity. Now that we have adaptive chains, we need to define a polynomial $\Psi_u(x)$ using the adaptive chains $C$ in $\cH_u^{(r)}$ (now redefined to be the set of adaptive chains) so that $\Psi_u(x) x_u = 1$ for all $x \in \Code$. Our key idea is to encode the behavior of $\Dec^{x}(u)$ as a certain polynomial. Then, forming chains corresponds to taking certain ``chain products'' of these polynomials, which defines a ``chain polynomial'' $\Psi_u(x)$ that we will refute.

We can represent $\Dec^{x}(u)$ as a decision tree, and for simplicity, let us assume that $\Dec^{x}(u)$ makes exactly $3$ queries and has perfect completeness. First, it has some distribution that it uses to generate the first query $v_1$. Then, it has a branch for each answer $a_1 \in \Fits$ that it could receive. After the branch, it has another distribution to generate the second query $v_2$, and then another branch for each answer $a_2 \in \Fits$. It then has a final distribution for the query $v_3$, and then it has a ``decoding function'' $f_{v_1, a_1, v_2, a_2, v_3}(a_3)$. Notice that we are allowed to have a different decoding function for each choice of ``adaptive constraint'' $C = (v_1, a_1, v_2, a_2, v_3)$, and so $f_{v_1, a_1, v_2, a_2, v_3}(a_3)$ need only depend on $a_3$. Because $\Dec^{x}(u)$ decodes with probability $1$, we must have that for any $x \in \Code$ with $x_{v_1} = a_1, x_{v_2} = a_2$, it holds that $f_{v_1, a_1, v_2, a_2, v_3}(x_{v_3}) = x_u$. Additionally, this implies that $f_{v_1, a_1, v_2, a_2, v_3}$ is deterministic, and so it is one of the following $4$ functions: $1, -1, a_3, -a_3$. For simplicity, let us pretend that all such decoding functions are simply $a_3$.

The above analysis effectively gives us constraints of the form
\begin{equation*}
``x_{v_1} = a_1 \wedge x_{v_2} = a_2 \implies x_u = x_{v_3} \mper''
\end{equation*}
We can represent these constraints as polynomials by the $\AND$ polynomial (\cref{def:andpoly}): we then have polynomial constraints $\AND(a_1 x_{v_1}, a_2 x_{v_2}) x_{v_3} x_u = \AND(a_1 x_{v_1}, a_2 x_{v_2})$.

Unlike in the linear case, the constraints come with weights, corresponding to the probability that the decoder $\Dec^{x}(u)$ makes the queries $C = (v_1, a_1, v_2, a_2, v_3)$ where $x$ satisfies $x_{v_1} = a_1, x_{v_2} = a_2$. Let $\wt_u(C)$ be the weight of such a constraint. We then have for any $x \in \Code$:
\begin{flalign*}
&\sum_{C = (v_1, a_1, v_2, a_2, v_3)} \wt_u(C) \AND(a_1 x_{v_1}, a_2 x_{v_2}) = 1\\
&\sum_{C = (v_1, a_1, v_2, a_2, v_3)} \wt_u(C) = 4 \mper
\end{flalign*}
The first equation sums the probabilities of querying certain $C$'s, which must sum to $1$ because this is the query distribution of $\Dec^{x}(u)$. The second equation observes that $\wt(C)/4$ is the probability that $\Dec^{y}(u)$ queries $C$ when $y$ is chosen uniformly at random.

Finally, we have
\begin{flalign*}
\sum_{C = (v_1, a_1, v_2, a_2, v_3)} \wt_u(C) \AND(a_1 x_{v_1}, a_2 x_{v_2}) x_{v_3} x_u = \sum_{C = (v_1, a_1, v_2, a_2, v_3)} \wt_u(C) \AND(a_1 x_{v_1}, a_2 x_{v_2}) = 1 \mcom
\end{flalign*}
using the fact that the polynomial constraints are satisfied and the first of the $2$ previous equations. We also note that, more generally, the left-hand side exactly computes $\E[\Dec_r^{x}(u) x_u]$ for a fixed $x$, where the expectation is over the internal randomness of the adaptive decoder.

\parhead{Forming and refuting chain polynomials.} 
The next step of the proof is to use the polynomial representation of $\Dec^{x}(u)$ above to represent the behavior of the chain decoder $\Dec_r^{x}(u)$ as a polynomial as well. Concretely, the polynomial for $\Dec_2^{x}(u)$ is
\begin{flalign*}
&x_u \sum_{C = (v_1, a_1, v_2, a_2, u_1)} \left(\wt_u(C)\AND(a_1 x_{v_1}, a_2 x_{v_2}) \left(\sum_{C' = (v_3, a_3, v_4, a_4, u_2)} \wt_{u_1}(C')\AND(a_3 x_{v_3}, a_4 x_{v_4}) x_{u_2} \right)\right) \mcom
\end{flalign*}
and we let $\Psi_u(x)$ denote the polynomial for the length $r$-chains, defined analogously. Because of perfect completeness, $\Psi_u(x) = 1$ for all $x \in \Code$. Notice that $\Psi_u(x)$ exactly computes $\E[\Dec_r^{x}(u) x_u]$ where $\Dec_r^{x}(u)$ is the adaptive chain decoder and the expectation is over the internal randomness of $\Dec_r^{x}(u)$.

We then follow the strategy of~\cite{KothariM23} and let $\Psi_b(x) = \sum_{i = 1}^k b_i \Psi_i(x)$. Standard reductions (\cref{fact:bgt17}) allow us to assume that the code $\Code$ is \emph{systematic} with only a small loss in parameters, so that for a random $x \in \Code$, the bits $x_1, \dots, x_k \in \Fits$ are independent. To prove a lower bound, it suffices to then argue that $\E_{b \in \Fits^k}[\max_{y \in \Fits^n} \Psi_b(y)] < k$, as $\Psi_b(\Code(b)) = k$.

The spectral refutation of~\cite{KothariM23}, built on the CSP refutation algorithm of~\cite{GuruswamiKM22}, does not directly bound $\E_{b \in \Fits^k}[\max_{y \in \Fits^n} \Psi(y)]$, because the polynomials constructed here are quite a bit more general than the case handled in~\cite{KothariM23}. However, because the spectral methods of~\cite{GuruswamiKM22,KothariM23} are sufficiently resilient, one can succeed in using the techniques to bound $\E_{b \in \Fits^k}[\max_{y \in \Fits^n} \Psi(y)]$ (\cref{lem:chainxorref}). This requires a generalization of the proof technique in~\cite{KothariM23}; we will comment on the new ideas we require later.

We obtain a bound of $\E_{b \in \Fits^k}[\max_{y \in \Fits^n} \Psi(y)] \leq W \cdot O(\sqrt{k \ell r \log n})$ where $\ell, r$ are parameters with $\ell^r \gg n$, and $W$ is the total ``weight'' of the $\Psi_u$, i.e., the sum of the weights of the coefficients. The quantity $W$ should be considered some normalized analog of the ``number of XOR constraints'' present in the polynomial.

Because the sum of the weights in a single $\Dec(u)$ is $4$, the total weight in $\Psi_u$ is at most $4^r$. This gives us a bound of $4^r \cdot O(\sqrt{k \ell r \log n})$; we then set $r = \sqrt{\log n}$, and $\ell = 2^{O(\sqrt{\log n})}$. Rearranging then implies that $k \leq 2^{O(\sqrt{\log n})}$, i.e., $n \geq k^{\Omega(\log k)}$.

The reason for the factor of $4$ loss (which causes the $4^r$ loss and prevents us from taking $r = O(\log n)$ to get an exponential lower bound) can be observed by looking at how this reduction behaves when the code $\Code$ is actually a linear $3$-LCC. For each linear constraint $x_{v_1} x_{v_2} x_{v_3} = x_u$, we produce $4$ constraints
\begin{equation*}
``x_{v_1} = a_1 \wedge x_{v_2} = a_2 \implies x_u = a_1 a_2 x_{v_3} \mcom''
\end{equation*}
one for every $a_1, a_2 \in \Fits^k$. So, we have produced a factor of $4$ more equations than was needed. Indeed, the reason that this loss does not appear for linear codes is because of the equality
\begin{equation*}
\sum_{a_1, a_2} \AND(a_1 x_{v_1}, a_2 x_{v_2}) a_1 a_2 x_{v_3} = x_{v_1} x_{v_2} x_{v_3}\mper
\end{equation*}
However, in the adaptive case, the query $v_3$ can depend on the answers, and we also need not have $f_{v_1, a_1, v_2, a_2}(a_3) = a_1 a_2 a_3$, so this does not necessarily hold.

\peter{here}
\parhead{Obtaining an exponential lower bound.} We will now explain how to modify the above strategy to obtain an exponential lower bound. Recall that the issue above is that, even if $\Code$ is a linear code, the total weight of terms in our polynomial for $\Dec(u)$ is $4$, rather than $1$. The key observation we will now make is that the total weight of the \emph{homogeneous degree $3$ terms} in this polynomial is at most $1$. This is because the degree $3$ term in the Fourier expansion of $\AND(a_1 x_{v_1}, a_2 x_{v_2}) x_{v_3}$ is $\frac{1}{4} x_{v_1} x_{v_2} x_{v_3}$, which decreases the weight by a factor of $4$.

We now form chains using the polynomials directly, but we only replace the final query $x_{v_3}$ when it appears in a homogeneous degree $3$ term in the Fourier expansion of the polynomial for $\Dec^x(u)$. Because the weight of the homogeneous degree $3$ terms is always at most $1$, the total weight in the new polynomial $\Psi_b$ is at most $1^r = 1$. However, we now have additional polynomials $\Phi_b^{(t)}$ for $1 \leq t \leq r$, which denote the ``residual'' terms left over from the degree $\leq 2$ components of the polynomials. Each of these has total weight at most $4$, because $\Phi_b^{(t)}$ is formed by multiplying $t-1$ degree $3$ terms with a final degree $2$ term, and thus the total weight in $\Phi_b^{(t)}$ is at most $1^{t-1} \cdot 4$. This implies that the total weight across all $\Phi_b^{(t)}$ and $\Psi_b$ is at most $4(r+1)$, which is now \emph{linear} in $r$.

\parhead{Additional technical complications.} Let us now discuss the generalizations that we need of the theorems in~\cite{KothariM23} in order to prove \cref{lem:chainxorref}. First of all, we need to now handle weighted hypergraphs that are not necessarily matchings, rather than just matchings where all hyperedges have equal weight. The definitions and methods in~\cite{KothariM23} are sufficiently general that this can be done with some effort. A key part of the generalization is relying on the smoothness of the decoder, which conceptually generalizes the concept of hypergraph matchings to weighted hypergraphs. A second issue encountered is that $\AND(a_1 x_{v_1}, a_2 x_{v_2})$ 
is not a homogeneous degree $2$ polynomial and the decoding function $f_{v_1, a_1, v_2, a_2}(a_3)$ might have a negative sign and be $-a_3$. This means that $\Psi_b$ is not a homogeneous degree $2r$ polynomial with nonnegative coefficients. Fortunately, this issue can again be circumvented by adding 
extra dummy variables $y_{-v}, y_{1^{(v)}}$, and $y_{-1^{(v)}}$ for each $v \in [n]$, where we expect these variables to take the values
$y_{-v} = -x_v, y_{1^{(v)}} = 1$, and $y_{-1^{(v)}} - 1$.

A third and perhaps more pressing issue encountered is that we now have additional polynomials $\Phi_b^{(t)}(y)$. In fact, such polynomials appear even if we only consider the adaptive chains of the chain decoder $\Dec_r^{x}(u)$, as the decoder might terminate without making a third query. For example, suppose the initial invocation of $\Dec^{x}(u)$ by $\Dec_r^{x}(u)$ leads to the first $2$ queries and answers being $C = (v_1, a_1, v_2, a_2)$. Then, $\Dec^{x}(u)$ generates a third query $u_1$. However, the decoding function $f_{v_1, a_1, v_2, a_2, u_1}(a_3)$ might not depend on $a_3 = x_{u_1}$, i.e., it could be the constant function $1$ or the constant function $-1$. In this case, we do not have a way to continue the chain as normal, and $\Dec_r^{x}(u)$ does not make any more queries.

It turns out that the $\Phi_b^{(t)}$'s, i.e., the polynomials for the chains that stop early, are (modulo some tricks) homogeneous degree $2t$ polynomials, which is \emph{even}. This should be compared to the polynomial $\Psi_b$, which is (modulo some tricks) a homogeneous degree $2r+1$ polynomial, which is \emph{odd}. It turns out that the methods of~\cite{KothariM23}, and indeed the more general case of refuting XOR instances~\cite{GuruswamiKM22}, are easier to analyze for even degree instances. This allows us to refute these ``early stop'' instances, and thereby show our key refutation lemma \cref{lem:chainxorref}.

\parhead{Handling imperfect completeness.}
\peter{here}
Finally, we explain how to handle the case when the decoder only succeeds with probability $1 - \eps$. The polynomials that we constructed from the chain decoder $\Dec_r(u)$ were convenient because they had a simple interpretation using the chain decoder $\Dec_r(u)$. This means that when $\eps > 0$, it is simple to lower bound $\E_{b \in \Fits^k}[\max_{y \in \Fits^n} \Psi(y)]$, because one can invoke the fact that $\E[\Dec_r^{x}(u) x_u] \geq 1 - 2 \eps$ and then apply a union bound to conclude a lower bound of $1 - 2 \eps r$.

However, these new polynomials do not appear to have any nice interpretation in terms of the original decoder $\Dec(u)$, and so it might be possible that we are unable to lower bound  $\E_{b \in \Fits^k}[\max_{y \in \Fits^n} (\Phi_b^{(1)}(y) + \dots +\Phi_b^{(r)}(y) + \Psi(y)]$. Fortunately, it turns out that, with a bit of a trickier proof, one can argue the same lower bound of $1 - 2 \eps r$. 

Either way, we can only take $r \approx 1/2\eps$ length chains while still showing a lower bound on $\E_{b \in \Fits^k}[\max_{y \in \Fits^n} \Psi_b(y)]$. In this parameter regime, when $\eps$ is a small constant, say, our final bound will be about $k \leq n^{1/r} \approx n^{2 \eps}$.

One may wonder why we cannot obtain a better lower bound of, say, $(1 - \eps)^r$, on either $\E[x_u \Dec_r^{x}(u)]$ or our final chain polynomial. Indeed, if $\E[x_u \Dec^{x}(u)]$ only depends on $x$, or $u$, but not both, then we could make the analysis work. The problem is that this is not something that we can enforce without loss of generality, as $\E[x_u \Dec^{x}(u)]$ could depend on both $x$ and $u$.

The reason this causes an issue can be seen from the $2$-chain decoder. Consider the $2$-chain polynomial from before:
\begin{flalign*}
&x_u \sum_{C = (v_1, a_1, v_2, a_2, u_1)} \left(\wt_u(C)\AND(a_1 x_{v_1}, a_2 x_{v_2}) \left(\sum_{C' = (v_3, a_3, v_4, a_4, u_2)} \wt_{u_1}(C')\AND(a_3 x_{v_3}, a_3 x_{v_3}) x_{u_2} \right)\right) \mper
\end{flalign*}
We have 
\begin{flalign*}
x_{u_1} \sum_{C' = (v_3, a_3, v_4, a_4, u_2)} \wt_{u_1}(C')\AND(a_3 x_{v_3}, a_3 x_{v_3}) x_{u_2} = \E[x_{u_1} \Dec^{x}(u_1)] = p_{x, u_1} \geq 1 - \eps \mcom
\end{flalign*}
so that
\begin{flalign*}
&x_u \sum_{C = (v_1, a_1, v_2, a_2, u_1)} \left(\wt_u(C)\AND(a_1 x_{v_1}, a_2 x_{v_2}) \left(\sum_{C' = (v_3, a_3, v_4, a_4, u_2)} \wt_{u_1}(C')\AND(a_3 x_{v_3}, a_3 x_{v_3}) x_{u_2} \right)\right) \\
&= x_u \sum_{C = (v_1, a_1, v_2, a_2, u_1)} \left(\wt_u(C)\AND(a_1 x_{v_1}, a_2 x_{v_2}) x_{u_1} p_{x, u_1} \right) \mper
\end{flalign*}
Now, 
\begin{flalign*}
& x_u \sum_{C = (v_1, a_1, v_2, a_2, u_1)} \left(\wt_u(C)\AND(a_1 x_{v_1}, a_2 x_{v_2}) x_{u_1} \right) = p_{x,u} \geq 1 - 2\eps \mcom
\end{flalign*}
for all $x \in \Code$, and so we would like conclude that 
\begin{flalign*}
&= x_u \sum_{C = (v_1, a_1, v_2, a_2, u_1)} \left(\wt_u(C)\AND(a_1 x_{v_1}, a_2 x_{v_2}) x_{u_1} p_{x, u_1} \right) \geq (1 - 2\eps)^2 \mper
\end{flalign*}
Unfortunately, the reweightings caused by the different $p_{x,u_1}$'s cannot be ignored because the terms in the sum are not nonnegative. In particular, it could be that $p_{x,u_1} = 1$ for the negative terms and $p_{x, u_1} = 1 - 2\eps$ for the positive terms, which would then (if $p_{x,u} = 1 - 2\eps$) make the sum $(1 - \eps)(1 - 2\eps) - \eps = (1 - 2\eps)^2 - 2\eps^2$, for example.
\subsection{Roadmap}
The rest of the paper is organized as follows. First, in \cref{sec:prelims}, we introduce some notation and recall basic facts about LCCs that we shall use in the proof. In \cref{sec:designs}, we prove \cref{mthm:designs}. In \cref{sec:chainpolys,sec:graphref,sec:decomp,sec:kikuchi,sec:rowpruning}, we prove \cref{mthm:nonlin}. This proof is broken into two stages: the reduction from the adaptive smooth decoder to the chain XOR instances is done in \cref{sec:chainpolys}, and in \cref{sec:graphref,sec:decomp,sec:kikuchi,sec:rowpruning} we refute these instances by proving \cref{lem:chainxorref}.

Finally, in \cref{app:reedmuller} we recall the folklore construction of design $3$-LCCs from Reed--Muller codes.

\section{Preliminaries}
\label{sec:prelims}
\subsection{Basic notation}
We let $[n]$ denote the set $\{1, \dots, n\}$. For two subsets $S, T \subseteq [n]$, we let $S \oplus T$ denote the symmetric difference of $S$ and $T$, i.e., $S \oplus T \coloneqq \{i : (i \in S \wedge i \notin T) \vee (i \notin S \wedge i \in T)\}$. For a natural number $t \in \N$, we let ${[n] \choose t}$ be the collection of subsets of $[n]$ of size exactly $t$. Given variables $x_1, \dots, x_n$ and a subset $C \subseteq [n]$, we let $x_S \coloneqq \prod_{v \in S} x_v$.

For a rectangular matrix $A \in \R^{m \times n}$, we let $\norm{A}_2 = \coloneqq \max_{x \in \R^m, y \in \R^n: \norm{x}_2 = \norm{y}_2 = 1} x^{\top} A y$ denote the spectral norm of $A$, and $\boolnorm{A} \coloneqq \max_{x \in \Fits^m, y \in \Fits^n} x^{\top} A y$. We note that $\boolnorm{A} \leq \sqrt{nm} \norm{A}_2$.

\subsection{XOR formulas}
A (weighted) XOR instance $\psi$ on $n$ variables $x_1, x_2,\ldots, x_n$ taking values in $\Fits$ is a collection of constraints of the form $\{x_C = b_C\}$ where $C \in \cH$ where $\cH \subseteq 2^{[n]}$ is the \emph{constraint hypergraph}, along with weights $\wt(C) \geq 0$ for each $C \in \cH$. The \emph{arity} of a constraint $\{x_C = b_C\}$ equals $|C|$. The arity of $\psi$ is the maximum arity of any constraint in it. The XOR formula associated with $\psi$ is the expression $\psi(x) = \sum_{C \in \cH} \wt(C) b_C x_C$ seen as a polynomial over $\Fits^n$. Notice that $\psi(x) = \sum_{C \in \cH} \wt(C)$ if $x$ satisfies all the constraints of $\psi$ and in general, evaluates to (weight of constraints satisfied by $x$) - (weight of constraints violated by $x$). The \emph{value} $\val(\psi)$ of a XOR instance $\psi$ (or, of the associated formula $\psi(x)$) is the maximum of $\psi(x)$ as $x$ ranges over $\Fits^n$. More generally, for a function $f(x)$, we shall let $\val(f) \coloneqq \max_{x \in \Fits^n} f(x)$.

\subsection{Locally correctable codes}
We refer the reader to the survey~\cite{Yekhanin12} for background. 

\begin{definition}[Locally correctable code]
\label{def:LCC}
A map $\Code \colon \Fits^k \to \Fits^n$ is a $(q, \delta, \eps)$-locally correctable code if there exists a randomized decoding algorithm $\Dec(\cdot)$ that takes input an oracle access to some $y \in \Fits^n$ and a $u \in [n]$, and has the following properties:
\begin{enumerate}[(1)]
\item ($q$ queries) For any $y \in \Fits^n$ and $u \in [n]$, $\Dec^{y}(u)$ makes at most $q$ queries to the string $y$;
\item ($(1-\eps)$-correction with $\delta n$ errors) For all $b \in \Fits^k$, $u \in [n]$, and all $y \in \Fits^n$ such that $\Delta(y, \Code(b)) \leq \delta n$, $\Pr[\Dec^{y}(u) = \Code(b)_u] \geq 1 - \eps$. Here, $\Delta(x,y)$ denotes the Hamming distance between $x$ and $y$, i.e., the number of indices $v \in [n]$ where $x_v \ne y_v$.
\end{enumerate}
We say that $\Code$ is \emph{linear} if the map $\Code$, when viewed as a map from $\Bits^k \to \Bits^n$ via the mapping $0 \leftrightarrow 1$ and $1 \leftrightarrow -1$, is a linear map. We note that for linear codes, $k = \dim(\cV)$, where $\cV$ is the image of $\Bits^k$ under the map $\Code$. We will typically let $\Lincode$, as opposed to $\Code$, denote a linear code, and view $\Lincode$ as a map $\Lincode \colon \Bits^k \to \Bits^n$.

We say that $\Code$ is systematic if for every $b \in \Fits^k$, $\Code(b) \vert_{[k]} = b$.
\end{definition}
For a code $\Code \colon \Fits^k \to \Fits^n$, we will write $x \in \Code$ to denote an $x = \Code(b)$ for some $b \in \Fits^k$.

\begin{definition}[Smooth LCCs~\cite{KatzT00}]
\label{def:smoothLCC}
A map $\Code \colon \Fits^k \to \Fits^n$ is a $\delta$-smooth $q$-locally correctable code with completeness $1 - \eps$ if there exists a randomized decoding algorithm $\Dec(\cdot)$ that takes input an oracle access to some $y \in \Fits^n$ and a $u \in [n]$, and has the following properties:
\begin{enumerate}[(1)]
\item ($q$ queries) For any $y \in \Fits^n$ and $u \in [n]$, $\Dec^{y}(u)$ makes at most $q$ queries to the string $y$;
\item ($(1-\eps)$-completeness) For all $b \in \Fits^k$, $u \in [n]$, $\Pr[\Dec^{\Code(b)}(u) = \Code(b)_u] \geq 1 - \eps$.
\item ($\delta$-smoothness) For all $b \in \Fits^k$, $u \in [n]$, $x = \Code(b)$, $v \in [n]$, $\Pr[\Dec^{\Code(b)}(u) \ \text{ queries $v$}] \leq \frac{1}{\delta n}$.
\end{enumerate}
We will call such codes $(q, \delta,\eps)$-smooth LCCs.
\end{definition}

\begin{remark}
\label{rem:smoothLCC}
Any $\delta$-smooth $q$-LCC with completeness $1 - \eps$ is a $(q, \eta\delta, \eps + \eta)$-LCC for any $\eta > 0$. Indeed, this follows because if we let $y \in \Fits^n$ be a corruption of a codeword $x \in \Code$ with $\eta \delta n$ errors, then the probability that the smooth decoder queries a corrupted entry is at most $\eta$.
\end{remark}

\begin{fact}[Systematic Nonlinear Codes, Lemma A.5, Thm A.6 in~\cite{BhattacharyyaGT17}]
\label{fact:bgt17}
Let $\Code \colon \Fits^k \to \Fits^n$ be a $\delta$-smooth $q$-LCC with completeness $1 - \eps$. Then, there is a \emph{systematic} code $\Code' \colon \Fits^{k'} \to \Fits^n$ that is a $\delta$-smooth $q$-LCC with completeness $1 - \eps$, where $k' = \Omega(k/\log(1/\delta))$.
\end{fact}

We next discuss a combinatorial characterization of \emph{linear} locally correctable codes. To begin with, we recall basic notions about hypergraphs.  

\begin{definition}
\label{def:hypergraph}
A weighted (and undirected) hypergraph $\cH$ on vertex set $[n]$ is a weight function $\wt_{\cH} \colon 2^{[n]} \to \R_{\geq 0}$, i.e., a function from unordered sets $C \subseteq [n]$ to $\R_{\geq 0}$. The hypergraph is $\leq q$-uniform if $\abs{C} > q$ implies that $\wt_{\cH}(C) = 0$ and $q$-uniform if $\abs{C} \ne q$ implies that $\wt_{\cH}(C) = 0$.

A weighted directed hypergraph $\cH$ on vertex set $[n]$ is a weight function $\wt_{\cH} \colon S \to \R_{\geq 0}$, where $S$ denotes the set of all \emph{ordered} subsets of $[n]$. The hypergraph is $\leq q$-uniform if for any ordered set $C \subseteq [n]$, $\abs{C} > q$ implies that $\wt_{\cH}(C) = 0$ and $q$-uniform if $\abs{C} \ne q$ implies that $\wt_{\cH}(C) = 0$.

For a subset $Q \subseteq [n]$, we define the degree of $Q$ in $\cH$, denoted $\deg_{\cH}(Q)$, to be $\sum_{C \in [n]^q : Q \subseteq C} \wt_{\cH}(C)$, where we say that $Q \subseteq C$ if this containment holds as sets.
\end{definition}

LCCs admit a standard combinatorial characterization (formalized in the definition below).

\begin{definition}[Linear LCC in normal form]
\label{def:normalLCC}
A linear code $\Lincode \colon \Bits^k \to \Bits^n$ is $(q, \delta)$-normally correctable if for each $u \in [n]$, there is a $q$-uniform hypergraph matching $\cH_u$ with at least $\delta n$ hyperedges such that for every $C \in \cH_u$ and $b \in \Fits^k$, it holds that $\prod_{v \in C} x_v = x_u$ where $x = \Code(b)$.
\end{definition}

Every linear LCC can be transformed into a linear LCC in normal form with only a small loss in parameters.
\begin{fact}[Reduction to LCC normal form, Theorem 8.1 in~\cite{Dvir16}]\label{fact:normalform}
Let $\Lincode \colon \Bits^k \to \Bits^n$ be a linear code that is $(q, \delta, \eps)$-locally correctable. Then, there is a linear code $\Lincode' \colon \Bits^k \to \Bits^{2n}$ that is $(q, \delta')$-normally correctable, with $\delta' \geq \delta/2q$.
\end{fact}

Below, we define \emph{design $3$-LCCs}, which are an idealized form of linear $3$-LCCs in normal form. We note that Reed--Muller codes, the best known construction of $3$-LCCs, are designs (see \cref{app:reedmuller}).
\begin{definition}[Design $3$-LCCs]
\label{def:designs}
Let $H \subseteq {[n] \choose 4}$ denote a collection of subsets of $n$ of size exactly $4$. We say that $H$ is a \emph{design} if, for every pair of vertices $u \ne v \in [n]$, there exists \emph{exactly} one $C \in H$ with $\{u,v\} \subseteq C$.

We say that such an $H$ is a design $3$-LCC of dimension $k$ if the subspace $\cV \defeq \{x \in \Bits^n : \sum_{v \in C} x_v = 0 \ \forall C \in H\} \subseteq \Bits^n$ has dimension $k$. 
\end{definition}

\begin{remark}[Connection between \cref{def:designs} and \cref{def:normalLCC}]
\label{rem:designtonormalform}
Given a design $3$-LCC $H$, we can construct the hypergraphs $H_u$ for $u \in [n]$ in \cref{def:designs} by letting $H_u \defeq \{C \setminus \{u\} : C \in H \text{ and } u \in C\}$ be the set of $C \in H$ that contain $u$ (and then remove $u$). Because $H$ is a design, for every pair $u \ne v \in [n]$, there exists $C \in H$ containing $u$ and $v$. So, there is exactly one $C' \in H_u$ containing $v$, which implies that $H_u$ is a perfect $3$-uniform hypergraph matching on $[n] \setminus \{u\}$, i.e., $\abs{H_u} = \frac{n - 1}{3}$.
\end{remark}

Finally, we recall the lower bound for linear $2$-LDCs from~\cite{GoldreichKST06}.
\begin{fact}[Lemma 3.3, Lemma 3.5 in~\cite{GoldreichKST06}]
\label{fact:2ldclb}
Let $\Lincode \colon \Bits^k \to \Bits^n$ be a linear map, and let $G_1, \dots, G_k$ be matchings on $n$ vertices such that for every $b \in \Bits^k$ and every $i \in [k]$ and every $(u,v) \in G_i$, it holds that $x_u + x_v = b_i$, where $x = \Lincode(b)$. Suppose that $\frac{1}{k} \sum_{i = 1}^k \abs{G_i} \geq \delta n$. Then, $2\delta k \leq \log_2 n$.
\end{fact}

\subsection{Concentration inequalities}
We will use the following non-commutative Khintchine inequality~\cite{LustPiquardG91}.
\begin{fact}[Rectangular Matrix Khintchine inequality, Theorem 4.1.1 of \cite{Tropp15}]
\label{fact:matrixkhintchine}
Let $X_1, \dots, X_k$ be fixed $d_1 \times d_2$ matrices and $b_1, \dots , b_k$ be i.i.d.\ from $\Fits$. Let $\sigma^2 \geq \max(\norm{\sum_{i = 1}^k X_i X_i^{\top}]}_2, \norm{\sum_{i = 1}^k X_i^{\top} X_i]}_2)$. Then
\begin{equation*}
\E\Bigl[\ \Norm{\sum_{i = 1}^k b_i X_i}_2\ \Bigr] \leq \sqrt{2\sigma^2 \log(d_1 + d_2)} \enspace.
\end{equation*}
\end{fact}
 
\subsection{A fact about binomial coefficients}
 \begin{fact}
 \label{fact:binest}
Let $n, r, t, \ell$ be integers with $t \leq r$ and $\ell \geq r$. Then, it holds that
 \begin{equation*}
 \frac{{r \choose t} t! {n \choose \ell} {n \choose {\ell - (2r - t)}} }{{n - 2r \choose \ell - r}  {n - 2r \choose \ell - r}} \leq \left(1 + \frac{O(\ell^2)}{n}\right) n^t \frac{{\ell - r \choose r - t}}{{\ell \choose r}}
  \end{equation*}
 \end{fact}
 \begin{proof}
 First, we have that
 \begin{flalign*}
 &\frac{ {n \choose \ell} {n \choose {\ell - (2r - t)}} }{{n - 2r \choose \ell - r}  {n - 2r \choose \ell - r}} \leq \left(1 + \frac{O(\ell^2)}{n}\right) \frac{n^{\ell}}{\ell!} \cdot \frac{n^{\ell -(2r - t)}}{(\ell - (2r - t))!} \cdot \frac{(\ell - r)!}{n^{\ell - r}}  \frac{(\ell - r)!}{n^{\ell - r}} \\
 &\leq \left(1 + \frac{O(\ell^2)}{n}\right) n^t \frac{(\ell - r)!}{\ell!} \cdot \frac{(\ell - r)!}{(\ell - (2r - t))!} \mper
  \end{flalign*}
  We now observe that
\begin{flalign*}
&{r \choose t} t! \frac{(\ell - r)!}{\ell!} \cdot \frac{(\ell - r)!}{(\ell - (2r - t))!}  =  \frac{r!}{(r-t)!} \cdot \frac{(\ell - r)!}{\ell!} \cdot \frac{(\ell - r)!}{(\ell - (2r - t))!} \\
&=\frac{1}{{\ell \choose r}} \cdot \frac{1}{(r-t)!} \cdot  \frac{(\ell - r)!}{(\ell - (2r - t))!} \\
&= \frac{{\ell - r \choose r - t}}{{\ell \choose r}}  \mcom
\end{flalign*}
 which finishes the proof.
 \end{proof}


\section{Proof of \cref{mthm:designs}}
\label{sec:designs}
In this section, we prove \cref{mthm:designs}. The proof is substantially simpler than the proof for general linear codes (\cite{KothariM23}) or \cref{mthm:nonlin}. The proof here will be self-contained, and will also serve as a partial warmup to \cref{mthm:nonlin}.

The proof presented follows the overall blueprint of the proof in~\cite{KothariM23}. Namely, we will use the design $3$-LCC $\Lincode$ to construct a $2$-query linear locally decodable code, and then we will apply the lower bound of~\cite{GoldreichKST06}.\footnote{The proof in~\cite{KothariM23} is presented using the perspective of spectral refutation and matrix concentration bounds, even though the final proof eventually is a reduction to a $2$-LDC. Here, we present the proof as a reduction as it is a more accessible and combinatorial analysis, although we note that one could prove the same result using matrix concentration as well.} As mentioned in \cref{sec:techsdesigns}, we will incorporate the clever second moment method proof of the row pruning step due to~\cite{Yankovitz24}, which is very similar to the edge deletion method of~\cite{HsiehKM23} done in the context of semirandom and smoothed CSP refutation~\cite{GuruswamiKM22}. The key reason that we save the final $\log n$ factor is by using a more carefully chosen Kikuchi graph, a sharp accounting of binomial coefficients, and the crucial use of the fact that in the design case, the hypergraph matchings are perfect.

Let us now proceed with the proof. Let $\Lincode \colon \Bits^k \to \Bits^n$ be a design $3$-LCC. Namely, there exists a $4$-uniform hypergraph design $H \subseteq {[n] \choose 4}$ such that for all $C \in H$, $\sum_{v \in C} x_v = 0$ for all $x \in \Lincode$. Without loss of generality, we may assume that $\Lincode$ is systematic, i.e., for each $b \in \Bits^k$, $\Lincode(b)_i = b_i$. To bound $k$, we will give another linear map $\Lincode' \colon \Bits^n \to \Bits^{2nN}$, where $N =  {n \choose \ell}$ for some parameter $\ell = (1 + o(1))\log_2 n$, and we will show that $\Lincode' \circ \Lincode \colon \Bits^k \to \Bits^N$ is a $2$-query linear locally decodable code with $\delta = \frac{1}{2}(1 - o(1))$. We can then apply \cref{fact:2ldclb} to conclude that $(1 - o(1)) k \leq 2 \delta k \leq \log_2 N \leq (\ell + 1) \log_2 n$ where $\ell = (1 +o(1))\log_2 n$. 

For each $u \in [n]$, we let $H_u$ denote the $3$-uniform hypergraph defined from $H$ as specified in \cref{rem:designtonormalform}, i.e., $H_u = \{C : C \cup \{u\} \in H\}$. As shown in \cref{rem:designtonormalform}, $H_u$ is a matching of size $\delta n = \frac{n-1}{3}$, i.e., $\delta \defeq \frac{1}{3} - \frac{1}{3n}$.

\parhead{Step 1: forming long chain derivations.} In the first step of the proof, we use the initial system of constraints $H$ to define a larger system of constraints, called long chain derivations. This is the key idea of \cite{KothariM23} that yields the first exponential lower bound for linear $3$-LCCs, and is the starting point of our proof.

\begin{definition}
\label{def:designchains}
Let $H_1, \dots, H_n$ be the $3$-uniform hypergraph matchings defined from the $4$-design $H$. An $r$-chain with \emph{head $u_0$} is an ordered sequence of vertices of length $3r + 1$, given by $C = (u_0, v_1, v_2, u_1, v_3, v_4, u_2, \dots, v_{2(r-1) + 1}, v_{2(r-1) + 2}, u_r)$, such that all the $v_h$'s are \emph{distinct}\footnote{In this section only, we will enforce that all the $v_h$'s are distinct, as this will be slightly more convenient.} and for each $h = 0, \dots, r-1$, it holds that $\{v_{2h + 1}, v_{2h + 2}, u_{h+1}\} \in H_{u_h}$. We let $\cH^{(r)}_u$ denote the set of $r$-chains with head $u$.

We let $C_L = (v_1, v_3, v_5, \dots, v_{2(r-1) + 1})$ denote the ``left half'' of the chain, and $C_R = (v_2, v_4, v_6, \dots, v_{2(r-1) + 2})$ denote the ``right half''. We call $u_r$ the ``tail''.
\end{definition}

We observe that $\cH^{(r)}_u$ has size at most $(6 \delta n)^r$ and size at least $(6 \delta n - 4r)^r$. Indeed, the upper bound follows because, given a partial chain $(u_0, v_1, v_2, \dots, u_h)$, there are exactly $6 \delta n$ choices of $(v_{2h+1}, v_{2h+2}, u_{h+1})$ (which we note are ordered), and the lower bound follows because there are always at least $6 \delta n - 4h \geq 6 \delta n - 4r$ choices, as each vertex $v$ can appear in either the first or second spot in at most $2$ \emph{ordered} hyperedges in $H_{u'}$ for any $u' \in [n]$.

The following observation asserts that the system of linear equations given by the chains are satisfied by every $x \in \Lincode$.
\begin{observation}
\label{obs:chaincompleteness}
Let $C  = (u_0, v_1, v_2, u_1, v_3, v_4, u_2, \dots, v_{2(r-1) + 1}, v_{2(r-1) + 2}, u_r) \in \cH_u^{(r)}$ be an $r$-chain, with left half $C_L$ and right half $C_R$. Then, for any $x \in \Lincode$, it holds that $ x_{u_r} + \sum_{v \in C_L} x_v + \sum_{v \in C_R} x_v = x_{u_0}$.
\end{observation}
\begin{proof}
For any chain $C$, we have that for all $h = 0,\dots, r-1$, it holds that $\{v_{2h + 1}, v_{2h + 2}, u_{h+1}\} \in H_{u_h}$, which implies that $x_{v_{2h+1}} + x_{v_{2h+2}} + x_{u_{h+1}} = x_{u_h}$ for all $x \in \Lincode$. By taking the product over all these equations, \cref{obs:chaincompleteness} follows.
\end{proof}

\parhead{Step 2: defining the Kikuchi graphs.}
In this step, we will define two linear maps $\Lincode_1 \colon \Bits^n \to \Bits^{L}$ and $\Lincode_2 \colon \Bits^n \to \Bits^R$, where $L = {[n] \choose \ell} \times [n]$, $R = {[n] \choose \ell}$, and $\ell$ is a parameter, as follows. Let $\Lincode_1(x)_{(S,v)} \defeq x_v + \sum_{v' \in S} x_{v'}$, and let $\Lincode_2(x)_T \defeq \sum_{v' \in T} x_{v'}$. 
Note that $\abs{L} = nN$ and $\abs{R} = N$, where $N = {n \choose \ell}$. 

Now, for each $u \in [n]$, we will use the set of $r$-chains $\cH_u^{(r)}$ to define a bipartite graph $G_u$ with left vertices $L$ and right vertices $R$ such that, for each edge $((S,v), T)$ in $G_u$, it holds that $\Lincode_1(x)_{(S,v)} + \Lincode_2(x)_T = x_u$. This graph $G_u$ will be the following Kikuchi graph.

\begin{definition}[Kikuchi graph]
\label{def:designkikuchi}
Let $\ell$ be a parameter, to be determined later, and let $G_u$ be the graph with left vertex set $L = {[n] \choose \ell} \times [n]$ and right vertex set $R = {[n] \choose \ell}$. For a chain $C = (u_0, v_1, v_2, u_1, v_3, v_4, u_2, \dots, v_{2(r-1) + 1}, v_{2(r-1) + 2}, u_r) \in \cH_u^{(r)}$ with left half $C_L$ and right half $C_R$, we add an edge $((S,w), T)$ to $G_u$ ``labeled'' by $C$ if $S = C_L \cup U$, $T = C_R \cup U$ where $\abs{U} = \ell - r$\footnote{Note that here we will use that all the $v_h$'s are distinct, so that $\abs{C_L} = \abs{C_R} = r$ and $\abs{C_L} + \abs{C_R} = 2r$.} and $w = u_r$. Two distinct chains may produce the same edge --- we add edges with multiplicity.
\end{definition}

We now make the following simple observations about the graph $G_u$.
\begin{observation}
\label{obs:designedgemult}
For any chain $C = (u_0, v_1, v_2, u_1, v_3, v_4, u_2, \dots, v_{2(r-1) + 1}, v_{2(r-1) + 2}, u_r) \in \cH_u^{(r)}$, the number of edges in $G_u$ ``labeled'' by $C$ is exactly ${n - 2r \choose \ell - r}$.

In particular, the average left degree of $G_u$, denoted by $d_{u,L}$ is ${n - 2r \choose \ell - r}/nN$, and the average right degree, denoted by $d_{u,R}$ is ${n - 2r \choose \ell - r}/N$.
\end{observation}
\begin{proof}
Let $C_L$ be the left half of $C$ and let $C_R$ be the right half. Because all the $v_h$'s are distinct, we have $\abs{C_L} = \abs{C_R} = r$ and $\abs{C_L \cup C_R} = 2r$. It follows that the number of pairs $((S,w),T)$ such that $((S,w),T)$ is an edge in $G_u$ labeled by $C$ is simply the number of choices for the set $U$, which is a subset of $[n] \setminus (C_L \cup C_R)$ of size $\ell - r$. Thus, there are exactly ${n - 2r \choose \ell - r}$ choices.
\end{proof}

\begin{observation}
\label{obs:designcompleteness}
For every edge $((S,w), T)$ in $G_u$ and $x \in \Lincode$, it holds that $\Lincode_1(x)_{(S,w)} + \Lincode_2(x)_{T} = x_u$.
\end{observation}
\begin{proof}
Suppose that $((S,w), T)$ in $G_u$ is an edge labeled by the chain $C$, which has left half $C_L$ and right half $C_R$. We then have that $w = u_r$, $u = u_0$, and $S = C_L \cup U$, $T = C_R \cup U$. Therefore,
\begin{flalign*}
&\Lincode_1(x)_{(S,w)} + \Lincode_2(x)_{T} = x_{u_r} + \sum_{z \in S} x_z + \sum_{z \in T} x_z \\
&= x_{u_r} + \sum_{z \in C_L} x_z + \sum_{z \in C_r} x_z + \sum_{z \in U} (x_z + x_z) = x_{u_r} + \sum_{z \in C_L} x_z + \sum_{z \in C_r} x_z = x_u \mcom
\end{flalign*}
where the last equality uses \cref{obs:chaincompleteness}.
\end{proof}

\parhead{The plan for the remainder of the proof.}
 Let us now take a brief moment to outline the steps for the remainder of the proof. To construct a $2$-LCC, it suffices to show that $G_u$ admits a matching $M_u$ of size $\Omega(N)$. Indeed, if this were the case, then the matching $M_u$ would be the matching that we require to invoke \cref{fact:2ldclb} and thus finish the proof.

To show that $G_u$ has a large matching, it suffices bound the maximum degree of the graph by $d$, as then $G_u$ must admit a matching of size at least $\abs{E(G_u)}/d$. However to do this, there are two issues to resolve. The most obvious issue is that the bipartite graph is unbalanced, i.e., $\abs{L} = n\abs{R}$, and so this prevents us from obtaining a matching of size $\Omega(\abs{L})$. This issue can be easily fixed by the following trick:\footnote{This is a nice trick of \cite{Yankovitz24} that, while it does not affect the final bounds, saves a use of the Cauchy--Schwarz inequality and thus makes the graph $G_u$ a bit simpler to describe.} for each right vertex $T \in R$, we can create $n$ copies of $T$, denoted by $T^{(1)}, \dots, T^{(n)}$, and split the edges adjacent to $T$ evenly across the copies. This decreases the average (and maximum) right degree by a factor of $(1-o(1))n$, and fixes the issue.

The second, and much more challenging problem, is that the graph $G_u$ need not be approximately biregular. Indeed, if the graph $G_u$ was exactly \emph{biregular}, then apply the above ``splitting trick'' would imply that the resulting graph has a \emph{perfect} matching of size $nN/2$.

This irregularity issue is a common problem for Kikuchi matrices and has arisen in many prior works~\cite{GuruswamiKM22,HsiehKM23,AlrabiahGKM23,KothariM23,Yankovitz24}.
The way to handle this issue is to show that $G_u$ admits a subgraph $G'_u$ that is \emph{approximately} biregular and still contains a significant fraction of the edges of $G_u$, i.e., $\abs{E(G'_u)} \geq \Omega(\abs{E(G_u)})$. We follow the terminology of prior work and call this step the ``row pruning'' step, which is so named because it involves pruning rows (and columns) of the adjacency matrix of $G_u$. This row pruning step is the crucial, and by far the most technical, component of the proof.

\parhead{Step 3: Finding a near-perfect matching in $G_u$.} We now argue that $G_u$ admits a degree-bounded subgraph $G'_u$ containing $(1-o(1))\abs{E(G_u)}$ edges. The strategy in~\cite{KothariM23} is to use the moment method to argue that with high probability, a random left (or right) vertex of the graph has degree at most $O(d_{u,L})$ (or $O(d_{u,R})$) with high probability. Here, we will diverge from the technical implementation of the proof in~\cite{KothariM23} and follow the approach of~\cite{HsiehKM23, Yankovitz24}, which is the observation that it suffices to compute first and second moments only. Indeed, it is computing higher moments that causes the loss of several extra $\log n$ factors in the proof of~\cite{KothariM23}, as compared to~\cite{Yankovitz24}.

The key reason we shall save the final $\log n$ factor is because the matchings $H_u$ are nearly perfect, i.e., they have size $\delta n$ where $\delta = \frac{1}{3} - \frac{1}{3n}$. This, combined with the careful choice of the matrix (see \cref{rem:newmatrix}) allows us to take $\ell = O(r)$ instead of $\ell = O(r^2)$, which saves a $\log n$ factor. We note that in order to get the sharp constant achieved in \cref{mthm:designs}, we need to show that $G_u$ contains a \emph{near-perfect} matching.

Let $\deg_{u,L}(S,w)$ denote the left degree of $(S,w)$ in $G_u$, and let $\deg_{u,R}(T)$ denote the right degree of $T$ in $G_u$.
In the following lemma, we compute the first\footnote{Note that \cref{obs:designedgemult} computes the first moments already.} and second moments of the degree functions. This lemma is the key technical lemma of the proof, and immediately implies the existence of a degree-bounded subgraph of $G_u$ of comparable density, as we shall shortly see.
\begin{lemma}[Second moment bounds for the left and right degree]
\label{lem:designrowpruning}
Let $\ell$ be a parameter with $\ell \geq r$ such that $r, \ell = o(n^{1/4})$. Let $G_u$ be the graph defined in \cref{def:designkikuchi}. Then, it holds that
\begin{flalign*}
&\E_{(S,w)}[\deg_L(S,w)^2] \leq (1 + o(1) + \eta)\E_{(S,w)}[\deg_L(S,w)] \mcom \\
&\E_{T}[\deg_R(T)^2] \leq (1 + o(1))\E_{T}[\deg_R(T)] \mper
\end{flalign*}
Here, the $o(1)$ is $O(\ell^2)/n$ and $\eta = n/{\ell \choose r}$.
\end{lemma}
We note that when we apply \cref{lem:designrowpruning}, we will take $r = \frac{1}{2}\log_2 n + O(\log \log n)$ and $\ell = 2r - 1$, which will end up satisfying the conditions with $\eta = 1/\polylog(n)$.

We postpone the proof of \cref{lem:designrowpruning} to \cref{sec:designrowpruning}. Let us now use \cref{lem:designrowpruning} to extract a near-perfect matching from $G_u$. We will assume that $\ell, r$ are chosen so that $\eta \leq 1/O(\log^2 n) = o(1)$, which will be the case when we choose parameters.

Using \cref{lem:designrowpruning}, we apply Chebyshev's inequality to observe that for the graph $G_u$:
\begin{enumerate}
\item There are at least $(1 - o(1))\abs{L}$ left vertices with degree $d_{u,L}(1 \pm o(1))$. Let $L'_u$ denote these left vertices.
\item There are at least $(1 - o(1))\abs{R}$ right vertices with degree $d_{u,R}(1 \pm o(1))$. Let $R'_u$ denote these right vertices.
\end{enumerate}
Let $G'_u = G_u[L'_u, R'_u]$ be the induced subgraph. First, we observe that $\abs{E(G'_u)} \geq (1 - o(1))\abs{E(G_u)}$. This is because there are at least $( 1 - o(1)) d_{u,L}\abs{L'_u} \geq (1 - o(1)) ( 1 - o(1)) d_{u,L}\abs{L} \geq (1 - o(1)) \abs{E(G)}$ edges in $G[L', R]$ and at least $(1 - o(1)) d_{u,R} \abs{R'_u} \geq (1 - o(1)) (1-o(1)) d_{u,R} \abs{R} \geq (1 - o(1)) \abs{E(G)}$ edges in $G[L, R']$, and therefore $G[L', R']$ must have at least $(1 - o(1)) \abs{E(G)}$ edges. Furthermore, each left vertex in $G'$ has degree at most $(1 + o(1)) d_{u,L}$, and similarly each right vertex has degree at most $(1 + o(1)) d_{u,R}$.

Recall that $n \cdot d_{u,L} = d_{u,R}$ and $\abs{L} = \abs{R} \cdot n$. Therefore, by making $n$ copies $T^{(1)}, \dots, T^{(n)}$ of each vertex $T$ in $R$ and splitting the edges equally across all copies (and doing the same induced transformation on $G'_u$), we can create a new bipartite graph $G''_u$ with left vertex set $L$ and right vertex set $R \times [n]$ where $G''_u$ has max left (or right!) degree $(1 + o(1))d_{u,L}$ and at least $(1 - o(1)) \abs{E(G)}$ edges. Therefore, $G''_u$ contains a matching $M_u$ of size at least $(1 - o(1)) \abs{E(G)}{d_{u,L}} \geq (1 - o(1)) \abs{L}$. Note that this matching is \emph{nearly perfect}, as the graph $G''_u$ has $2\abs{L}$ vertices, $\abs{L}$ left vertices and $\abs{L}$ right vertices.

\parhead{Step 4: proving the final bound.}
Recall that we began with a linear map $\Lincode \colon \Bits^k \to \Bits^n$ that is a design $3$-LCC. We then built the maps $\Lincode_1 \colon \Bits^n \to \Bits^{L}$ and $\Lincode_2 \colon \Bits^n \to \Bits^R$, where $L = {[n] \choose \ell} \times [n]$ and $R = {[n] \choose \ell}$, and the matchings $M_u$ for each $u \in [n]$ on the left vertex set $L$ and the right vertex set $R \times [n]$. To do this, we needed to apply \cref{lem:designrowpruning}, which requires that $\ell, r = o(n^{1/4})$. We thus set $r = \ceil{\frac{1}{2}\log_2 n + \Gamma \log_2 \log_2 n}$ for a sufficiently large constant $\Gamma$ and $\ell = 2r - 1$, which satisfies the conditions. We additionally have $\eta = 1/\log_2^2 n$, as
\begin{flalign*}
{\ell \choose r} = {2r - 1 \choose r} \geq \frac{2^{2r - 1}}{2r} \geq \frac{n \cdot 2^{\Gamma \log_2 \log_2 n}}{O(\log n)} \geq n \cdot (\log_2 n)^{\Gamma - 1 - o(1)} \geq n (\log_2^2 n) \mcom
\end{flalign*}
where we use that ${2r - 1 \choose t}$ is maximized at $t = r$ and $t = r - 1$.

Let $\Lincode'_2 \colon \Bits^n \to \Bits{R \times [n]}$ be the map where $\Lincode'_2(x)_{T^{(h)}} = \Lincode_2(x)_T$, where $T^{(h)}$ is the $h$-th copy of $T$ in $R \times [n]$. A simple corollary of \cref{obs:designcompleteness} is that, for any $x \in \Lincode$, $u \in [n]$, and edge $((S,w), T^{(h)})$ in $M_u$, it holds that $\Lincode_1(x)_{(S,w)} + \Lincode'_2(x)_{T^{(h)}} = x_u$. In particular, since $\Lincode$ is systematic, for any $i \in [k]$, edge $((S,w), T^{(h)})$ in $M_u$, and $b \in \Bits^k$, it holds that $\Lincode_1(x)_{(S,w)} + \Lincode'_2(x)_{T^{(h)}} = x_i = b_i$.

Let $\Lincode' \colon \Bits^n \to \Bits^{L \cup (R \times [n])} \cong \Bits^{2nN}$ be the map where $\Lincode'(x)_{(S,w)} = \Lincode_1(x)$ and $\Lincode'(x)_{T^{(h)}} = \Lincode'_2(x)_{T^{(h)}}$. We have that $\Lincode \circ \Lincode'$ is linear map from $\Bits^k \to \Bits^{2nN}$ and that $M_i$ is a matching of size $\geq (1 - o(1)) nN = \frac{1}{2}(1 - o(1)) \cdot 2 nN$ that decodes $b_i$. Therefore, by \cref{fact:2ldclb}, we conclude that $(1 - o(1)) k \leq \log_2 N \leq (\ell + 1) (\log_2 n) = 2 r \log_2 n = (1+o(1)) (\log_2 n)^2$, which proves \cref{mthm:designs}.

\subsection{Bounding the second moment of the left and right degrees: proof of \cref{lem:designrowpruning}}
\label{sec:designrowpruning}
In this subsection, we compute upper bounds on the second moments of degree functions. This constitutes the main technical component of the proof.

As one can imagine, computing second moments requires counting the number of chains $C \in \cH_u^{(r)}$ where the left half $C_L$ (or right half $C_R$) contains a particular set $Z$. Because of this, we first prove the following claim.

\begin{claim}[Ideal smoothness of chains from designs]
\label{claim:designsmoothness}
Let $H$ be a design $3$-LCC and let $H_1, \dots, H_n$ be the $3$-uniform hypergraphs defined in \cref{rem:designtonormalform}. Let $r \geq 1$ be an integer, and let $Z \subseteq [n]$ be a subset of size $t$, for some $0 \leq t \leq r$. Then, the number of chains $C \in \cH_u^{(r)}$ with $Z \subseteq C_R$ is at most ${r \choose t} t! (3 \delta n)^{r - t} \cdot 2^r$. And, for any $w \in [n]$, the number of chains $C \in \cH_u^{(r)}$ with tail $w$ and $Z \subseteq C_L$ is at most ${r \choose t} t! (3 \delta n)^{r - t - 1} \cdot 2^r$ if $t \leq r - 1$ and $r! \cdot 2^r$ if $\abs{Z} = r$. 
\end{claim}
\begin{proof}
First, let us count the number of chains $C \in \cH_u^{(r)}$ with $Z \subseteq C_R$. We compute this in a similar way to our upper bound on $\abs{\cH_u^{(r)}}$. First, we pick the ${r \choose t}$ locations in $C_R$ (recall that $C_R$ is implicitly ordered by the order that the vertices appear in the chain) that will contain $Z$, and then we pick one of the $t!$ ways of ordering the entries of $Z$ in these locations. Formally, we view this as fixing an ordered tuple $Q \in \{[n] \cup \star\}^{r}$, where the set of non-$\star$ elements of $Q$ is equal to $Z$. The notation $Q_h = \star$ means that the element $v_{2(h-1) + 2}$ in the chain $C$ is ``free'', and $Q_h = v$ means that we must have $v_{2(h-1) + 2} = v$.

Next, we count the number of chains as follows. We start with $u_0 = u$, and then we choose an ordered constraint $(v_1, v_2, u_{1}) \in H_{u_0}$ as follows. If $Q_1 \ne \star$, then we clearly have at most $2$ choices, as we have forced $v_2 = v$ for where $v = Q_1$, which leaves at most one (unordered) $C \in H_{u_0}$ that contains $v$, and then we have $2$ ways to order $C$. If this is not one of the locations where we have placed an entry of $Z$, i.e., $Q_1 = \star$, then we have at most $6\delta n$ choices. In total, we pay at most ${r \choose t} t! (6 \delta n)^{r - \abs{Z}} 2^{\abs{Z}} = {r \choose t} t! (3 \delta n)^{r - \abs{Z}} 2^{r}$.

Now, we fix $w \in [n]$ and count the number of chains $C \in \cH_u^{(r)}$ with tail $w$ and $Z \subseteq C_L$. We first observe that if $\abs{Z} = r$, then we have at most $2^r \cdot r!$ choices. Indeed, this means that $Z = C_L$, so we first pick an ordering on $Z$ (to determine the ordering of the vertices in $C_L$), and then we pay a factor of $2$ per step in the chain (as in the analysis in the previous paragraph). In total, there are $2^r \cdot r!$ choices.

Next, suppose that $\abs{Z} \leq r - 1$. As before, we pay ${r \choose t} \cdot t!$ to determine $Q$, i.e., the locations and ordering of $Z$ within the (ordered) set $C_L$. Let us now consider a fixed choice of the locations and ordering. We have two cases.

In the first case, suppose that $Q_r = \star$, i.e., the vertex of $C_L$ in the ``last link'' (namely, $v_{2(r-1) + 1}$), is not one of the locations chosen. Then, we can proceed as in the case of $C_R$, where we pay a factor of $2$ to choose a link where $v_{2h+1}$ is determined by $Q$, and a factor of $6 \delta n$ on the other steps. There is one exception, which is the last step of the chain. Now, because we have also fixed the tail $w$, there are again only $2$ choices for this step, even though $Q_r = \star$. Thus, in total, we have paid at most $2^{\abs{Z} + 1} (6 \delta n)^{r - \abs{Z} - 1} = (3 \delta n)^{r - \abs{Z}} \cdot 2^r$.

In the second case, suppose that $Q_r \ne \star$, so that the vertex $v_{2(r-1) + 1}$ is one of the locations chosen. Let $h^*$ denote the index of the last $\star$ in $Q$, so $Q_{h^*} = \star$ and $Q_h \ne \star$ for all $h^* < h \leq r$. We now start \emph{at the tail} of the chain and work our way backwards until we reach the $h$-th link in the chain. In the first step, we have already fixed the tail $w$ and the vertex $v_{2(r-1) + 1}$, and so because $H$ is a \emph{design}, there are at most $2$ \emph{ordered} tuples $(v, v', v_{2(r-1) + 1}, w)$ where $\{v, v', v_{2(r-1) + 1}, w\} \in H$, as there is one such unordered tuple and then we can swap the locations of $v$ and $v'$. We continue backwards along the chain in this way until we reach the location $h^*$, so that $v_{2(h^* - 1) + 1}$ is not determined by $Q$ since $Q_{h^*} = \star$. In particular, we have completely determined $u_{h^*}$, along with the all elements \emph{after} $u_{h^*}$ in the chain, namely $(v_{2h^* + 1}, v_{2h^* + 2}, \dots, u_r)$.

Next, we proceed from the start of the chain, again paying $2$ for each non-$\star$ entry and $6 \delta n$ for each $\star$ entry, until we reach the $h^*$-th link. We have thus determined the chain up until (and including) $u_{h^* - 1}$, i.e., $(u_0, v_1, v_2, \dots, u_{h^* - 1})$. For the final $2$ vertices $(v_{2(h^* - 1) + 1}, v_{2(h^* - 1) + 2})$, we have at most $2$ choices, because there is at most one hyperedge in $H_{u_{h^* - 1}}$ that contains $u_{h^*}$, and then we have $2$ ways to order the vertices. In total, we have paid $(6 \delta n)^{r - \abs{Z} - 1} \cdot 2^{\abs{Z} + 1} = (3 \delta n)^{r - \abs{Z} - 1} \cdot 2^r$, the same as in the other case.

In total, when $\abs{Z} = t \leq r - 1$, we have at most ${r \choose t} t! (3 \delta n)^{r - \abs{Z} - 1} \cdot 2^r$ choices.
\end{proof}

With \cref{claim:designsmoothness} in hand, we are almost ready to compute the second moments. To begin, we will first compute good upper bounds on the first moments $\E_{(S,w)}[\deg_{u,L}(S,w)]$ and $\E_{T}[\deg_{u,R}(T)]$. For the remainder of the proof, we may omit the subscript $u$ in some places for convenience.

We have
\begin{flalign*}
&\frac{1}{{n \choose \ell}} {n - 2r \choose \ell - r} (6 \delta n - 4r)^r \leq d_R = \E_{T}[\deg_R(T)] \leq  \frac{1}{{n \choose \ell}} {n - 2r \choose \ell - r} (6 \delta n)^r \mcom \\
& \frac{1}{n \cdot {n \choose \ell}} {n - 2r \choose \ell - r}  \cdot (6 \delta n - 4r)^r \leq d_L = \E_{(S,v)}[\deg_L(S,v)] \leq  \frac{1}{n \cdot {n \choose \ell}} {n - 2r \choose \ell - r}  \cdot (6 \delta n)^r \mper
\end{flalign*}
This is because each chain $C$ contributes ${n - 2r \choose \ell - r}$ edges to the graph $G$, and we have already computed $(6 \delta n - 4r)^r \leq \abs{\cH_u^{(r)}} \leq (6 \delta n)^r$. We also clearly have $(6 \delta n - 4r)^r \geq (6 \delta n)^r (1 - O(r^2/n))$, and so we have:
\begin{flalign}
&\left(1 - \frac{O(r^2)}{n}\right)\frac{1}{{n \choose \ell}} {n - 2r \choose \ell - r} (6 \delta n)^r \leq d_R = \E_{T}[\deg_R(T)] \leq  \frac{1}{{n \choose \ell}} {n - 2r \choose \ell - r} (6 \delta n)^r \mcom \label{eq:rightmean}\\
&\left(1 - \frac{O(r^2)}{n}\right) \frac{1}{n \cdot {n \choose \ell}} {n - 2r \choose \ell - r}  \cdot (6 \delta n)^r \leq d_L = \E_{(S,v)}[\deg_L(S,v)] \leq  \frac{1}{n \cdot {n \choose \ell}} {n - 2r \choose \ell - r}  \cdot (6 \delta n)^r \label{eq:leftmean} \mper
\end{flalign}

\parhead{Computing second moment of the right degree.} We now compute the second moments. We will begin with $\E_{T}[\deg_R(T)^2]$, as this case is simpler. We have
 \begin{flalign*}
 &\E_{T}[\deg_R(T)^2] \\
 &\leq  \sum_{\substack{C = (C_L, C_R, w) \\ C' = (C'_L, C'_R, w')}} \Pr[C_R, C'_R \subseteq T] \ \ \text{($T$ adjacent to edge labeled by $C$ implies $C_R \subseteq T$)}\\
 &= \sum_{C = (C_L, C_R, w)} \sum_{t = 0}^r \sum_{\substack{C' = (C'_L, C'_R, w')  \\ \abs{C_R \cap C'_R} = t}} \Pr[C_R, C'_R \subseteq T] \\
 &= \sum_{C = (C_L, C_R, w)} \sum_{t = 0}^r \sum_{\substack{C' = (C'_L, C'_R, w')  \\ \abs{C_R \cap C'_R} = t}} \frac{{n \choose \ell - (2r - t)}}{{n \choose \ell}} \ \ \text{(as $C_R \cup C'_R \subseteq T$ and $\abs{C_R \cup C'_R} = 2r - t$)}\\
 &\leq \sum_{C = (C_L, C_R, w)} \sum_{t = 0}^r {r \choose t} \cdot {r \choose t} t! (3 \delta n)^{r - t} \cdot 2^r \cdot \frac{{n \choose \ell - (2r - t)}}{{n \choose \ell}} \ \text{(by \cref{claim:designsmoothness} and ${r \choose t}$ to pick $Z \subseteq C_R$ where $C_R \cap C'_R = Z$)}\\
 &\leq \sum_{t = 0}^r (6 \delta n)^r {r \choose t} {r \choose t} t!  (3 \delta n)^{r - t} \cdot 2^r \cdot \frac{{n \choose \ell - (2r - t)}}{{n \choose \ell}} \\
 &\leq \left(1 + \frac{O(r^2)}{n}\right) d_R^2  \sum_{t = 0}^r  {r \choose t} {r \choose t} t! (3 \delta n)^{- t} \frac{{n \choose \ell}{n \choose \ell - (2r - t)}}{{n - 2r \choose \ell - r} {n - 2r \choose \ell - r}} \ \ \text{(by \cref{eq:rightmean}} \mper
 \end{flalign*}
 Now, we apply \cref{fact:binest} to conclude that
 \begin{flalign*}
 &\E_{T}[\deg_R(T)^2] \leq \left(1 + \frac{O(\ell^2)}{n}\right)  d_R^2  \sum_{t = 0}^r  {r \choose t} (3 \delta n)^{- t} n^t \frac{{\ell - r \choose r - t}}{{\ell \choose r}}  \\
 &= \left(1 + \frac{O(\ell^2)}{n}\right)  d_R^2  \sum_{t = 0}^r (3 \delta)^{-t} \frac{{r \choose t}{\ell - r \choose r - t}}{{\ell \choose r}} \mper
\end{flalign*}
Now, we observe that $\sum_{t = 0}^r \frac{{r \choose t}{\ell - r \choose r - t}}{{\ell \choose r}}  = 1$, as this is the probability mass function of a hypergeometric distribution, and that $3 \delta = 1 - \frac{1}{n}$ (as $H$ is a \emph{design}), and so $(3 \delta)^{-t} \leq (3 \delta)^{-r} \leq \left(1 + \frac{O(r)}{n}\right)$.
Thus, 
 \begin{flalign*}
 &\E_{T}[\deg_R(T)^2] \leq \left(1 + \frac{O(\ell^2)}{n}\right)  d_R^2 \mcom
 \end{flalign*}
 which gives the desired bound on the second moment.
 
 \parhead{Computing second moment of left degree.} We now compute $\E_{(S,v)}[\deg_L(S,v)^2]$. We have
 \begin{flalign*}
 &\E_{(S,v)}[\deg_L(S,v)^2] \leq  \sum_{C = (C_L, C_R, w), C' = (C'_L, C'_R, w)} \Pr[C_L, C'_L \subseteq S \wedge v = w] \ \ \text{(both chains have same fixed tail $w$)} \\
 &=  \sum_{C = (C_L, C_R, w)} \sum_{t = 0}^r \sum_{\substack{C' = (C'_L, C'_R, w)  \\ \abs{C_L \cap C'_L} = t}} \Pr[C_L, C'_L \subseteq S \wedge v = w] \\
 &=   \left(\sum_{C = (C_L, C_R, w)} \sum_{t = 0}^{r-1} \sum_{\substack{C' = (C'_L, C'_R, w)  \\ \abs{C_L \cap C'_L} = t}}\Pr[C_L, C'_L \subseteq S \wedge v = w]\right)  + \frac{{n - 2r \choose \ell - r}}{n \cdot {n \choose \ell}}  \cdot (6 \delta n)^r \cdot r! 2^r  \mcom
 \end{flalign*}
 where the last equality is because when $t = r$, then $C_L = C'_L$, and so $\Pr[C_L \subseteq S \wedge v = w] = \frac{{n - 2r \choose \ell - r}}{n \cdot {n \choose \ell}}$, and by \cref{claim:designsmoothness}, there are $r! 2^r$ choices for $C'$.

 Let us quickly handle this second term. We have by \cref{eq:leftmean},
 \begin{flalign*}
\frac{{n - 2r \choose \ell - r}}{n \cdot {n \choose \ell}}  \cdot (6 \delta n)^r \cdot r! 2^r  \leq \left(1 + \frac{O(r^2)}{n}\right) d_L \cdot r! 2^r \mper
 \end{flalign*}
We now compare $d_L$ and $r! 2^r$. By \cref{eq:leftmean}, we have
 \begin{flalign*}
 d_L \geq \left(1 - \frac{O(r^2)}{n}\right)\frac{\ell!}{n^{\ell + 1}} \cdot \frac{(n - 2r)^{\ell - r}}{(\ell - r)!} \cdot (6 \delta n)^r \geq \left(1 - \frac{O(r^2)}{n} - \frac{O(r \ell)}{n}\right) (6 \delta)^r \cdot \frac{1}{n} \cdot \frac{\ell!}{(\ell - r)!} \mper
 \end{flalign*}
 Therefore, 
  \begin{flalign*}
 &\frac{d_L}{2^r r!} \geq  \left(1 - \frac{O(r^2)}{n} - \frac{O(r \ell)}{n}\right) (3 \delta)^r \cdot \frac{1}{n} \cdot \frac{\ell!}{(\ell - r)!r!} =  \left(1 - \frac{O(r^2)}{n} - \frac{O(r \ell)}{n}\right) \left(1 - \frac{1}{n}\right)^r \cdot \frac{1}{n} \cdot {\ell \choose \ell - r} \\
 &=  \left(1 - \frac{O(r \ell)}{n}\right) \left(1 - \frac{1}{n}\right)^r \cdot \frac{1}{n} \cdot {\ell \choose \ell - r} \mper
 \end{flalign*}
As ${\ell \choose \ell - r} = \eta n$ is the definition of $\eta$ in \cref{lem:designrowpruning}, we conclude that 
  \begin{flalign*}
 \frac{d_L}{2^r r!} \geq \eta \left(1 - \frac{O(r \ell)}{n}\right) \mcom
 \end{flalign*}
 and so the second term is $\eta d_L^2 \left(1 + \frac{O(r \ell)}{n}\right)$.

 We now return to the main calculation. We have
\begin{flalign*}
 &\E_{(S,v)}[\deg_L(S,v)^2] \leq  \left(\sum_{C = (C_L, C_R, w)} \sum_{t = 0}^{r-1} \sum_{\substack{C' = (C'_L, C'_R, w)  \\ \abs{C_L \cap C'_L} = t}}\Pr[C_L, C'_L \subseteq S \wedge v = w]\right)  + \eta d_L^2 \left(1 + \frac{O(r \ell)}{n}\right) \\
 &\leq \eta d_L^2 \left(1 + \frac{O(r \ell)}{n}\right)  + \sum_{C = (C_L, C_R, w)} \sum_{t = 0}^{r-1} \sum_{\substack{C' = (C'_L, C'_R, w)  \\ \abs{C_L \cap C'_L} = t}} \frac{{n \choose {\ell - (2r - t)}}}{n {n \choose \ell}}  \ \ \text{(as $C_L \cup C'_L \subseteq S$ and $\abs{C_L \cup C'_L} = 2r - t$)} \\
 &\leq \eta d_L^2 \left(1 + \frac{O(r \ell)}{n}\right)  + \sum_{C = (C_L, C_R, w)} \sum_{t = 0}^{r-1} 
{r \choose t} {r \choose t} t! 2^r (3 \delta n)^{r - t - 1} 
 \frac{{n \choose {\ell - (2r - t)}}}{n {n \choose \ell}}  \ \ \text{(by \cref{claim:designsmoothness} and ${r \choose t}$ to pick $Z = C_L \cap C'_L$)} \\
 &\leq \eta d_L^2 \left(1 + \frac{O(r \ell)}{n}\right)  + \sum_{t = 0}^{r-1} 
(6 \delta n)^r {r \choose t} {r \choose t} t! 2^r (3 \delta n)^{r - t - 1} 
 \frac{{n \choose {\ell - (2r - t)}}}{n {n \choose \ell}}  \\
  &\leq \eta d_L^2 \left(1 + \frac{O(r \ell)}{n}\right)+  \frac{(6 \delta n)^{2r}}{3 \delta n} \sum_{t = 0}^{r-1}  {r \choose t} {r \choose t} t!  (3 \delta n)^{-t}  \frac{{n \choose {\ell - (2r - t)}}}{n {n \choose \ell}} \\
  &\leq \eta d_L^2 \left(1 + \frac{O(r \ell)}{n}\right)+ \left(1 + \frac{O(r^2)}{n}\right)d_L^2 \cdot  (3 \delta)^{-1} \sum_{t = 0}^{r-1}  {r \choose t} {r \choose t} t!  (3 \delta n)^{-t}  \frac{{n \choose \ell} {n \choose {\ell - (2r - t)}}}{ {n - 2r \choose \ell - r}   {n - 2r \choose \ell - r} }  \ \ \text{(by \cref{eq:leftmean})} \\
  &\leq \eta d_L^2 \left(1 + \frac{O(r \ell)}{n}\right)+\left(1 + \frac{O(\ell^2)}{n}\right)d_L^2 \cdot  (3 \delta)^{-1} \sum_{t = 0}^{r-1}  {r \choose t}   (3 \delta n)^{-t} n^t \frac{{\ell - r \choose r - t}}{{\ell \choose r}} \ \ \text{(by \cref{fact:binest})} \\
   &\leq \eta d_L^2 \left(1 + \frac{O(r \ell)}{n}\right)+ \left(1 + \frac{O(\ell^2)}{n}\right)d_L^2\cdot  (3 \delta)^{-1} \sum_{t = 0}^{r-1}   (3 \delta)^{-t} \frac{{r \choose t} {\ell - r \choose r - t}}{{\ell \choose r}} \mper
 \end{flalign*}
Now, we have $\sum_{t = 0}^{r} \frac{{r \choose t} {\ell - r \choose r - t}}{{\ell \choose r}} = 1$ as this is the probability mass function of a hypergeometric distribution. As $3 \delta = 1 - 1/n$, it follows that $(3 \delta)^{-t - 1} \leq (3 \delta)^{-r} \leq 1 + O(r/n)$, and therefore we conclude that $\E_{(S,v)}[\deg_L(S,v)^2] \leq \left(1 + \frac{O(\ell^2)}{n} + \eta \right)d_L^2$.

\section{From Adaptive Decoders to Chain XOR Polynomials}
\label{sec:chainpolys}
In this section, we begin the proof of \cref{mthm:nonlin}. We start with a (possibly nonlinear) smooth $3$-LCC with a (possibly adaptive) decoder. First, we define an abstract notion of a (smooth) $3$-LCC hypergraph collection and use it to define a polynomial that we call a ``chain XOR instance''. The chain XOR instances that we introduce are a generalization of chain XOR derivations constructed in~\cite{KothariM23}, and are able to handle the case of weighted and nonuniform hypergraphs.
Next, we will argue (\cref{lem:chainxorref}) that any ``chain XOR instance'' from a $3$-LCC hypergraph collection must have small value.

Finally, we will show that, given a $3$-LCC, we can extract a $3$-LCC hypergraph collection such that the resulting chain XOR instance has high value, which finishes the proof.

We begin by defining a ($\delta$-smooth) $3$-LCC hypergraph collection. One should view this as a generalization of the standard ``combinatorial'' definition of (linear) $3$-LCCs. In the below definition, the hypergraph $H_u$ is $3$-uniform and intuitively captures decoding constraints that make $3$ queries; the hypergraph $G_u$ is $2$-uniform, i.e., it is a graph, and it intuitively captures decoding constraints that only make at most $2$ queries. 
\begin{definition}[$3$-LCC hypergraph collection]
\label{def:hypergraphcollection}
A $3$-LCC hypergraph collection on $[n]$ vertices is a collection of pairs $(H_u, G_u)$, one for each $u \in [n]$, where $G_u$ is a (weighted and directed) $2$-uniform  hypergraph and $H_u$ is a (weighted and directed) $3$-uniform hypergraph\footnote{Note that \cref{def:hypergraph} requires that each tuple with nonzero weight has \emph{distinct} vertices.} such that for every $u \in [n]$, $\sum_{C \in [n]^2} \wt_{G_u}(C) + \sum_{C \in [n]^3} \wt_{H_u}(C) \leq 4$ and $\sum_{C \in [n]^3} \wt_{H_u}(C) \leq 1$.

For each $u \in [n]$, we define the polynomial $f_u(x) = \phi_u(x) + \psi_u(x)$, where $\phi_u(x) = \sum_{C \in [n]^2} \wt_{G_u}(C) x_C$ is the homogeneous degree-$2$ component of $f_u$ and $\psi_u(x) = \sum_{C \in [n]^3} \wt_{H_u}(C) x_C$ is the homogeneous degree-$3$ component of $f_u$.

We furthermore say that the hypergraph collection is $\delta$-smooth if for every $u, v \in [n]$, $\sum_{C \in [n]^2 : v \in C} \wt_{G_u}(C) + \sum_{C \in [n]^3 : v \in C} \wt_{H_u}(C) \leq \frac{1}{\delta n}$
\end{definition}

We now use the above collection of polynomials to construct \emph{chain XOR polynomials}, a generalization of chain XOR instances defined in~\cite{KothariM23}. To define these polynomials, we first define the $t$-chain hypergraphs $\cH_u^{(t)}$ and $\cG_u^{(t)}$.
\begin{definition}[$t$-chain hypergraph $\cH_u^{(t)}$]
\label{def:tchainhypergraph}
Let $t \geq 1$ be an integer, and let $(G_u, H_u)_{u \in [n]}$ denote a $3$-LCC hypergraph collection. For any $u \in [n]$, let $\cH_u^{(t)}$ denote
the weight function $\wt_{\cH_u^{(t)}} \colon [n]^{3t+1} \to \R_{\geq 0}$, i.e., from length $3t+1$ tuples of the form $C = (u_0,v_1, v_2, u_1, v_3, v_4, u_2, \dots, u_{t-1}, v_{2(t-1) + 1}, v_{2(t-1) + 2}, u_t)$ to $\R_{\geq 0}$, where $\wt_{\cH_u^{(t)}}(C)= 0$ if $u_0 \ne u$, and otherwise:
\begin{equation*}
\wt_{\cH_u^{(t)}}(C)= \prod_{h = 0}^{t-1} \wt_{H_{u_h}}(v_{2h + 1}, v_{2h + 2}, u_{h+1}) \mper
\end{equation*}
For a $t$-chain $C$, we call $u_0$ the head, the $u_h$'s the \emph{pivots} for $1 \leq h \leq t-1$, and $u_t$ the \emph{tail} of the chain $C$. The monomial associated to $C$, which we denote by $g_{C}$, is defined to be $x_{u_t} \prod_{h = 0}^{t-1} x_{v_{2h + 1}} x_{v_{2h+2}}$. We call the $t$-chain hypergraph $\cH_u^{(t)}$ ``hypergraph-tailed'', as the last link uses one of the hypergraphs $H_v$.
\end{definition}
We note that for any $u \in [n]$, $\cH_u^{(1)}$ is equivalent to $H_u$, i.e., $\cH_u^{(1)} = \{u\} \times H_u$.

\begin{definition}[$t$-chain hypergraph $\cG_u^{(t)}$]
\label{def:tchaingraph}
Let $t \geq 1$ be an integer, and let $(G_u, H_u)_{u \in [n]}$ denote a $3$-LCC hypergraph collection. For any $u \in [n]$, let $\cG_u^{(t)}$ denote
the weight function $\wt_{\cG_u^{(t)}} \colon [n]^{3t} \to \R_{\geq 0}$, i.e., from length $3t$ tuples of the form $C = (u_0,v_1, v_2, u_1, v_3, v_4, u_2, \dots, u_{t-1}, v_{2(t-1) + 1}, v_{2(t-1) + 2})$ to $\R_{\geq 0}$, where $\wt_{\cG_u^{(t)}}(C)= 0$ if $u_0 \ne u$, and otherwise:
\begin{equation*}
\wt_{\cH_u^{(t)}}(C)= \wt_{G_{u_{t-1}}}(v_{2(t-1) + 1}, v_{2(t-1) + 2}) \cdot \prod_{h = 0}^{t-2} \wt_{H_{u_h}}(v_{2h + 1}, v_{2h + 2}, u_{h+1}) \mper
\end{equation*}
Note that the chains in $\cG^{(t)}$ have no tail vertex $u_t$.
The monomial associated to $C$, which we denote by $x_{C}$, is defined to be $g_C = \prod_{h = 0}^{t-1} x_{v_{2h + 1}} x_{v_{2h+2}}$. We call the $t$-chain hypergraph $\cG_u^{(t)}$ ``graph-tailed'', as the last link uses one of the graphs $G_v$.
\end{definition}
We note that for any $u \in [n]$, $\cG_u^{(1)}$ is equivalent to $G_u$, i.e., $\cG_u^{(1)} = \{u\} \times G_u$.

We are now ready to define the chain XOR instances.
\begin{definition}[Chain XOR instance]
\label{def:chainxor}
Let $(G_u, H_u)_{u \in [n]}$ denote a $3$-LCC hypergraph collection. Let $k \leq n$ and $r \geq 0$ be an integer. For each $1 \leq t \leq r+1$, we define the ``graph-tailed'' polynomial
\begin{equation*}\Phi_b^{(t)}(x) = \sum_{i = 1}^k \sum_{C \in [n]^{3t}}  \wt_{\cG_i^{(t)}}(C) \cdot b_i g_C \mcom
\end{equation*}
and we also define the ``hypergraph-tailed'' polynomial
\begin{equation*}\Psi_b(x) = \sum_{i = 1}^k \sum_{C \in [n]^{3(r+1)+1}}  \wt_{\cH_i^{(r+1)}}(C) \cdot b_i g_C \mper
\end{equation*}
We will omit the subscript $b$ when it is clear from context. We note that in the above definitions, each $g_C$ is the monomial associated with the chain $C$, as defined in \cref{def:tchainhypergraph,def:tchaingraph}.
\end{definition}

\begin{remark}[Iterative view of the chain construction]
\label{rem:chainiterative}
We can view the chains as being constructed iteratively in the following way. We start with a fixed $u_0$, and have $2$ choices. We either pick a hyperedge $(v_1, a_2, v_2, a_2, u_1) \in H_{u_0}$, and then recurse onto $u_1$, or else we pick an edge $(v_1, a_2, v_2, a_2) \in G_{u_0}$, in which case the chain is in $\cG_u^{(1)}$ and we stop.
\end{remark}

With the above setup in hand, we can now state the main technical lemma.
\begin{lemma}[Refuting the chain XOR instances]
\label{lem:chainxorref}
Let $(G_u, H_u)_{u \in [n]}$ denote a $\delta$-smooth $3$-LCC hypergraph collection and let $k \leq n$. Let $\ell, d, r \geq 1$ be parameters such that $d^{r+1} \geq n$, $\ell \geq 6 d (r+1) / \delta$, and $\ell r = o(n)$. Furthermore, suppose that $k \geq 1/\delta$. Then, for each $1 \leq t \leq r + 1$, it holds that 
\begin{flalign*}
&\E_{b \gets \Fits^k}[\val(\Phi^{(t)}_b)] \leq O(\sqrt{k \ell r \log n}) \mcom \\
&\E_{b \gets \Fits^k}\left[\val(\Psi_b)\right] \leq \left(\frac{k(r+1)}{\delta}  O(\sqrt{k \ell r \log n})\right)^{1/2} \mper
\end{flalign*}
\end{lemma}
The proof of \cref{lem:chainxorref} has two steps. First, in \cref{sec:graphref}, we refute the graph-tailed instances. Then, in \cref{sec:decomp,sec:kikuchi,sec:rowpruning}, we refute the hypergraph-tailed instances.

To finish the proof of \cref{mthm:nonlin}, it remains to argue that, given any $(3, \delta, \eps)$-smooth LCC, one can extract a $3$-LCC hypergraph collection such that the resulting chain XOR polynomials (\cref{def:chainxor}) have large value. This is captured by the following lemma.

\begin{lemma}
\label{lem:adaptivepolys}
Let $\Code \colon \Fits^k \to \Fits^n$ be a $3$-LCC. Let $\Code' \colon \Fits^k \to \Fits^{4n}$ be defined as $\Code'(b) = (\Code(b), -\Code(b), 1^n, (-1)^n)$, i.e., $\Code'$ is a ``padded'' version of $\Code$, and let $\Dec(\cdot)$ denote its (possibly adaptive) decoder.

Then, there exists a $3$-LCC hypergraph collection $(H_u, G_u)_{u \in [4n]}$ with the following properties.
\begin{enumerate}
\item For every $1 \leq u \leq 4n$ and every codeword $x \in \Code'$, we have $f_u(x) x_u = \E[\Dec^{(x)}(u) x_u]$, where the expectation is taken over the randomness of the decoder. In particular, if $\Code$ has completeness $1 - \eps$, then $f_u(x) x_u \geq 1 - 2 \eps$ for all $x \in \Code'$.
\item If $\Code$ is systematic and has completeness $1 - \eps$, then for any $r$ such that $1 - 2(r+1)\eps > 0$, it holds that for every $b \in \Fits^k$ and $x = \Code'(b)$, $\Psi_b(x) + \sum_{t = 1}^{r+1} \Phi^{(t)}_b(x) \geq k(1 - 2(r+1)\eps)$.
\item If $\Code$ is $\delta$-smooth, then $\Code'$ is $\delta/4$-smooth, and  $(H_u, G_u)_{u \in [n']}$ is a $(\delta/c)$-smooth hypergraph collection for some constant $c \geq 4$. 
\end{enumerate}
\end{lemma}
We prove \cref{lem:adaptivepolys} in \cref{sec:adaptivesmooth}.

Let us now finish the proof of \cref{mthm:nonlin}.

\begin{proof}[Proof of \cref{mthm:nonlin}]
Let $\Code$ be a $3$-LCC that is $\delta$-smooth and has completeness $1 - \eps$. By \cref{fact:bgt17}, by adjusting $k$ by a factor of $\log(1/\delta)$, we can assume that $\Code$ is additionally systematic. By \cref{lem:adaptivepolys}, the padded code $\Code'$ is $(\delta/4)$-smooth with completeness $1 - \eps$ and has a $(\delta/4)$-smooth uniform hypergraph collection $(H_u, G_u)_{u \in [4n]}$. Let $r$ be such that $1 - 2(r+1)\eps > 0$. We have that for every $b \in \Fits^k$ and $x = \Code'(b)$, $\Psi_b(x) + \sum_{t = 1}^{r+1} \Phi^{(t)}_b(x) \geq k(1 - 2(r+1)\eps)$.

On the other hand, by \cref{lem:chainxorref}, it holds that \begin{flalign*}
&\E_{b \gets \Fits^k}[\val(\Phi^{(t)}_b)] \leq O(\sqrt{k \ell r \log n}) \mcom \\
&\E_{b \gets \Fits^k}\left[\val(\Psi_b)\right] \leq \left(\frac{k(r+1)}{\delta}  O(\sqrt{k \ell r \log n})\right)^{1/2} \mcom
\end{flalign*}
where $d$ and $\ell$ are parameters chosen so that $d^r \geq n$, $\ell \geq 6 d r / \delta$, and $\ell r = o(n)$.

First, let us handle the case in \cref{mthm:nonlin} when $\eps = 0$. Here, we set $r = O(\log n)$, $d = 2$, and $\ell = O(d r/\delta) = \delta^{-1} O(\log n)$. We clearly have that all the conditions of \cref{lem:chainxorref} are satisfied. Hence, we have that
\begin{flalign*}
&k = k(1 - 2(r+1)\eps) \leq \E_{b}[\Psi_b(\Code'(b)) + \sum_{t = 1}^{r+1} \Phi^{(t)}_b(\Code'(b))] \\
&\leq (r+1) \cdot O(\sqrt{k \ell r \log n}) + \left(\frac{k(r+1)}{\delta}  O(\sqrt{k \ell r \log n})\right)^{1/2} \leq O\left(\sqrt{\frac{k \log^5 n}{\delta}} + \frac{k^{3/4} \log^{5/4} n}{\delta^{3/4}}\right) \\
&\implies k \leq O(\log^5 n/\delta^3) \mcom
\end{flalign*}
which proves the statement when $\eps = 0$.

Now, let us consider the case when $\eps > 0$. When we apply \cref{lem:chainxorref}, we now set parameters as follows. Let $\eta > 0$, and set $r_0$ be such that $r_0 +1  = \floor{\frac{1 - \eta}{2 \eps}}$ and $r_1 = \log _2 n$. We then let $r = \min(r_0, r_1)$. Note that by choice of $r$, $\frac{1 - \eta}{2 \eps} \geq r+1$, and so $1 - 2(r+1)\eps \geq 2\eta$, and we also have $r \leq O(\log n)$.

Now, we set $d$ to be such that $d^{r+1} \geq n$, so we have to set $d = n^{1/(r+1)}$. Finally, we set $\ell = d r / \delta$. We thus have that
\begin{flalign*}
&2 \eta k = k(1 - 2(r+1)\eps) \leq \E_{b}[\Psi_b(\Code'(b)) + \sum_{t = 1}^{r+1} \Phi^{(t)}_b(\Code'(b))] \\
&\leq (r+1) \cdot O(\sqrt{k \ell r \log n}) + \left(\frac{k(r+1)}{\delta}  O(\sqrt{k \ell r \log n})\right)^{1/2} \\
&\leq (r+1) \cdot O(\sqrt{k n^{1/(r+1)} r^2 \log n/\delta}) + \left(\frac{k(r+1)}{\delta}  O(\sqrt{k n^{1/(r+1)} r^2 \log n/\delta})\right)^{1/2} \mper
\end{flalign*}
This implies that either
\begin{flalign*}
\eta^2 k \leq \frac{1}{\delta} \cdot O( n^{1/(r+1)}  \log^5 n) \mcom
\end{flalign*}
or
\begin{flalign*}
\eta^4 k \leq \frac{1}{\delta^3}  O( n^{1/(r+1)}  \log^5 n) \mper
\end{flalign*}
The second equation is always the dominant term, which finishes the proof. Note that the final $\log(1/\delta)$ loss comes from \cref{fact:bgt17}.
\end{proof}

The remainder of the paper is dedicated to proving \cref{lem:adaptivepolys,lem:chainxorref}. First, we show \cref{lem:adaptivepolys} in \cref{sec:adaptivesmooth}. We break the proof of \cref{lem:chainxorref} across \cref{sec:graphref,sec:decomp,sec:kikuchi,sec:rowpruning}, which will complete the proof of \cref{mthm:nonlin}.

\subsection{Additional terminology}
In the proof of \cref{lem:chainxorref}, we will often refer to the set of chains where some of the links, i.e., pairs $(v_{2h+1}, v_{2h+2})$ are forced to contain some $v \in [n]$. Towards this, we introduce the following terminology. 
\begin{definition}[Chains containing $Q$]
Let $t, r$ be integers with $t \leq r$. For any $Q = (Q_1, \dots, Q_t, Q_{t+1}) \in \{[n] \cup \star\}^{t+1}$, we say that a length $3r+1$ tuple $C = (u_0,v_1, v_2, u_1, v_3, v_4, u_2, \dots, u_{t-1}, v_{2(r-1) + 1}, v_{2(r-1) + 2}, u_r)$ contains $Q$, denoted by $Q \subseteq C$, if $Q_{t+1} \in \{\star, u_r\}$ and for $1 \leq h \leq t$, if $Q_h \neq \star$, then either $Q_h = v_{2(r - 1 - t + h) + 1}$ or $Q_h = v_{2(r - 1 - t + h) + 2}$.

We say that a $Q$ is \emph{contiguous} if there exists $s \leq t$ such that $Q_{h} \neq \star$ for every $h \geq s+1$ and $Q_{h} = \star$ for every $1 \leq h \leq s$, i.e., the first $s$ entries are $\star$, and the remaining entries are non-$\star$. We note that by definition, $Q_{t+1} \ne \star$ always.

We say that $Q$ is \emph{complete} if $Q$ does not contain any $\star$. We say that $Q' \supseteq Q$ if whenever $Q_h \neq \star$, $Q'_h = Q_h$. We define the size $\abs{Q}$ to be the number of coordinates in $Q$ that do not equal $\star$.
\end{definition}

Finally, we prove a simple bound on the total weight of the hyperedges in $\cH_u^{(t)}$ and $\cG_u^{(t)}$.
\begin{observation}
\label{obs:weightpreserved}
Let $(G_u, H_u)_{u \in [n]}$ denote a $3$-LCC hypergraph collection. Then, for any $t \geq 1$ and $u \in [n]$, it holds that $\sum_{C \in [n]^{3t + 1}} \wt_{\cH_u^{(t)}}(C) \leq 1$ and $\sum_{C \in [n]^{3t}} \wt_{\cG_u^{(t)}}(C) \leq 4$.
\end{observation}
\begin{proof}
Let us first prove the statement for $\cH_u^{(t)}$. This follows by induction. The base case of $t = 1$ is simple, as by definition we have 
\begin{flalign*}
&\sum_{C \in [n]^4} \wt_{\cH_u^{(1)}}(C) = \sum_{(u,C) \in [n]^4} \wt_{\cH_u^{(1)}}(u, C) = \sum_{C\in [n]^3} \wt_{H_u}(C) \leq 1 \mper
\end{flalign*}
We now show the induction step. Let $C \in [n]^{3t + 1}$ have tail $u_t$. Let $S$ denote the set of tuples in $[n]^{3t + 4}$ that extend $C$, i.e., the first $3t+1$ coordinates are $C$. We observe that $S = C \times [n]^3$. Moreover, we have
\begin{flalign*}
&\sum_{C' \in S} \wt_{\cH_u^{(t+1)}}(C') =  \sum_{C' \in [n]^3} \wt_{\cH_u^{(t)}}(C)\wt_{H_{u_t}}(C') \leq  \wt_{\cH_u^{(t)}}(C) \mper
\end{flalign*}
Summing over $C$ and applying the induction hypothesis proves the claim.

Now, we prove the statement for $\cG_u^{(t)}$. Let $C \in [n]^{3t + 1}$ have tail $u_t$. Let $S$ denote the set of tuples in $[n]^{3t + 3}$ that extend $C$, i.e., the first $3t+1$ coordinates are $C$. We observe that $S = C \times [n]^2$. We have
\begin{flalign*}
&\sum_{C' \in S} \wt_{\cG_u^{(t+1)}}(C') =  \sum_{C' \in [n]^2} \wt_{\cH_u^{(t)}}(C)\wt_{G_{u_t}}(C') \leq  4\wt_{\cH_u^{(t)}}(C) \mper
\end{flalign*}
Summing over $C$ and applying the claim for $\cH_u^{(t)}$ then proves the claim for $\cG_u^{(t)}$.
\end{proof}

\subsection{Constructing Polynomials from Adaptive Smoothed Decoders}
\label{sec:adaptivesmooth}

In this subsection, we prove \cref{lem:adaptivepolys}. 
Let $\Code \colon \Fits^k \to \Fits^n$ be a $3$-LCC with an adaptive decoder $\Dec(\cdot)$. The vast majority of the proof will be for proving Item (1), which will be done in two steps. First, we will prove the following lemma, which is an analogue of Item (1) in \cref{lem:adaptivepolys} but for the $\AND$ polynomial, which is defined below.
\begin{definition}[$\AND$ polynomial]
\label{def:andpoly}
Let $\AND \colon \Fits^2 \to \Bits$ be the function where $\AND(\sigma, \sigma') = 1$ if $\sigma = \sigma' = 1$, and $0$ otherwise. We note that $\AND(\sigma, \sigma') = \frac{1}{2}(1 + \sigma) \cdot \frac{1}{2}(1 + \sigma')$.
\end{definition}
\begin{lemma}
\label{lem:adaptiveANDpolys}
Let $\Code \colon \Fits^k \to \Fits^n$ be a $3$-LCC with an adaptive decoder $\Dec(\cdot)$ that uses at most $r$ bits of randomness. Then, for every $u \in [n]$, there are weight functions $\wt_{H_u} \colon [n] \times \Fits \times [n] \times \Fits \times [n] \times \Bits^r \to \R_{\geq 0}$ and $\wt_{G_u} \colon [n] \times \Fits \times [n] \times \Fits \times \Bits^r \to \R_{\geq 0}$ and bits $\sigma_{(u,v_1, a_1, v_2, a_2, v_3, \rand)} \in \Fits$, $\sigma_{(u,v_1, a_1, v_2, a_2, \rand)} \in \Fits$ such that for every $x \in \Code$,
\begin{flalign}
&\sum_{C = (v_1, a_1, v_2, a_2, \rand)} \left(\wt_{G_u}(C) + \sum_{v_3 \in [n]}\wt_{H_u}(C, v_3)\right) = 4 \mcom \label{eq:wtbound}\\
&\sum_{C = (v_1, a_1, v_2, a_2, \rand)} \left(\wt_{G_u}(C) + \sum_{v_3 \in [n]}\wt_{H_u}(C, v_3)\right) \cdot \AND(a_1 x_{v_1}, a_2 x_{v_2}) = 1 \mcom \label{eq:wtcompl}\\
&\sum_{C = (v_1, a_1, v_2, a_2, \rand)} \left(\wt_{G_u}(C) \sigma_{(u,C)} + \sum_{v_3 \in [n]}\wt_{H_u}(C, v_3) \sigma_{(u,C, v_3)} x_{v_3}\right) \cdot \AND(a_1 x_{v_1}, a_2 x_{v_2}) = \E[\Dec^{x}(u)] \mcom\label{eq:polycompl}
\end{flalign}
where the expectation $\E[\Dec^{x}(u)]$ is over the internal randomness of the decoder.

Furthermore, if $\Dec(\cdot)$ is $\delta$-smooth, then for any $v \in [n]$, we have
\begin{flalign*}
\sum_{\substack{(C,v_3) = (v_1, a_1, v_2, a_2, v_3)\\v_1 = v \vee v_2 = v \vee v_3 = v}}  \wt_{H_u}(C, v_3) + \sum_{\substack{C =(v_1, a_1, v_2, a_2)\\v_1 = v \vee v_2 = v}} \wt_{G_u}(C) \leq \frac{4}{\delta n} \mper
\end{flalign*}
\end{lemma}

We postpone the proof of \cref{lem:adaptiveANDpolys} to \cref{sec:adaptiveANDpolys}, and now finish the rest of the proof of \cref{lem:adaptivepolys}.

Let $\Code' \colon \Fits^k \to \Fits^{4n}$ be the ``padded'' version of $\Code$, i.e., for each $b \in \Fits^k$, $\Code'(b) = (\Code(b), -\Code(b), 1^n, (-1)^n)$. Note that if $\Code$ is systematic, then so is $\Code'$.

Let us extend $\Dec(\cdot)$ to be a decoder $\Dec'(\cdot)$ for $\Code'$ by defining its behavior on $u \in \{2n+1, \dots, 4n\}$ to be: (1) if $u$ is a ``$1$ bit'', i.e., $u \in \{2n + 1, \dots, 3n\}$, then $\Dec'(u)$ queries a random pair of the ``padded'' bits of the same sign (namely, it queries either two bits that are supposed to be $1$ or two bits that are $-1$), and (2) if $u$ is a ``$-1$ bit'', i.e., $u \in \{3n + 1, \dots, 4n\}$, then $\Dec('u)$ queries a random pair of the ``padded'' bits that have opposite signs. We note that if the original decoder $\Dec(\cdot)$ has completeness $1 - \eps$, then so does the padded decoder $\Dec'(\cdot)$, and if the original decoder is $\delta$-smooth, then the padded decoder $\Dec'(\cdot)$ is $(\delta/3)$-smooth.

\parhead{Proof of Item (1).} We are now ready to prove Item (1) in \cref{lem:adaptivepolys}. Fix $u \in [3n]$. We will now construct the desired hypergraph pair $(H'_u, G'_u)$ as follows. 

First, if $u \in [4n] \setminus [2n]$ is one of the ``constant'' padded vertices, then this is simple. We let $H'_u$ be empty, i.e., all weights are $0$, and if $u \in \{n +1, \dots, 2n\}$ is one of the ``$1$ bit'' padded vertices, then we let $G'_u$ denote the graph with weight $1/2n(n-1)$ on all ordered pairs of vertices $(v_1, v_2)$ where $v_1, v_2 \in \{2n + 1, \dots, 3n\}$ or $v_1, v_2 \in \{3n + 1, \dots, 4n\}$. If $u \in \{2n + 1, \dots, 3n\}$ is one of the ``$-1$ bit'' padded vertices, then we let $G'_u$ denote the graph with weight $1/2n^2$ on all ordered pairs of vertices $(v_1, v_2)$ where $v_1 \in \{2n + 1, \dots, 3n\}$, $v_2 \in \{3n + 1, \dots, 4n\}$ or vice-versa. It is straightforward to observe that this satisfies the desired condition, as for every $x \in \Code'$, $x_{v_1} x_{v_2} = 1$ if $v_1,v_2 \in \{n + 1, \dots, 2n\}$ or $v_1, v_2 \in \{2n + 1, \dots, 3n\}$, and in the other case $x_{v_1} x_{v_2} = -1$ holds for all codewords.

It remains to handle the case when $u \in [2n]$. We will do this for the case when $u \in [n]$, and then observe that we can handle the case of $u \in [2n] \setminus [n]$ by flipping the ``sign'' of the first query.

Let $u \in [n]$. We construct $(H'_u, G'_u)$ from the pair $(H_u, G_u)$ given to us in \cref{lem:adaptiveANDpolys}, as follows. Recall that each term $C = (v_1, a_1, v_2, a_2, \rand)$ with $\wt_{G_u}(C) > 0$ contributes the term $\wt_{G_u}(C) \sigma_{C} \AND(a_1 x_{v_1}, a_2 x_{v_2})$ in \cref{eq:polycompl}. We have that for any $x \in \Code'$,
\begin{flalign*}
&\sigma_C \AND(a_1 x_{v_1}, a_2 x_{v_2}) = \frac{1}{4} \sigma_C \left(1 + a_1 x_{v_1} + a_2 x_{v_2} + a_1 a_2 x_{v_1} x_{v_2}\right) \\
&= \frac{1}{4} \left(x_{\sigma_C^{(v_1)}} x_{1^{(v_2)}} + x_{a_1 v_1} x_{\sigma_C^{(v_2)}} + x_{\sigma_C^{(v_1)}} x_{a_2 v_2} + x_{\sigma_C a_1 v_1} x_{a_2 v_2}\right) \mcom
\end{flalign*}
where (1) for any $\sigma \in \Fits$, $x_{\sigma^{(v_1)}}$ refers to the $v_1$-th copy of $\sigma$, i.e., if $\sigma = 1$ then $x_{\sigma^{(v_1)}} = x_{2n + v_1}$ and if $\sigma = -1$ then $x_{\sigma^{(v_1)}} = x_{3n + v_1}$, and (2) $x_{a_1 v_1}$ is $x_{v_1}$ if $a_1 = 1$ and $x_{n + v_1}$, i.e., the copy of $-x_{v_1}$, if $a_1 = -1$, and similar notation is used for $x_{v_2}$.  

Now, we add $4$ edges to $G'_u$ for each such edge in $G_u$. Namely, for any $C = (v_1, a_1, v_2, a_2, \rand)$ with $\wt_{G_u}(C) > 0$ and term $\wt_{G_u}(C) \sigma_{C} \AND(a_1 x_{v_1}, a_2 x_{v_2})$, we add the $4$ edges $(\sigma_C^{(v_1)}, 1^{(v_2)})$, $(a_1 v_1, \sigma_C^{(v_2)})$, $(\sigma_C^{(v_1)}, a_2 v_2)$, $(\sigma_C a_1 v_1, a_2 v_2)$ to $G_u$, each with weight $\frac{1}{4}\wt_{G_u}(C)$. We note that it is possible to add the same edge to $G'_u$ multiple times: in this case, we ``merge'' the edges by adding their weights together.

We now process the edges in $H_u$ to form $H'_u$ in a similar way. The difference here is that we will add the ``degree $3$ term'' to $H'_u$, and all other terms will again be added to $G'_u$. More formally, each term $C = (v_1, a_1, v_2, a_2, v_3, \rand)$ with $\wt_{H_u}(C) > 0$ contributes the term $\wt_{H_u}(C) \sigma_{C} \AND(a_1 x_{v_1}, a_2 x_{v_2}) x_{v_3}$ in \cref{eq:polycompl}. From this term, we add the edge $(\sigma_C a_1 v_1, a_2 v_2, v_3)$ to $H'_u$ with weight $\frac{1}{4}\wt_{H_u}(C)$, and we add the $3$ edges $(\sigma_C^{(v_1)}, v_3)$, $(\sigma_C a_1 v_1, v_3)$, $(\sigma_C a_2 v_2, v_3)$ to $G'_u$ with weight $\frac{1}{4}\wt_{G_u}(C)$.

This defines the pair $(H'_u, G'_u)$ for all $u \in [4n] \setminus \{n+1, \dots, 2n\}$. To define the pair for $u \in \{n+1, \dots, 2n\}$, we simply observe that $-u \in [n]$, and so we flip the ``sign'' of the first vertex in each edge in $H'_{-u}$ or $G'_{-u}$, and this defines $(H'_u, G'_u)$ for $u \in \{n+1, \dots, 2n\}$.

We now need to show that the pair $(H'_u, G'_u)$ satisfies the normalization conditions of \cref{def:hypergraphcollection}, namely that $\sum_{C \in [n]^2} \wt_{G'_u}(C) + \sum_{C \in [n]^3} \wt_{H'_u}(C) \leq 4$ and $\sum_{C \in [n]^3} \wt_{H'_u}(C) \leq 1$. We note that for $u \in [4n] \setminus [2n]$, this clearly holds, so it suffices to argue this for $u \in [n]$ (which then implies the statement for $u \in [2n]$).
The first inequality follows from \cref{eq:wtbound}, as the total weight of all edges is preserved. The second inequality follows because the total weight in $H'_u$ is at most $1/4$ of the total weight in $H_u$, which is at most $4$.

Finally, to finish the proof of Item (1), let $f_u$ be the polynomial defined in \cref{def:hypergraphcollection} from $(H'_u, G'_u)$. We clearly have that for any $x \in \Code'$, $f_u(x) = \E[\Dec'(u)]$, as by construction $f_u(x)$ is equal to the left hand side of \cref{eq:polycompl}.

\parhead{Proof of Item (2).} We are now ready to prove Item (2). For simplicity, we will replace $4n$ with $n$. Let $p_{\cG_u^{(t)}}(x) = \sum_{C \in [n]^{3t}}  \wt_{\cG_u^{(t)}}(C) \cdot x_u g_C$ and $p_{\cH_u^{(t)}} = \sum_{C \in [n]^{3t+1}}  \wt_{\cH_u^{(t)}}(C) \cdot x_u g_C$ denote the ``graph-tailed'' and ``hypergraph-tailed'' polynomials \emph{with head $u$}, defined in a similar manner to the polynomials in \cref{def:chainxor}. We will show by induction on $r$ that if $1 - 2(r+1) \eps > 0$, then $p_{\cH_u^{(r+1)}}(x) + \sum_{t = 1}^{r+1} p_{\cG_u^{(r+1)}}(x)  \geq 1 - 2(r+1) \eps$.

For the base case, we observe that when $r = 0$, $p_{\cH_u^{(1)}}(x) = x_u \psi_u(x)$ and $p_{\cG_u^{(1)}}(x) = x_u \phi_u(x)$ (see \cref{def:hypergraphcollection}). Therefore, for any $x \in \Code'$, $p_{\cH_u^{(1)}}(x) + p_{\cG_u^{(1)}}(x) = x_u f_u(x) = \E[x_u \Dec^{(x)}(u)] \geq 1 - 2 \eps$, as $\Code$ (and therefore $\Code'$) has completeness $1 - \eps$. Note that we also have $\abs{1 - x_u f_u(x)} = \abs{x_u - f_u(x)} \leq 2 \eps$ for all $x \in \Code'$.

For the induction step, we have by definition (see \cref{rem:chainiterative}) that for any $x \in \Code'$,
\begin{flalign*}
&p_{\cH_u^{(r+1)}}(x) + p_{\cG_u^{(r+1)}(x)} - p_{\cH_u^{(r)}}(x) = \sum_{v \in [n]} \sum_{C \in [n]^{3r+1} : \mathrm{tail}(C) = v}  \wt_{\cH_u^{(r)}}(C) \cdot x_u g_C \cdot (x_v f_v(x) - 1) \\
&\leq  \left(\sum_{v \in [n]} \sum_{C \in [n]^{3r+1} : \mathrm{tail}(C) = v}  \wt_{\cH_u^{(r)}}(C) \cdot \abs{x_u g_C} \cdot \abs{1 - (x_v f_v(x))}\right) \\
&\leq \left(\sum_{v \in [n]} \sum_{C \in [n]^{3r+1} : \mathrm{tail}(C) = v}  \wt_{\cH_u^{(r)}}(C) \cdot 1 \cdot 2 \eps \right)  \\
&\leq 2 \eps\mper
\end{flalign*}
The final inequality uses the fact that $\sum_{C \in [n]^{3r+1}}  \wt_{\cH_u^{(r)}}(C) \leq 1$, which follows by induction using that $\sum_{C \in [n]^3} \wt_{H'_u}(C) \leq 1$.

Hence, we conclude that 
\begin{flalign*}
p_{\cH_u^{(r+1)}}(x) + \sum_{t = 1}^{r+1} p_{\cG_u^{(r+1)}}(x) = \left(p_{\cH_u^{(r+1)}}(x) + p_{\cG_u^{(r+1)}}(x) - p_{\cH_u^{(r)}}(x) \right) + \left(p_{\cH_u^{(r)}}(x)+ \sum_{t = 1}^{r} p_{\cG_u^{(r+1)}}(x)\right) \geq - 2 \eps + 1 - 2r \eps \mcom
\end{flalign*}
which finishes the proof of Item (2).

\parhead{Proof of Item (3).} The proof of smoothness is straightforward. First, if $u \in [4n] \setminus [2n]$, then the condition immediately holds by construction. Let us now consider the interesting case of $u \in [n]$. By \cref{lem:adaptiveANDpolys}, the pair $(H_u, G_u)$ satisfies the smoothness condition. Thus, it remains to verify that the pair $(H'_u, G'_u)$ is $\delta/c)$-smooth for some constant $c$. This follows immediately because, for each edge in $H'_u$ or $G'_u$ that contains some $v' \in [4n]$, we can uniquely identify $v \in [n]$ such that the ``original edge'' in $(H_u, G_u)$ that the new edge ``comes from'' contains $v$. Hence, if we consider the total weight of all hyperedges in $(H'_u, G'_u)$ containing some vertex $v' \in [4n]$, there is a $v \in [n]$ such that the weight is upper-bounded by the total weight of all hyperedges in $(H_u, G_u)$ containing $v$. The extra constant factor $c$ comes from the fact that the number of vertices is now $4n$ and the constant factor loss in \cref{lem:adaptiveANDpolys}.
\subsection{Proof of \cref{lem:adaptiveANDpolys}}
\label{sec:adaptiveANDpolys}
In this subsection, we prove \cref{lem:adaptiveANDpolys}.
Let $\Code \colon \Fits^k \to \Fits^n$ be a $3$-LCC with an adaptive decoder. For each $u \in [n]$, we use the decoding algorithm $\Dec(u)$ to define weight functions $\wt_{H_u}$ and $\wt_{G_u}$. In what follows, we consider a fixed $u \in [n]$.

First, without loss of generality, we may assume that the decoder $\Dec(u)$ makes \emph{exactly} $3$ queries. We can view the decoder as a decision tree: first, $\Dec(u)$ generates the first query $v_1$ from some distribution. Then, $\Dec(u)$ receives a bit $a_1 \in \Fits$, the answer to the query $v_1$. This answer selects the branch of the decision tree, which determines the distribution of the next query $v_2$. Then, the decoder receives another answer $a_2 \in \Fits$, which selects the branch of the decision tree, and gives the distribution of the final query $v_3$. Finally, the decoder first selects some randomness $\rand \in \Bits^r$, receives an answer $a_3$, and then it computes a (deterministic) function $f_{(v_1, a_1, v_2, a_2, v_3, \rand)}$ of $a_3$ to produce its output. This function is deterministic because the randomness is handled in $\rand$.
We note that there are exactly $4$ valid deterministic functions: $1$, $-1$, $a_3$, and $-a_3$, so $f_{(v_1, a_1, v_2, a_2, v_3, \rand)}$ must be one of these.

 For each choice of $C = (v_1, a_1, v_2, a_2, v_3, \rand) \in ([n] \times \Fits)^2 \times [n] \times \Bits^r$, we let $\wt_u(C)$ be the probability that the decoder makes the set of queries $C$ (with the appropriate answers) when given oracle access to \emph{any} $x$ that is consistent with $C$, meaning that $x_{v_1} = a_1$ and $x_{v_2} = a_2$. Indeed, this does not depend on the choice of $x$, as there is some probability $p_{v_1}$ that the decoder queries $v_1$ (which does not depend on $x$), and then given $x_{v_1} = a_1$, there is a probability $p_{v_2}$ that the decoder queries $v_2$, etc.

We now partition the query sets into two types. If $C$ is such that $f_{(v_1, a_1, v_2, a_2, v_3, \rand)}$ is a constant function $\sigma \in \Fits$ (so it does not depend on $a_3$), then we set $\wt_{G_u}(v_1, a_1, v_2, a_2, \rand) = \wt_u(C)$ and $\sigma_{(v_1, a_1, v_2, a_2, \rand)} = \sigma$. Otherwise, we have that $C$ is such that $f_{(v_1, a_1, v_2, a_2, v_3, \rand)}= \sigma a_3$, and then we set $\wt_{H_u}(v_1, a_1, v_2, a_2, v_3, \rand) = \wt_u(C)$ and $\sigma_{(v_1, a_1, v_2, a_2, v_3, \rand)} = \sigma$.

We now show that this weight function has the desired properties. Indeed, we have essentially encoded the behavior of the arbitrary decoder as this system of polynomials. 

First, let us show that
\begin{flalign*}
&\sum_{C = (v_1, a_1, v_2, a_2, \rand)} \left(\wt_{G_u}(C) + \sum_{v_3 \in [n]}\wt_{H_u}(C, v_3)\right) = 4 \mper
\end{flalign*}
Consider the decoder $\Dec'(u)$ that simulates $\Dec_u$ by generating random bits as the answers to the queries of $\Dec(u)$. It follows that the probability that $\Dec'(u)$ queries a particular $C$ is $\wt(C)/4$, and hence \cref{eq:wtbound} holds.

Next, let us show that for any $x \in \Code$
\begin{flalign*}
&\sum_{C = (v_1, a_1, v_2, a_2, \rand)} \left(\wt_{G_u}(C) + \sum_{v_3 \in [n]}\wt_{H_u}(C, v_3)\right) \cdot \AND(a_1 x_{v_1}, a_2 x_{v_2}) = 1 \mper 
\end{flalign*}
Indeed, we observe that for any $x \in \Code$ and any $C$, $\wt_{H_u}(C, v_3)\cdot \AND(a_1 x_{v_1}, a_2 x_{v_2})$ is $0$ if $C$ is inconsistent with $x$, and otherwise it is the probability that $\Dec^{x}(u)$ queries $C$, and the same statement holds for $\wt_{G_u}(C)\AND(a_1 x_{v_1}, a_2 x_{v_2})$. Hence, the sum must be $1$.

Finally, we have
\begin{flalign*}
&x_u \sum_{C = (v_1, a_1, v_2, a_2, \rand)} \left(\wt_{G_u}(C) \sigma_{C} + \sum_{v_3 \in [n]}\wt_{H_u}(C, v_3) \sigma_{(C, v_3)} x_{v_3}\right) \cdot \AND(a_1 x_{v_1}, a_2 x_{v_2}) = \E[\Dec^{x}(u) x_u] \mper
\end{flalign*}
Indeed, this is because for any $C = (v_1, a_1, v_2, a_2, v_3, \rand)$ and any $x \in \Code$, then the execution of $\Dec^{(x)}(u)$ queries $C$ with probability $\wt_{H_u}(C, v_3) \AND(a_1 x_{v_1}, a_2 x_{v_2})$, and then the output of the decoder is the decoding function, which is $\sigma_{(C, v_3)} x_{v_3}$. A similar statement holds for $C =(v_1, a_1, v_2, a_2, \rand)$ as well, which finishes the proof.

\ignore[old proof]{
\section{From Adaptive Decoders to Chain Polynomials}
\label{sec:chainpolys}
In this section, we begin the proof of \cref{mthm:nonlin}. We will transform a $3$-LCC with an adaptive decoder into a system of satisfiable polynomial constraints that we call ``chain polynomials''. The polynomials will be products of $\AND$ polynomials, which we recall below.
\begin{definition}[$\AND$ polynomial]
\label{def:andpoly}
Let $\AND \colon \Fits^2 \to \Bits$ be the function where $\AND(\sigma, \sigma') = 1$ if $\sigma = \sigma' = 1$, and $0$ otherwise. We note that $\AND(\sigma, \sigma') = \frac{1}{2}(1 + \sigma) \cdot \frac{1}{2}(1 + \sigma')$.
\end{definition}

The key structure that we shall extract from the $3$-LCC is captured by the following lemma. 
\begin{lemma}
\label{lem:adaptivepolys}
Let $\Code \colon \Fits^k \to \Fits^n$ be a $3$-LCC with an adaptive decoder $\Dec(\cdot)$. Then, for every $u \in [n]$, there are weight functions $\wt_{H_u} \colon [n] \times \Fits \times [n] \times \Fits \times [n] \to \R_{\geq 0}$ and $\wt_{G_u} \colon [n] \times \Fits \times [n] \times \Fits \to \R_{\geq 0}$ and bits $\sigma_{(u,v_1, a_1, v_2, a_2, v_3)} \in \Fits$, $\sigma_{(u,v_1, a_1, v_2, a_2)} \in \Fits$ such that for every $x \in \Code$,
\begin{flalign}
&\sum_{C = (v_1, a_1, v_2, a_2)} \left(\wt_{G_u}(C) + \sum_{v_3 \in [n]}\wt_{H_u}(C, v_3)\right) = 4 \mcom \label{eq:wtbound}\\
&\sum_{C = (v_1, a_1, v_2, a_2)} \left(\wt_{G_u}(C) + \sum_{v_3 \in [n]}\wt_{H_u}(C, v_3)\right) \cdot \AND(a_1 x_{v_1}, a_2 x_{v_2}) = 1 \mcom \label{eq:wtcompl}\\
&x_u \sum_{C = (v_1, a_1, v_2, a_2)} \left(\wt_{G_u}(C) \sigma_{(u,C)} + \sum_{v_3 \in [n]}\wt_{H_u}(C, v_3) \sigma_{(u,C, v_3)} x_{v_3}\right) \cdot \AND(a_1 x_{v_1}, a_2 x_{v_2}) = \E[\Dec^{x}(u) x_u] \mcom\label{eq:polycompl}
\end{flalign}
where the expectation $\E[\Dec^{x}(u) x_u]$ is over the internal randomness of the decoder. In particular, if $\Dec$ has perfect completeness, then $\E[\Dec^{x}(u) x_u] = 1$.

Furthermore, if $\Dec(\cdot)$ is $\delta$-smooth, then for any $v \in [n]$, we have
\begin{flalign*}
\sum_{\substack{(C,v_3) = (v_1, a_1, v_2, a_2, v_3)\\v_1 = v \vee v_2 = v \vee v_3 = v}}  \wt_{H_u}(C, v_3) + \sum_{\substack{C =(v_1, a_1, v_2, a_2)\\v_1 = v \vee v_2 = v}} \wt_{G_u}(C) \leq \frac{4}{\delta n} \mper
\end{flalign*}
\end{lemma}
We prove \cref{lem:adaptivepolys} in \cref{sec:adaptivesmooth}.

We now continue and use the above collection of polynomials to construct \emph{polynomial chains}, a generalization of chain XOR instances defined in~\cite{KothariM23}.

\begin{definition}[$t$-chain hypergraph $\cH_u^{(t)}$]
\label{def:tchainpolyhypergraph}
Let $t \geq 1$ be an integer. For any $u \in [n]$, let $\cH_u^{(t)}$ denote
the weight function $\wt_{\cH_u^{(t)}} \colon \left([n] \times ([n] \times \Fits)^2\right)^{t} \times [n] \to \R_{\geq 0}$, i.e., from tuples of the form $C = (u_0,v_1, a_1, v_2, a_2, u_1, v_3, a_3, v_4, a_4, u_2, \dots, u_{t-1}, v_{2(t-1) + 1},a_{2(t-1) + 1}, v_{2(t-1) + 2}, a_{2(t-1) + 1}, u_t)$ to $\R_{\geq 0}$, where $\wt_{\cH_u^{(t)}}(C)= 0$ if $u_0 \ne u$, and otherwise:
\begin{equation*}
\wt_{\cH_u^{(t)}}(C)= \prod_{h = 0}^{t-1} \wt_{H_{u_h}}(v_{2h + 1}, a_{2h+1}, v_{2h + 2}, a_{2h+2}, u_{h+1}) \mper
\end{equation*}
For a $t$-chain $C$, we call $u_0$ the head, the $u_h$'s the \emph{pivots} for $1 \leq h \leq t-1$, and $u_t$ the \emph{tail} of the chain $C$. We call $C_L = (v_1, a_1, v_3, a_3, \dots, v_{2(t-1) + 1},a_{2(t-1) + 1})$ the \emph{left half} of the chain and $C_R = (v_2, a_2, v_4, a_4, \dots, v_{2(t-1) + 2},a_{2(t-1) + 2})$ the \emph{right half}.

The $h$-th link in defined to be $(u_h, v_{2h + 1}, a_{2h+1}, v_{2h + 2}, a_{2h+2}, u_{h+1})$.
\end{definition}

\begin{definition}[$t$-chain hypergraph $\cG_u^{(t)}$]
\label{def:tchainpolygraph}
Let $t \geq 1$ be an integer. For any $u \in [n]$, let $\cG_u^{(t)}$ denote
the weight function $\wt_{\cG_u^{(t)}} \colon \left([n] \times ([n] \times \Fits)^2\right)^{t} \to \R_{\geq 0}$, i.e., from tuples of the form $C = (u_0,v_1, a_1, v_2, a_2, u_1, v_3, a_3, v_4, a_4, u_2, \dots, u_{t-1}, v_{2(t-1) + 1},a_{2(t-1) + 1}, v_{2(t-1) + 2}, a_{2(t-1) + 1})$\footnote{Note the difference with \cref{def:tchainpolyhypergraph}: there is no final tail vertex here.} to $\R_{\geq 0}$, where $\wt_{\cG_u^{(t)}}(C)= 0$ if $u_0 \ne u$, and otherwise:
\begin{equation*}
\wt_{\cG_u^{(t)}}(C)= \wt_{G_{u_{t-1}}}(v_{2(t-1) + 1}, a_{2(t-1) + 1}, v_{2(t-1) + 2}, a_{2(t-1) + 2}) \cdot \prod_{h = 0}^{t-2} \wt_{H_{u_h}}(v_{2h + 1}, a_{2h+1}, v_{2h + 2}, a_{2h+2}, u_{h+1}) \mper
\end{equation*}
As before, we call $C_L = (v_1, a_1, v_3, a_3, \dots, v_{2(t-1) + 1},a_{2(t-1) + 1})$ the \emph{left half} of the chain and $C_R = (v_2, a_2, v_4, a_4, \dots, v_{2(t-1) + 2},a_{2(t-1) + 2})$ the \emph{right half}.

Note that the chains in $\cG^{(t)}$ have no tail vertex $u_t$. We call the $t$-chain hypergraph $\cG_u^{(t)}$ ``graph-tailed'', as the ``last link'' has $2$ vertices only.
\end{definition}
\begin{remark}[Iterative view of the chain construction]
We can view the chains as being constructed iteratively in the following way. We start with a fixed $u_0$. Then, we add $(v_1, a_2, v_2, a_2)$. We now have $2$ choices. We can either stop (and the chain is then in $\cG_u^{(1)}$), or we can add $u_1$ to the end of the chain (and then the chain is in $\cH_u^{(1)}$). For each chain in $\cH_u^{(1)}$, we can then continue by adding on a ``link'' at $u_1$. On the other hand, there is no way to continue a chain in $\cG_u^{(1)}$ in this way, as it does not contain $u_1$.
\end{remark}

\begin{definition}[Chain Polynomials]
\label{def:chainpolynomials}
Let $C = (u_0,v_1, a_1, v_2, a_2, u_1, \dots, u_t)$ be a $t$-chain in $\cH_u^{(t)}$. The chain polynomial, denoted by $f_{C}$ is a polynomial in the variables $x \vert_{C_L}, x \vert_{C_R}, x_{u_t}$ (where $C_L$ and $C_R$ are the right and left halves of the chains), defined as
\begin{flalign*}
f_C(x \vert_{C_L}, x \vert_{C_R}, x_{u_t}) = x_{u_t} \prod_{h = 0}^{t-1} \AND(a_{2h+1} x_{v_{2h + 1}}, a_{2h+2} x_{v_{2h + 2}}) \sigma_{(u_{h-1}, v_{2h+1}, a_{2h+1}, v_{2h+2}, a_{2h+2}, u_h)} \mper
\end{flalign*}
For a chain $C \in \cG_u^{(t)}$, we let
\begin{flalign*}
f_C(x \vert_{C_L}, x \vert_{C_R}) = \prod_{h = 0}^{t-1} \AND(a_{2h+1} x_{v_{2h + 1}}, a_{2h+2} x_{v_{2h + 2}})\sigma_{(u_{h-1}, v_{2h+1}, a_{2h+1}, v_{2h+2}, a_{2h+2})} \mper
\end{flalign*}
\end{definition}

We are now ready to state the key facts about the chain polynomials.
\begin{claim}[Key facts of chain polynomials]
\label{claim:chaincompleteness}
Let $\Code \colon \Fits^k \to \Fits^n$ be a systematic $3$-LCC with a (potentially adaptive) decoder. Fix $r \geq 0$, and for $1 \leq t \leq r + 1$, let $\cG_u^{(t)}, \cH_u^{(t)}$ for $u \in [n]$ be the chains defined in \cref{def:tchainpolyhypergraph,def:tchainpolygraph}, constructed from the polynomial system of equations in \cref{lem:adaptivepolys}. Then, for each $u \in [n]$, the following holds:
\begin{enumerate}[(1)]
\item The chain polynomials correctly decode $x_u$. Namely, for each $x \in \Code$, it holds that
\begin{flalign*}
&x_u \left(\sum_{C \in \cH_u^{(r+1)}} \wt_{\cH_u^{(r+1)}}(C) f_C(x \vert_{C_L}, x \vert_{C_R}, x_{u_{r+1}}) + \sum_{t = 1}^r \sum_{C \in \cG_u^{(t)}} \wt_{\cG_u^{(t)}}(C) f_C(x \vert_{C_L}, x \vert_{C_R})\right) = 1 \mcom
\end{flalign*}
\item The total weight of the chains of length at most $r+1$ is
\begin{flalign*}
&\sum_{C \in \cH_u^{(r+1)}} \wt_{\cH_u^{(t)}}(C)  + \sum_{t = 1}^{r+1} \sum_{C \in \cG_u^{(t)}} \wt_{\cG_u^{(t)}}(C) \leq 4^{r+1} \mper
\end{flalign*}
\end{enumerate}
\end{claim}

\begin{proof}
Let us first show the first equation. We prove this by induction on $r$. The base case of $r = 0$ is simple, as we have
\begin{flalign*}
&x_u \left(\sum_{C \in \cH_u^{(1)}} \wt_{\cH_u^{(1)}}(C) f_C(x \vert_{C_L}, x \vert_{C_R}, x_{u_{1}}) + \sum_{C \in \cG_u^{(1)}} \wt_{\cG_u^{(1)}}(C) f_C(x \vert_{C_L}, x \vert_{C_R})\right) \\
&=x_u \left(\sum_{C = (v_1, a_1, v_2, a_2, v_3)} \wt_{H_u}(C) \AND(a_1 x_{v_1}, a_2 x_{v_2}) \sigma_{(u,C)} x_{v_3} + \sum_{C = (v_1, a_1, v_2, a_2)} \wt_{G_u}(C) \AND(a_1 x_{v_1}, a_2 x_{v_2}) \sigma_{(u,C)}\right) \\
&=\E[x_u \Dec^{x}(u)] = 1 \ \ \text{(by \cref{eq:polycompl})} \mper
\end{flalign*}

We now prove the induction step. Suppose that 
\begin{flalign*}
&x_u \left(\sum_{C \in \cH_u^{(r)}} \wt_{\cH_u^{(r)}}(C) f_C(x \vert_{C_L}, x \vert_{C_R}, x_{u_{r}}) + \sum_{t = 1}^r \sum_{C \in \cG_u^{(t)}} \wt_{\cG_u^{(t)}}(C) f_C(x \vert_{C_L}, x \vert_{C_R})\right) = 1 \mper
\end{flalign*}
We then have that for each $C \in \cH_u^{(r)}$ with tail $u_r$ and $\wt_{\cH_u^{(r)}}(C) > 0$,
\begin{flalign*}
&f_C(x \vert_{C_L}, x \vert_{C_R}, x_{u_{r}}) = \left( \prod_{h = 0}^{r - 1} \AND(a_{2h+1} x_{v_{2h + 1}}, a_{2h+2} x_{v_{2h + 2}}) \sigma_{(u_{h-1}, v_{2h+1}, a_{2h+1}, v_{2h+2}, a_{2h+2}, u_h)} \right) \cdot x_{u_r} \mcom
\end{flalign*}
and we have (via the base case) that
\begin{flalign*}
&x_{u_r} = \sum_{C = (v_{2r+1}, a_{2r+1}, v_{2r+2}, a_{2r+2}, u_{r+1})} \wt_{H_u}(C) \AND(a_{2r+1} x_{v_{2r+1}}, a_{2r+2} x_{v_{2r+2}}) \sigma_{(u_r,C, u_{r+1})} x_{u_{r+1}} \\
&+ \sum_{C = (v_{2r+1}, a_{2r+1}, v_{2r+2}, a_{2r+2})} \wt_{G_u}(C) \AND(a_{2r+1} x_{v_{2r+1}}, a_{2r+2} x_{v_{2r+2}}) \sigma_{(u_r,C)} \mper
\end{flalign*}
We now simply multiply the two polynomials and sum over $C \in \cH_u^{(r)}$ to finish the proof of Item (1).

We now turn to Item (2). We will again prove this by induction, where the base case follows from \cref{eq:wtbound}. To prove the induction step, we observe that
\begin{flalign*}
\sum_{C \in \cH_u^{(r+1)}} \wt_{\cH_u^{(r+1)}}(C) + \sum_{C \in \cG_u^{(r+1)}} \wt_{\cG_u^{(r+1)}}(C) = 4 \sum_{C \in \cH_u^{(r)}} \wt_{\cH_u^{(r)}}(C) \mper
\end{flalign*}
Indeed, this follows by (1) picking a length $r$-chain $C$, (2) extending it to a length $r+1$ chain (that is either in $\cH_u^{(r+1)}$ or $\cG_u^{(r+1)}$), and then applying \cref{eq:wtbound}. We then have by the induction hypothesis
\begin{flalign*}
&\sum_{C \in \cH_u^{(r+1)}} \wt_{\cH_u^{(r+1)}}(C)  + \sum_{t = 1}^{r+1} \sum_{C \in \cG_u^{(t)}} \wt_{\cG_u^{(t)}}(C) = 4\sum_{C \in \cH_u^{(r)}} \wt_{\cH_u^{(r)}}(C)  + \sum_{t = 1}^{r} \sum_{C \in \cG_u^{(t)}} \wt_{\cG_u^{(t)}}(C) \\
&\leq 4\left(\sum_{C \in \cH_u^{(r)}} \wt_{\cH_u^{(r)}}(C)  + \sum_{t = 1}^{r} \sum_{C \in \cG_u^{(t)}} \wt_{\cG_u^{(t)}}(C)\right) \leq 4^{r+1} \mcom
\end{flalign*}
which finishes the proof of Item (2).
\end{proof}

We are now ready to define the chain polynomial instances.
\begin{definition}[Chain polynomial instance]
Let $r \geq 1$ be an integer and let $b \in \Fits^k$. For each $1 \leq t \leq r$, we define the ``graph-tailed'' polynomial
\begin{equation*}\Phi_b^{(t)}(x) = \sum_{i = 1}^k \sum_{C \in \cG_i^{(t)}}  \wt_{\cG_i^{(t)}}(C) \cdot b_i f_C(x)\mcom
\end{equation*}
and we also define the ``hypergraph-tailed'' polynomial
\begin{equation*}\Psi_b(x) = \sum_{i  = 1}^k \sum_{C \in \cH_i^{(r+1)}}  \wt_{\cH_i^{(r+1)}}(C) \cdot b_i f_C(x) \mper
\end{equation*}
We will omit the subscript $b$ when it is clear from context. We note that in the above definitions, each $f_C$ is the chain polynomial as defined in \cref{def:chainpolynomials}.
\end{definition}

With the above setup in hand, we can now state the main technical lemmas.
\begin{lemma}[Refuting the chain polynomial instances]
\label{lem:chainpolyref}
Let $\ell, d, r$ be parameters such that $d^r \geq n$, $\ell \geq 6 d r / \delta$, and $\ell r = o(n)$.\footnote{Note that this is achievable by setting $\ell = 6 d r / \delta$ and $r = O(\log n/\log d)$.} Furthermore, suppose that $k \geq 1/\delta$. Then, for every $1 \leq t \leq r + 1$, it holds that 
\begin{flalign*}
&\E_{b \gets \Fits^k}[\val(\Phi^{(t)}_b)] \leq 4^t O(\sqrt{k \ell r \log n}) \mcom \\
&\E_{b \gets \Fits^k}\left[ \val(\Psi_b)\right] \leq 4^r \left(\frac{k(r+1)}{\delta}  O(\sqrt{k \ell r \log n})\right)^{1/2} \mper
\end{flalign*}
\end{lemma}
We now use \cref{lem:chainxorref} to finish the proof.
\begin{proof}[Proof of \cref{mthm:nonlin} from \cref{lem:chainxorref}]
By construction of the chain polynomials, i.e., \cref{claim:chaincompleteness}, we have that for every $b \in \Fits^k$ and $x = \Code(b)$, $\Psi_b(x) + \sum_{t = 1}^{r+1} \Phi_b^{(t)}(x) = k$. This is because $b_i = x_i$ (as the code is \emph{systematic}), and by \cref{claim:chaincompleteness}, the ``$b_i$ part'' of the polynomial is equal to $x_i$ when $x$ is a codeword. Therefore, there must exist $t$ such that $\E_b[\val(\Phi_b^{(t)})] \geq k/(r+2)$, or else $\E_b[\val(\Psi_b(x))] \geq k/(r+2)$.

Let us take $r = O(\sqrt{\log n})$, $d = 2^{O(\sqrt{\log n})}$, and $\ell = O(d r/\delta) = \delta^{-1} 2^{O(\sqrt{\log n})} \sqrt{\log n}$. We clearly have that all the conditions of \cref{lem:chainxorref} are satisfied. 

If $\E_b[\val(\Phi_b^{(t)})] \geq k/(r+2)$ for some $t$, then we have
\begin{flalign*}
&\frac{k}{r+2} \leq \E_b[\val(\Phi_b^{(t)})] \leq 4^t O(\sqrt{k \ell r \log n}) \\
&\implies k \leq r^2 4^{O(t)} O(\ell r \log n) \leq r^2 4^{O(r)} O(\ell r \log n) \\
&\leq  2^{O(\sqrt{\log n})} \cdot \frac{r^4 \log n}{\delta} \leq 2^{O(\sqrt{\log n})} \\
&\implies (\log \delta k) \leq O(\sqrt{\log n}) \mcom
\end{flalign*}
or equivalently, $n \geq 2^{\Omega((\log \delta k)^2)} = (\delta k)^{\Omega(\log (\delta k))}$.

Otherwise, if $\E_b[\val(\Psi_b(x))] \geq k/(r+2)$, then we have
\begin{flalign*}
&\frac{k}{r+2} \leq \E_b[\val(\Psi_b(x))] \leq 4^r \left(\frac{k(r+1)}{\delta}  O(\sqrt{k \ell r \log n})\right)^{1/2} \\
&\implies k \leq \frac{r^6}{\delta^2} 2^{O(r)}   O( \ell r \log n) = \frac{1}{\delta^3} 2^{O(\sqrt{\log n})} \\
&\implies \delta^3 k \leq 2^{O(\sqrt{\log n})} \\
&\implies n \geq (\delta^3 k)^{\Omega(\log \delta^3 k)} \mper
\end{flalign*}
This finishes the proof; we note that the additional $\log (1/\delta)$-factor comes from \cref{fact:bgt17}, we loses a factor of $\log(1/\delta)$ in $k$ when one makes the code systematic.
\end{proof}

The remainder of paper is dedicated to proving \cref{lem:adaptivepolys,lem:chainpolyref}. First, we show \cref{lem:adaptivepolys} in \cref{sec:adaptivesmooth}. Then, in \cref{sec:chainxor}, we generalize the chain XOR instances of~\cite{KothariM23} to weighted, directed, and nonuniform hypergraphs, and we show (in \cref{sec:chainxorandchainpoly} that refuting these ``nicer'' polynomials (\cref{lem:chainxorref}) suffices to prove \cref{lem:chainpolyref}. We then break the proof of \cref{lem:chainxorref} across \cref{sec:graphref,sec:decomp,sec:kikuchi,sec:rowpruning}, which will complete the proof of \cref{mthm:nonlin}.

\subsection{Constructing Polynomials from Adaptive Smoothed Decoders}
\label{sec:adaptivesmooth}

In this subsection, we prove \cref{lem:adaptivepolys}. 
Let $\Code \colon \Fits^k \to \Fits^n$ be a systematic $3$-LCC an adaptive decoder. For each $u \in [n]$, we use the decoding algorithm $\Dec(u)$ to weight functions $\wt_{H_u}$ and $\wt_{G_u}$. In what follows, we consider a fixed $u \in [n]$.

First, without loss of generality, we may assume that the decoder $\Dec(u)$ makes \emph{exactly} $3$ queries. We can view the decoder as a decision tree: first, $\Dec(u)$ generates the first query $v_1$ from some distribution. Then, $\Dec(u)$ receives a bit $a_1 \in \Fits$, the answer to the query $v_1$. This answer selects the branch of the decision tree, which determines the distribution of the next query $v_2$. Then, the decoder receives another answer $a_2 \in \Fits$, which selects the branch of the decision tree, and gives the distribution of the final query $v_3$. Finally, the decoder receives an answer $a_3$, and then it computes a (deterministic) function $f_{(v_1, a_1, v_2, a_2)}$ of $a_3$ to produce its output. This function must be deterministic as it must always output $x_u$, by perfect completeness.\footnote{We note that if the function is not deterministic then it is simply a convex combination of deterministic functions, and we can also handle this case. See \cref{app:imperfectcompleteness}.}
We note that there are exactly $4$ valid deterministic functions: $1$, $-1$, $a_3$, and $-a_3$, so $f_{(v_1, a_1, v_2, a_2)}$ must be one of these.

 For each choice of $C = (v_1, a_1, v_2, a_2, v_3) \in ([n] \times \Fits)^2 \times [n]$, we let $\wt_u(C)$ be the probability that the decoder makes the set of queries $C$ (with the appropriate answers) when given oracle access to \emph{any} $x$ that is consistent with $C$, meaning that $x_{v_1} = a_1$ and $x_{v_2} = a_2$. Indeed, this does not depend on the choice of $x$, as there is some probability $p_{v_1}$ that the decoder queries $v_1$ (which does not depend on $x$), and then given $x_{v_1} = a_1$, there is a probability $p_{v_2}$ that the decoder queries $v_2$, etc.

We now partition the query sets into two types. If $C$ is such that $f_{(v_1, a_1, v_2, a_2)}$ is a constant function $\sigma \in \Fits$ (so it does not depend on $a_3$), then we set $\wt_{G_u}(v_1, a_1, v_2, a_2) = \wt_u(C)$ and $\sigma_{(v_1, a_1, v_2, a_2)} = \sigma$. Otherwise, we have that $C$ is such that $f_{(v_1, a_1, v_2, a_2)}= \sigma a_3$, and then we set $\wt_{H_u}(v_1, a_1, v_2, a_2, v_3) = \wt_u(C)$ and $\sigma_{(v_1, a_1, v_2, a_2, v_3)} = \sigma$.

We now show that this weight function has the desired properties. Indeed, we have essentially encoded the behavior of the arbitrary decoder as this system of polynomials. 

First, let us show that
\begin{flalign*}
&\sum_{C = (v_1, a_1, v_2, a_2)} \left(\wt_{G_u}(C) + \sum_{v_3 \in [n]}\wt_{H_u}(C, v_3)\right) = 4 \mper
\end{flalign*}
Consider the decoder $\Dec'(u)$ that simulates $\Dec_u$ by generating random bits as the answers to the queries of $\Dec(u)$. It follows that the probability that $\Dec'(u)$ queries a particular $C$ is $\wt(C)/4$, and hence \cref{eq:wtbound} holds.

Next, let us show that for any $x \in \Code$
\begin{flalign*}
&\sum_{C = (v_1, a_1, v_2, a_2)} \left(\wt_{G_u}(C) + \sum_{v_3 \in [n]}\wt_{H_u}(C, v_3)\right) \cdot \AND(a_1 x_{v_1}, a_2 x_{v_2}) = 1 \mper 
\end{flalign*}
Indeed, we observe that for any $x \in \Code$ and any $C$, $\wt_{H_u}(C, v_3)\cdot \AND(a_1 x_{v_1}, a_2 x_{v_2})$ is $0$ if $C$ is inconsistent with $x$, and otherwise it is the probability that $\Dec^{x}(u)$ queries $C$, and the same statement holds for $\wt_{G_u}(C)\AND(a_1 x_{v_1}, a_2 x_{v_2})$. Hence, the sum must be $1$.

Finally, we have
\begin{flalign*}
&x_u \sum_{C = (v_1, a_1, v_2, a_2)} \left(\wt_{G_u}(C) \sigma_{(u,C)} + \sum_{v_3 \in [n]}\wt_{H_u}(C, v_3) \sigma_{(u,C, v_3)} x_{v_3}\right) \cdot \AND(a_1 x_{v_1}, a_2 x_{v_2}) = \E[\Dec^{x}(u) x_u] \mper
\end{flalign*}
Indeed, this is because for any $C = (v_1, a_1, v_2, a_2, v_3)$ and any $x \in \Code$, if $\AND(a_1 x_{v_1}, a_2 x_{v_2}) = 1$ then the output of the decoding function (which is $\sigma_{(u,C, v_3)} x_{v_3}$) is equal to $x_u$, by perfect completeness. And, a similar statement holds for $C =(v_1, a_1, v_2, a_2)$ as well. This finishes the proof.

\section{Chain XOR Polynomials and the Main Technical Lemma}
\label{sec:chainxor}
In this section, we will introduce an abstract notion of chains that produces a polynomial that we call a ``chain XOR instance'' and state a technical lemma (\cref{lem:chainxorref}) that bounds the value of such instances. Then, in \cref{sec:chainxorandchainpoly} we show that this technical lemma implies \cref{lem:chainpolyref}. This notion of chain XOR instances is a generalization of chain XOR derivations constructed in~\cite{KothariM23}. The notions here handle the case of weighted and nonuniform hypergraphs.

We begin by defining a ($\delta$-smoothed) $3$-LCC hypergraph collection.

\begin{definition}[$3$-LCC hypergraph collection]
A $3$-LCC hypergraph collection on $[n]$ vertices is a collection of pairs $(H_u, G_u)$, one for each $u \in [n]$, where $G_u$ is a (weighted and directed) $2$-uniform  hypergraph and $H_u$ is a (weighted and directed) $3$-uniform hypergraph\footnote{Note that \cref{def:hypergraph} requires that each tuple with nonzero weight has \emph{distinct} vertices.} such that for every $u \in [n]$, $\sum_{C \in [n]^2} \wt_{G_u}(C) + \sum_{C \in [n]^3} \wt_{H_u}(C) = 1$.

We furthermore say that the hypergraph collection is $\delta$-smooth if for every $u, v \in [n]$, $\sum_{C \in [n]^2 : v \in C} \wt_{G_u}(C) + \sum_{C \in [n]^3 : v \in C} \wt_{H_u}(C) \leq \frac{1}{\delta n}$
\end{definition}

We now define the $t$-chain hypergraphs.
\begin{definition}[$t$-chain hypergraph $\cH_u^{(t)}$]
\label{def:tchainhypergraph}
Let $t \geq 1$ be an integer, and let $(G_u, H_u)_{u \in [n]}$ denote a $3$-LCC hypergraph collection. For any $u \in [n]$, let $\cH_u^{(t)}$ denote
the weight function $\wt_{\cH_u^{(t)}} \colon [n]^{3t+1} \to \R_{\geq 0}$, i.e., from length $3t+1$ tuples of the form $C = (u_0,v_1, v_2, u_1, v_3, v_4, u_2, \dots, u_{t-1}, v_{2(t-1) + 1}, v_{2(t-1) + 2}, u_t)$ to $\R_{\geq 0}$, where $\wt_{\cH_u^{(t)}}(C)= 0$ if $u_0 \ne u$, and otherwise:
\begin{equation*}
\wt_{\cH_u^{(t)}}(C)= \prod_{h = 0}^{t-1} \wt_{H_{u_h}}(v_{2h + 1}, v_{2h + 2}, u_{h+1}) \mper
\end{equation*}
For a $t$-chain $C$, we call $u_0$ the head, the $u_h$'s the \emph{pivots} for $1 \leq h \leq t-1$, and $u_t$ the \emph{tail} of the chain $C$. The monomial associated to $C$, which we denote by $g_{C}$, is defined to be $x_{u_t} \prod_{h = 0}^{t-1} x_{v_{2h + 1}} x_{v_{2h+2}}$. We call the $t$-chain hypergraph $\cH_u^{(t)}$ ``hypergraph-tailed'', as the last link uses one of the hypergraphs $H_v$.
\end{definition}
We note that for any $u \in [n]$, $\cH_u^{(1)}$ is equivalent to $H_u$, i.e., $\cH_u^{(1)} = \{u\} \times H_u$.

\begin{definition}[$t$-chain hypergraph $\cG_u^{(t)}$]
\label{def:tchaingraph}
Let $t \geq 1$ be an integer, and let $(G_u, H_u)_{u \in [n]}$ denote a $3$-LCC hypergraph collection. For any $u \in [n]$, let $\cG_u^{(t)}$ denote
the weight function $\wt_{\cG_u^{(t)}} \colon [n]^{3t} \to \R_{\geq 0}$, i.e., from length $3t$ tuples of the form $C = (u_0,v_1, v_2, u_1, v_3, v_4, u_2, \dots, u_{t-1}, v_{2(t-1) + 1}, v_{2(t-1) + 2})$ to $\R_{\geq 0}$, where $\wt_{\cG_u^{(t)}}(C)= 0$ if $u_0 \ne u$, and otherwise:
\begin{equation*}
\wt_{\cH_u^{(t)}}(C)= \wt_{G_{u_{t-1}}}(v_{2(t-1) + 1}, v_{2(t-1) + 2}) \cdot \prod_{h = 0}^{t-2} \wt_{H_{u_h}}(v_{2h + 1}, v_{2h + 2}, u_{h+1}) \mper
\end{equation*}
Note that the chains in $\cG^{(t)}$ have no tail vertex $u_t$.
The monomial associated to $C$, which we denote by $x_{C}$, is defined to be $g_C = \prod_{h = 0}^{t-1} x_{v_{2h + 1}} x_{v_{2h+2}}$. We call the $t$-chain hypergraph $\cG_u^{(t)}$ ``graph-tailed'', as the last link uses one of the graphs $G_v$.
\end{definition}
We note that for any $u \in [n]$, $\cG_u^{(1)}$ is equivalent to $G_u$, i.e., $\cG_u^{(1)} = \{u\} \times G_u$.

We now make the following observation.
\begin{observation}
\label{obs:weightpreserved2}
Let $(G_u, H_u)_{u \in [n]}$ denote a $3$-LCC hypergraph collection. Then, for any $t \geq 1$ and $u \in [n]$, it holds that $\sum_{C \in [n]^{3t + 1}} \wt_{\cH_u^{(t)}}(C) + \sum_{t' = 1}^{t} \sum_{C \in [n]^{3t'}} \wt_{\cG_u^{(t')}}(C) = 1$.
\end{observation}
\begin{proof}
This follows by induction. The base case of $t = 1$ is simple, as by definition we have 
\begin{flalign*}
&\sum_{C \in [n]^4} \wt_{\cH_u^{(1)}}(C) + \sum_{C \in [n]^3} \wt_{\cG_u^{(1)}}(C) = \sum_{(u,C) \in [n]^4} \wt_{\cH_u^{(1)}}(C) + \sum_{(u,C) \in [n]^3} \wt_{\cG_u^{(1)}}(C) = \sum_{C\in [n]^3} \wt_{H_u}(C) + \sum_{C \in [n]^2} \wt_{G_u}(C) \mper
\end{flalign*}
We now show the induction step. Let $C \in [n]^{3t + 1}$ have tail $u_t$. Let $S_1$ denote the set of tuples in $[n]^{3t + 3}$ that extend $C$, i.e., the first $3t+1$ coordinates are $C$, and similarly let $S_2$ denote the set of tuples in $[n]^{3t + 4}$ that extend $C$. We observe that $S_1 = C \times [n]^2$ and $S_2 = C \times [n]^3$. Moreover, we have
\begin{flalign*}
&\sum_{C' \in S_1} \wt_{\cG_u^{(t+1)}}(C') + \sum_{C' \in S_2} \wt_{\cH_u^{(t+1)}}(C') = \sum_{C' \in [n]^2} \wt_{\cH_u^{(t)}}(C)\wt_{G_{u_t}}(C') + \sum_{C' \in [n]^3} \wt_{\cH_u^{(t)}}(C)\wt_{H_{u_t}}(C') =  \wt_{\cH_u^{(t)}}(C) \mper
\end{flalign*}
Hence, it follows that 
\begin{flalign*}
&\sum_{C \in [n]^{3t + 4}} \wt_{\cH_u^{(t+1)}}(C) + \sum_{t' = 1}^{t+1} \sum_{C \in [n]^{3t'}} \wt_{\cG_u^{(t')}}(C) = \sum_{C \in [n]^{3t + 4}} \wt_{\cH_u^{(t+1)}}(C) + \sum_{C \in [n]^{3t+3}} \wt_{\cG_u^{(t')}}(C) + \sum_{t' = 1}^{t} \sum_{C \in [n]^{3t'}} \wt_{\cG_u^{(t')}}(C) \\
&=\sum_{C \in [n]^{3t + 4}} \wt_{\cH_u^{(t+1)}}(C) + \sum_{C \in [n]^{3t+3}} \wt_{\cG_u^{(t+1)}}(C) + \sum_{t' = 1}^{t} \sum_{C \in [n]^{3t'}} \wt_{\cG_u^{(t')}}(C) = \sum_{C \in [n]^{3t + 1}} \wt_{\cH_u^{(t)}}(C) + \sum_{t' = 1}^{t} \sum_{C \in [n]^{3t'}} \wt_{\cG_u^{(t')}}(C) = 1 \mcom
\end{flalign*}
where the last step is by the induction hypothesis.
\end{proof}

We are now ready to define the chain XOR instances.
\begin{definition}[Chain XOR instance]
Let $(G_u, H_u)_{u \in [n]}$ denote a $3$-LCC hypergraph collection. Let $k \leq n$ and $r \geq 1$ be an integer. For each $1 \leq t \leq r$, we define the ``graph-tailed'' polynomial
\begin{equation*}\Phi_b^{(t)}(x) = \sum_{i \in K} \sum_{C \in [n]^{3t}}  \wt_{\cG_i^{(t)}}(C) \cdot b_i g_C \mcom
\end{equation*}
and we also define the ``hypergraph-tailed'' polynomial
\begin{equation*}\Psi_b(x) = \sum_{i \in K} \sum_{C \in [n]^{3t+1}}  \wt_{\cH_i^{(r)}}(C) \cdot b_i g_C \mper
\end{equation*}
We will omit the subscript $b$ when it is clear from context. We note that in the above definitions, each $g_C$ is the monomial associated with the chain $C$, as defined in \cref{def:tchainhypergraph,def:tchaingraph}.
\end{definition}

With the above setup in hand, we can now state the main technical lemma.
\begin{lemma}[Refuting the chain XOR instances]
\label{lem:chainxorref}
Let $(G_u, H_u)_{u \in [n]}$ denote a $\delta$-smooth $3$-LCC hypergraph collection and let $k \leq n$. Let $\ell, d, r$ be parameters such that $d^r \geq n$, $\ell \geq 6 d r / \delta$, and $\ell r = o(n)$. Furthermore, suppose that $k \geq 1/\delta$. Then, for each $1 \leq t \leq r + 1$, it holds that 
\begin{flalign*}
&\E_{b \gets \Fits^k}[\val(\Phi^{(t)}_b)] \leq O(\sqrt{k \ell r \log n}) \mcom \\
&\E_{b \gets \Fits^k}\left[\val(\Psi_b)\right] \leq \left(\frac{k(r+1)}{\delta}  O(\sqrt{k \ell r \log n})\right)^{1/2} \mper
\end{flalign*}
\end{lemma}
The proof of \cref{lem:chainxorref} has two steps. First, in \cref{sec:graphref}, we refute the graph-tailed instances. Then, in \cref{sec:decomp,sec:kikuchi,sec:rowpruning}, we refute the hypergraph-tailed instances.

As we shall show in \cref{sec:chainxorandchainpoly}, \cref{lem:chainxorref} implies \cref{lem:chainpolyref}. For now, we devote the rest of this section to establishing some shared terminology which will be useful in the later sections.

\parhead{Chains that fix some positions.} We will often refer to the set of chains where some of the links, i.e., pairs $(v_{2h+1}, v_{2h+2})$ are forced to contain some $v \in [n]$. Towards this, we introduce the following terminology. 
\begin{definition}[Chains containing $Q$]
Let $t, r$ be integers with $t \leq r$. For any $Q = (Q_1, \dots, Q_t, Q_{t+1}) \in \{[n] \cup \star\}^{t+1}$, we say that a length $3r+1$ tuple $C = (u_0,v_1, v_2, u_1, v_3, v_4, u_2, \dots, u_{t-1}, v_{2(r-1) + 1}, v_{2(r-1) + 2}, u_r)$ contains $Q$, denoted by $Q \subseteq C$, if $Q_{t+1} \in \{\star, u_r\}$ and for $1 \leq h \leq t$, if $Q_h \neq \star$, then either $Q_h = v_{2(r - 1 - t + h) + 1}$ or $Q_h = v_{2(r - 1 - t + h) + 2}$.

We say that a $Q$ is \emph{contiguous} if there exists $s \leq t$ such that $Q_{h} \neq \star$ for every $h \geq s+1$ and $Q_{h} = \star$ for every $1 \leq h \leq s$, i.e., the first $s$ entries are $\star$, and the remaining entries are non-$\star$. We note that by definition, $Q_{t+1} \ne \star$ always.

We say that $Q$ is \emph{complete} if $Q$ does not contain any $\star$. We say that $Q' \supseteq Q$ if whenever $Q_h \neq \star$, $Q'_h = Q_h$. We define the size $\abs{Q}$ to be the number of coordinates in $Q$ that do not equal $\star$.
\end{definition}

\subsection{Relating the chain polynomials and chain XOR instances}
\label{sec:chainxorandchainpoly}
Recall that the chain polynomials $\Phi_b^{(t)}$ and $\Psi_b$ in \cref{sec:chainpolys} are products of $\AND$ functions, which means that (1) they are inhomogeneous polynomials, and (2) some of the coefficients can be negative. This is contrast to the chain XOR instances produced in \cref{sec:chainxor}, which are homogeneous and with positive coefficients, and this is very helpful in the proof of \cref{lem:chainxorref}. 

The goal of this section is to show that, given the output of \cref{lem:adaptivepolys}, we can construct a $3$-LCC hypergraph collection $(H'_u, G'_u)_{u \in [n]}$ such that the chain XOR instances ${\Phi'}_b^{(t)}$ and ${\Psi'}_b$ produced from $(H'_u, G'_u)$ are (up to a scaling factor) equivalent to the chain polynomial instances $\Phi_b^{(t)}$ and $\Psi_b$ from \cref{sec:chainpolys}.

First, we explain how to convert the polynomials $\Phi_b^{(t)}$ and $\Psi_b$ into equivalent homogeneous polynomials $\tilde{\Phi}_b^{(t)}$ and $\tilde{\Psi}_b$ over a larger set of $4n$ variables. In particular, these new polynomials will have the following properties (1) $\val(\tilde{\Phi}_b^{(t)}) \geq \val(\Phi_b^{(t)})$ and $\val(\tilde{\Psi}_b) \geq \val(\Psi_b)$, (2) $\tilde{\Psi}_b$ is a degree $2(r+1) + 1$ homogeneous polynomial and $\tilde{\Phi}_b^{(t)}$ is a degree $2t$ homogeneous polynomial. Then, we will construct a $3$-LCC hypergraph collection $(H'_u, G'_u)_{u \in [4n]}$, and show that the chain XOR instances ${\Phi'}_b^{(t)}$ and ${\Psi'}_b$ produced from this collection are equal to $4^{-t} \Phi_b^{(t)}$ and $4^{-r} \tilde{\Psi}_b$.

\parhead{Defining the homogeneous polynomials.} This transformation to produce $\Phi_b^{(t)}$ and $\tilde{\Psi}_b$ is straightforward. First, we define a map $\pi \colon \Fits^n$ to $\Fits^{4n}$ as follows. For each $x \in \Fits^n$ we define $y = \pi(x)$ by adding, for each $v \in [n]$, the $4$ bits $x_v$, $-x_v$, $1$ and $-1$ to $y$. We refer to these bits as $+v, -v, 1^{(v)}, -1^{(v)}$, i.e., $y_{+v} = x_v$, $y_{-v} = -x_v$, $y_{1^{(v)}} = 1$, and $y_{-1^{(v)}} = -1$. We think of $1^{(v)}$ as the $v$-th copy of $1$, and similarly $-1^{(v)}$ is the $v$-th copy of $-1$.

Now, we transform the polynomials $\Phi_b^{(t)}$ and $\Psi_b$. Each term that contains a function $\AND(a_1 x_{v_1}, a_2 x_{v_2}) \cdot \sigma$ for $\sigma \in \Fits$ is replaced by the $8$ terms 
\begin{equation*}
\frac{1}{8} y_{\sigma^{(v_1)}} y_{1^{(v_2)}} + \frac{1}{8} y_{\sigma a_1 v_1} y_{1^{(v_2)}} + \frac{1}{8} y_{\sigma^{(v_1)}} y_{a_2 v_2}  + \frac{1}{4} y_{\sigma a_1 v_1} y_{a_2 v_2} + \frac{1}{8} y_{1^{(v_1)}} y_{\sigma^{(v_2)}} + \frac{1}{8} y_{a_1 v_1} y_{\sigma^{(v_2)}} + \frac{1}{8} y_{1^{(v_1)}} y_{\sigma a_2 v_2}  + \frac{1}{8} y_{ a_1 v_1} y_{\sigma a_2 v_2} \mcom
\end{equation*}
where, e.g., $y_{\sigma a_1 v_1}$ is $y_{+v}$ if $\sigma a_1 = 1$ and $y_{-v}$ if $a_1 = -1$, and $y_{\sigma^{(v)}}$ is either $y_{1^{(v)}}$ if $\sigma = 1$ or $y_{-1^{(v)}}$ if $\sigma = -1$. By construction, if $y = \pi(x)$, then $\AND(a_1 x_{v_1}, a_2 x_{v_2}) \cdot \sigma$ is equal to this new polynomial, and this polynomial is a homogeneous degree $2$ polynomial in $y$ with nonnegative coefficients. 
Furthermore, the coefficients of the new polynomial all sum to $1$.

\parhead{Defining the homogeneous polynomials via XOR chains on hypergraphs.}
In order to refute the homogeneous polynomials using \cref{lem:chainxorref}, we will need to write them as ``chain XOR instances'' on a certain set of hypergraphs. Now, because the coefficients of the new $\AND$ polynomials all sum to $1$, we can essentially replace each hyperedge $C = (v_1, a_1, v_2, a_2, v_3)$, e.g., with $8$ new hyperedges each of weight $1/8 \wt(C)$. This defines, for each $u \in [n]$, a pair $(H'_u, G'_u)$ of weighted $3$-uniform and $2$-uniform hypergraphs.

We can then form the chain XOR instances ${\Phi'}_b^{(t)}$ and ${\Psi'}_b$ from this hypergraph collection. It will be fairly immediate to observe that the resulting polynomials produced via this process are the same as $\tilde{\Phi}_b^{(t)}$ and $\tilde{\Psi}_b$ up to a scaling factor -- in other words, the operations of ``form chains'' and ''add extra variables'' commute. As a result, the ``chain XOR instances'' ${\Phi'}_b^{(t)}$ and ${\Psi'}_b$ we get from this process are equal to the polynomials $4^{-t} \tilde{\Phi}_b^{(t)}$ and $4^{-r} \tilde{\Psi}_b$ that we have just defined, and thus we can refute them using \cref{lem:chainxorref}, which will imply \cref{lem:chainpolyref}.

In the remainder of this section, we define the pairs $(H'_u, G'_u)$ of weighted $3$- and $2$-uniform hypergraphs. Then, we will finally observe that these XOR instances are equivalent to the polynomials $\tilde{\Phi}_b^{(t)}$ and $\tilde{\Psi}_b$.

\begin{definition}[The pairs $(H'_u, G'_u)$]
Let $u \in [n]$ and let $n' \defeq 4n$. We identify $4n$ with the $4n$ vertices $+v, -v, 1^{(v)}$, and $-1^{(v)}$.

We define the weight function $\wt_{H'_u}$ and $\wt_{G'_u}$ as follows. For each $C = (v_1, a_1, v_2, a_2, v_3)$ with bit $\sigma_C \in \Fits$, we set 
\begin{enumerate}[(1)]
\item $\wt_{H'_u}(\sigma_C a_1 v_1, a_2 v_2) = \frac{1}{8} \wt_{H_u}(C) \cdot \frac{1}{4}$,
\item $\wt_{H'_u}(\sigma_C^{(v_1)}, a_2 v_2) = \frac{1}{8} \wt_{H_u}(C) \cdot \frac{1}{4}$,
\item $\wt_{H'_u}(\sigma_C a_1 v_1, 1^{(v_2)}) = \frac{1}{8} \wt_{H_u}(C)\cdot \frac{1}{4}$,
\item $\wt_{H'_u}(\sigma_C^{(v_1)}, 1^{(v_2)}) = \frac{1}{8} \wt_{H_u}(C)\cdot \frac{1}{4}$,
\item $\wt_{H'_u}( a_1 v_1, \sigma_C a_2 v_2) = \frac{1}{8} \wt_{H_u}(C)\cdot \frac{1}{4}$,
\item $\wt_{H'_u}(1^{(v_1)}, \sigma_C a_2 v_2) = \frac{1}{8} \wt_{H_u}(C)\cdot \frac{1}{4}$,
\item $\wt_{H'_u}( a_1 v_1, \sigma_C^{(v_2)}) = \frac{1}{8} \wt_{H_u}(C)\cdot \frac{1}{4}$,
\item $\wt_{H'_u}(1^{(v_1)}, \sigma_C^{(v_2)}) = \frac{1}{8} \wt_{H_u}(C)\cdot \frac{1}{4}$,
\end{enumerate}
and we do the analogous transformation to define $\wt_{G'_u}$ from $\wt_{G_u}$. 

Furthermore, if the original decoder $\Dec(\cdot)$ was $\delta$-smooth, then for any new vertex $v' \in [n']$, it holds that
\begin{flalign*}
\sum_{C' = (v'_1, v'_2, v'_3) : v' \in C'} \wt_{H'_u}(C') + \sum_{C' = (v'_1, v'_2) : v' \in C'} \wt_{G'_u}(C') \leq \frac{4}{\delta n'} \mper
\end{flalign*}
\end{definition}

\begin{remark}
So far, we have only defined a pair $(H'_u,G'_u)$ for the original vertices $u$, not the new vertices $u' \in [n']$, so technically we have not defined a full hypergraph collection. However, we can easily define equivalent hypergraphs for all the new vertices, but this turns out to be unnecessary as the only hyperedges in $H'_u$ with nonzero weight have $C' = (v'_1, v'_2, v_3)$ where $v_3 \in [n]$ is one of the original vertices, and so the chain XOR instances formed will never use the hypergraphs $(H'_{u'}, G'_{u'})$ if $u'$ is a ``new vertex''. So, we do not need to define $H'_{u'}$ and $G'_{u'}$ where $u'$ is a new variable. This is also the reason for the upper bound of $\frac{16}{\delta n'}$ instead of $\frac{4}{\delta n'}$ -- a fixed third vertex $v_3 \in [n]$ could have $\frac{4}{\delta n}$-fraction of the weight in the original decoder, which is now $\frac{16}{\delta n'} \cdot \frac{1}{4}$ (as we scale down all weights by $1/4$).
\end{remark}

This now leads us to the following key observation.
\begin{observation}
Let $\tilde{\Phi}_b^{(t)}$ and $\tilde{\Psi}_b$ be the polynomials defined via the transformation to ${\Phi}_b^{(t)}$ and ${\Psi}_b$, and let ${\Phi'}_b^{(t)}$ and ${\Psi'}_b$ be the chain XOR instances of the $3$-LCC hypergraph collection $(H'_u, G'_u)_{u \in [n]}$. Then, $\tilde{\Phi}_b^{(t)} = 4^{-t} {\Phi'}_b^{(t)}$ and $\tilde{\Psi}_b = 4^{-r} \Psi'_b$.
\end{observation}
This observation follows immediately from the definitions. Namely, for each $C = (v_1, a_1, v_2, a_2, v_3)$ in the original $H_u$, we have now added $8$ different constraints of weight $1/32$ times the original weight of $C$, such that the XOR instance on the new constraints is equal to the $\AND$ polynomial of the previous constraints.

The above observation immediately shows that \cref{lem:chainxorref} implies \cref{lem:chainpolyref}, and so it remains to prove \cref{lem:chainxorref}.
}
\section{Refuting the Graph-Tail Instances}
\label{sec:graphref}
In this section, we prove the first equation of \cref{lem:chainxorref}. Let $r \geq 1$ and let $1 \leq t \leq r + 1$ be fixed.
We begin by defining the Kikuchi matrices.

\begin{definition}
\label{def:graphkikuchi}
Let $r \geq 1$ and $1 \leq t \leq r + 1$. Let $i \in [k]$. For a tuple $C = (i, v_1, v_2, u_1, v_3, v_4, \dots, v_{2(t-1) + 1}, v_{2(t-1) + 2}) \in [n]^{3t}$, we define the matrix $A_{i}^{(C)} \in \Bits^{N}$ where $N = {n \choose \ell}^{t}$, to be the matrix indexed by tuples of sets $\vec{S} = (S_0, \dots, S_{t-1})$, where $A_{i}^{(C)}((S_0, \dots, S_{t-1}), (T_0, \dots, T_{t-1})) = 1$ if for all $h = 0, \dots, t-1$, $S_h \oplus T_h = \{v_{2h + 1}, v_{2h + 2}\}$ with $v_{2h+1} \in S_h, v_{2h + 2} \in T_h$. If this does not hold, then the entry of the matrix is $0$.

We let $A_{i} = \frac{1}{D_t} \sum_{C \in [n]^{3t}} \wt_{\cG_{i}^{(t)}}(C) A_{i}^{(C)}$ and $A = \sum_{i = 1}^k b_i A_i$. Here, $D_t = {n - 2 \choose \ell - 1}^{t}$.
\end{definition}

Next, we relate $\Phi^{(t)}(x)$ to a quadratic form on the matrix $A$.

\begin{lemma}
\label{lem:graphcompleteness}
Let $x \in \Fits^n$, and let $x' \in \Fits^{N}$, where $N = {n \choose \ell}^{t}$, denote the vector where the $(S_0, S_1, \dots, S_{t-1})$-th entry of $x'$ is $\prod_{h = 0}^{t-1} x_{S_h}$.
Let $i \in [k]$ and $t \in \{0, \dots, r\}$. Then, for any $C = (i, v_1, v_2, u_1, v_3, v_4, \dots, v_{2(t-1) + 1}, v_{2(t-1) + 2}) \in [n]^{3t}$, it holds that
\begin{flalign*}
{x'}^{\top} A_{i}^{(C)} x' = D_t \prod_{h = 0}^{t-1} x_{v_{2h + 1}} x_{v_{2h + 2}} \mcom
\end{flalign*}
 i.e., the product of the monomials associated to $C$, where $D_t = {n - 2 \choose \ell - 1}^{t}$. Moreover, for any matrix $B_i^{(C)}$ obtained by ``zeroing out'' exactly $\alpha D_t$ entries of $A_{i}^{(C)}$, the equality holds with a factor of $1 - \alpha$ on the right.
 
 In particular, ${x'}^{\top} A x' = \Phi^{(t)}(x)$.
\end{lemma}
\begin{proof}
Let  $\vec{S} = (S_0, S_1, \dots, S_{t-1})$ and $\vec{T} = (T_0,  \dots, T_{t-1})$ be such that $A_{i}^{(C)}(\vec{S},\vec{T}) = 1$. 
Then, we have that
\begin{flalign*}
x'_{\vec{S}} x'_{\vec{T}} &= \prod_{h = 0}^{t-1} x_{S_h} x_{T_h}  =  \prod_{h = 0}^{t - 1} x_{S_h \oplus T_h} = \prod_{h = 0}^{t-1} x_{v_{2h + 1}} x_{v_{2h + 2}} \enspace,
\end{flalign*}
which is equal to the product of monomials on the right-hand side of the equation we wish to show.

It thus remains to argue that $A_{i}^{({C})}$ has exactly $D_t$ nonzero entries. We observe that, for each $h = 0, \dots, t-1$, there are exactly ${n - 2 \choose \ell - 1}$ pairs $(S_h, T_h)$ such that $S_h \oplus T_h = C_h$ with $v_{2h + 1} \in S_h$ and $v_{2h + 2} \in T_h$. Indeed, this is because by \cref{def:hypergraph}, these vertices must be distinct, and then we must simply choose a set of size $\ell - 1$ that does not contain either of $v_{2h + 1}$ and $v_{2h + 2}$ and this determines $S_h$ and $T_h$. Thus, $D_t =  {n - 2 \choose \ell - 1}^{t}$, as required.
\end{proof}

We would like to now apply matrix Khintchine (\cref{fact:matrixkhintchine}) to bound $\E_b[\norm{A}_2]$ and thus bound $\E_b[\val(\Phi^{(t)}_b(x))]$. However, to do this, we need good bounds on the $\norm{A_i}_2$ of the individual matrices $A_i$. It turns out that the bounds we require for this approach to work are false, but one can find a submatrix $B_i$ of $A_i$ such that the bounds hold. To argue this, we will need the following first moment bounds.
\begin{lemma}[First and conditional moment bounds]
\label{lem:graphrowpruning}
Fix $r \geq 1$, $1 \leq t \leq r+1$, and $i \in [k]$. Let $A_{i}$ be the Kikuchi matrix defined in \cref{def:graphkikuchi}.

Let $\vec{S} = (S_0, \dots, S_{t-1}) \in {[n] \choose \ell}^{t}$ be a row of the matrix, and let $\deg_{i}(\vec{S})$ denote the $\ell_1$-norm of the $\vec{S}$-th row of $A_{i}$. Then, 
\begin{equation*}
\E_{\vec{S}}[\deg_{i}(\vec{S})] \leq \frac{4}{N} \mcom
\end{equation*}
where $N = {n \choose \ell}^{t}$.

Furthermore, let $C \in [n]^{3t}$ be a chain with head $i$. Let $\cD_{C}$ denote the uniform distribution over rows of $A_{i}^{(C)}$ that contain a nonzero entry. Then, it holds that
\begin{equation*}
\E_{\vec{S} \sim \cD_{C}}[\deg_{i}(\vec{S})] \leq \left(1 + \frac{O(\ell r)}{n} \right) \cdot \frac{16}{N} \mper
\end{equation*}

Finally, the same bounds hold for the columns of the matrix.
\end{lemma}

With \cref{lem:graphrowpruning}, we can now do the following. Let $\Gamma$ be a sufficiently large constant, let $\cB_1 = \{\vec{S} : \deg_i(\vec{S}) \geq \Gamma/N\}$ be the set of rows with $\ell_1$-norm at least $\Gamma/N$, and similarly let $\cB_2$ be defined for the columns. We observe that by the conditional moment bounds in \cref{lem:graphrowpruning} and Markov's inequality, each $A_{i}^{(C)}$ has at least $1 - O(1/\Gamma)$-fraction of its nonzero rows not in $\cB_1$, and similarly for columns and $\cB_2$. It thus follows that after setting all the rows in $\cB_1$ and columns in $\cB_2$ to $0$, the resulting matrix still has at least $1 - O(1/\Gamma)$-fraction of its original nonzero entries. By taking $\Gamma$ large enough, we can ensure that this fraction is at least $1/2$. Now, we let $B_i^{(C)}$ be the matrix where we have deleted all rows in $\cB_1$ and columns in $\cB_2$ from $A_i^{(C)}$, and we have additionally set more entries to $0$ so that $B_i^{(C)}$ has \emph{exactly} $D_t/2$ nonzero entries, where $t$ is such that $C \in [n]^{3t}$.

Let us define: $B_{i} =  \frac{1}{D_t} \sum_{C \in [n]^{3t}} \wt_{\cG_{i}^{(t)}}(C) B_{i}^{(C)}$ and $B = \sum_{i = 1}^k b_i B_i$. By \cref{lem:graphcompleteness} (and the ``moreover'' part), we have that for every $x \in \Fits^n$, there exists $x' \in \Fits^N$ such that ${x'}^{\top} B x' = \frac{1}{2} \Phi^{(t)}(x)$. By construction, we have that $\norm{B_i}_2 \leq \Gamma/N$, as this is an upper bound on the $\ell_1$-norm of any row/column in $B_i$.

Thus, applying matrix Khintchine (\cref{fact:matrixkhintchine}), we obtain
\begin{flalign*}
\E_b[\val(\Phi^{(t)}_b)] \leq \E_b[N \norm{B}_2] \leq N \cdot \frac{\Gamma}{N} O(\sqrt{k \log N}) = O(\sqrt{k \ell r \log n}) \mcom
\end{flalign*}
where we use that $\Gamma$ is constant. This finishes the proof of the first equation in \cref{lem:chainxorref}, up to the proof of \cref{lem:graphrowpruning}.

\begin{proof}[Proof of \cref{lem:graphrowpruning}]
We will only prove the statement for the rows. One can observe from the proof that it will immediately hold for the columns also.

We begin by estimating the first moment, i.e., $\E_{\vec{S}}[\deg_{i}(\vec{S})]$. By definition, we have that 
\begin{flalign*}
&\E_{\vec{S}}[\deg_{i}(\vec{S})] = \frac{1}{N} \frac{1}{D_t} \sum_{C \in [n]^{3t}} \wt_{\cG^{(t)}_i}(C) \cdot D_t \leq \frac{4}{N} \mcom
\end{flalign*}
as the sum of the weights of all chains is at most $4$ by \cref{obs:weightpreserved}.

We now fix $t \in \{1, \dots, r+1\}$, $C \in [n]^{3t}$ with head $i$. Let $\cD_{C}$ denote the uniform distribution over rows of $A_{i}^{(C)}$ that contain a nonzero entry. We compute the conditional expectation as follows. First, we shall bound, for $C' \in [n]^{3t}$ with head $i$, the number of rows $\vec{S}$ such that $A_i^{(C)}$ and $A_{i}^{(C')}$ both have a nonzero entry in the $\vec{S}$-th row, \emph{normalized} by the scaling factor $1/D_{t}$. This quantity will depend on some parameter $z$, which is the number of ``shared vertices'' between $C$ and $C'$. Then, we will bound, for each $z$, the total weight of all $C' \in [n]^{3t}$ that has at least $z$ ``shared vertices'' with $C$.

\parhead{Step 1: bounding the normalized number of entries for a fixed $C'$.} To begin, we define the number of ``shared vertices'' between two pairs of chains $C$ and $C'$.
\begin{definition}[Left vertices]
Let $C \in [n]^{3t}$. The tuple of \emph{left vertices} of $C$ is the sequence $L(C) = (v_1, v_3, v_5 ,\dots, v_{2(t-1) + 1})$. We note that if $\vec{S}$ is a row such that $A_{i}^{(C)}$ has nonzero entry in the $\vec{S}$-th row, then $v_{2h + 1} \in S_h$ for $h = 0, \dots, t-1$.
\end{definition}

\begin{definition}[Intersection patterns]
Let $C \in [n]^{3t}$ and $C' \in [n]^{3t}$.

The \emph{intersection pattern} of $C$ with $C'$, given by $Z \in \Bits^{t}$, is defined as $Z_h = 1$ if $L(C)_h = L(C')_h$, and it is $0$ otherwise.
\end{definition}

We now fix $C' \in [n]^{3t}$ and count the number of rows as a function of the intersection pattern $Z$. 
We observe that in order for a row $\vec{S}$ to have a nonzero entry for both pairs of chains, we must have $\{L(C)_h, L(C')_h\} \subseteq S_{h-1}$ for all $h = 1, \dots, t$.

We observe that for each intersection point, i.e., an $h$ such that $L(C)_h =  L(C')_h$, there are ${n \choose \ell - 1}$ choices for the corresponding set, as it needs to only contain one vertex. For each nonintersection point, i.e., an $h \in \{1, \dots, t\}$ where $L(C)_h \ne  L(C')_h$, we have ${n \choose \ell - 2}$ choices, because the set needs to contain both vertices.  In total, we have ${n \choose \ell - 1}^z {n \choose \ell - 2}^{t - z}$.

Now, this implies an upper bound of ${n \choose \ell - 1}^z {n \choose \ell - 2}^{t - z}/ D_{t}$ on the normalized number of entries, which we can compute as
\begin{flalign*}
&{n \choose \ell - 1}^z {n \choose \ell - 2}^{t - z}  / D_{t} = \frac{{n \choose \ell - 1}^z {n \choose \ell - 2}^{t - z} }{{n - 2 \choose \ell - 1}^{t} } = 2 \left(\frac{ {n \choose \ell - 2}} { {n \choose \ell - 1} }\right)^{t - z} \cdot \left( \frac{ {n \choose \ell - 1} }{ {n - 2 \choose \ell - 1}} \right)^{t} \\
&\leq  \left(\frac{\ell - 1} { n - \ell + 2 }\right)^{t - z} \cdot \left( \frac{ n(n-1)} {(n - \ell + 1)(n - \ell)} \right)^{t} \leq \left( \frac{\ell}{n} \right)^{t - z} \cdot \left(1 + \frac{O(\ell r)}{n} \right) \mper
\end{flalign*}

\parhead{Step 2: bounding the weight of $C'$ with a fixed intersection pattern $Z$.} Let us fix the intersection pattern $Z$. We observe that this determines a set of $\abs{Z}$ vertices that must be contained in $C'$. We will abuse notation and let $Z \in {[n] \cup \{\star\}}^{t}$ denote this sequence of vertices (with $\star$'s for the unfixed entries). Let $t''$ denote the largest $h \in \{1, \dots, t\}$ for which $Z_{t''} \ne \star$. We then have
\ignore[old calculation]{
\begin{flalign*}
& \sum_{C' \in [n]^{3t} : Z \subseteq C} \wt_{\cG_i^{(t)}}(C) \\
&=  \sum_{C'' \in [n]^{3t''} : Z \subseteq C''} \left(  \wt_{\cG_i^{(t'')}}(C'') + \sum_{C' \in [n]^{3(t - t'')}} \wt_{\cG_i^{(t)}}(C'', C') \right) \\
&=  \sum_{C'' \in [n]^{3t''} : Z \subseteq C''} \left(  \wt_{\cG_i^{(t'')}}(C'') + \sum_{(u,C') \in [n]^{3(t - t'')}} \wt_{\cH_i^{(t'')}}(C'',u) \wt_{\cG_u^{(t - t'')}}(u,C') \right) \\
&=  \sum_{C'' \in [n]^{3t''} : Z \subseteq C''} \left(  \wt_{\cG_i^{(t'')}}(C'') + \sum_{u \in [n]} \wt_{\cH_i^{(t'')}}(C'',u) \sum_{C' \in [n]^{3(t - t'') - 1}}  \wt_{\cG_u^{(t - t'')}}(u,C') \right) \\
&\leq \sum_{C'' \in [n]^{3t''} : Z \subseteq C''} \left(  \wt_{\cG_i^{(t'')}}(C'') + \sum_{u \in [n]} \wt_{\cH_i^{(t'')}}(C'',u) \right) \mper
\end{flalign*}
Above, we use that $\sum_{C' \in [n]^{3(t - t'') - 1}}  \wt_{\cG_u^{(t - t'')}}(u,C') \leq 1$, which follows by \cref{obs:weightpreserved}.
}
\begin{flalign*}
& \sum_{C' \in [n]^{3t} : Z \subseteq C} \wt_{\cG_i^{(t)}}(C) \\
&=  \sum_{C'' \in [n]^{3t''} : Z \subseteq C''} \left(  \sum_{C' \in [n]^{3(t - t'')}} \wt_{\cG_i^{(t)}}(C'', C') \right) \\
&=  \sum_{C'' \in [n]^{3t''} : Z \subseteq C''} \left( \sum_{(u,C') \in [n]^{3(t - t'')}} \wt_{\cH_i^{(t'')}}(C'',u) \wt_{\cG_u^{(t - t'')}}(u,C') \right) \\
&=  \sum_{C'' \in [n]^{3t''} : Z \subseteq C''} \left(  \sum_{u \in [n]} \wt_{\cH_i^{(t'')}}(C'',u) \sum_{C' \in [n]^{3(t - t'') - 1}}  \wt_{\cG_u^{(t - t'')}}(u,C') \right) \\
&\leq 4 \sum_{C'' \in [n]^{3t''} : Z \subseteq C''} \left(  \sum_{u \in [n]} \wt_{\cH_i^{(t'')}}(C'',u) \right) \mper
\end{flalign*}
Above, we use that $\sum_{C' \in [n]^{3(t - t'') - 1}}  \wt_{\cG_u^{(t - t'')}}(u,C') \leq 4$, which follows by \cref{obs:weightpreserved}.

We now clearly have that $\sum_{C'' \in [n]^{3t''} : Z \subseteq C''} \left(  \sum_{u \in [n]} \wt_{\cH_i^{(t'')}}(C'',u) \right)  \leq 4(\delta n)^{-\abs{Z}}$. This follows by $\delta$-smoothness, as when we sum over a link with no fixed vertex, it has weight $1$ (unless it is the last link, where it has weight $4$), and when we sum over a link where $Z_h \ne \star$, by $\delta$-smoothness it must have weight at most $1/\delta n$. We thus have a bound of $4 (\delta n)^{-\abs{Z}}$.

\parhead{Putting it all together.}
By combining steps (1) and (2) (and paying an additional ${t \choose z}$ factor to choose the nonzero entries of $Z$), we thus obtain the final bound of
\begin{flalign*}
&\E_{\vec{S} \sim \cD_{C}}[\deg_{i}(\vec{S})] \leq \frac{4}{D_t} \sum_{z = 0}^{t} {t \choose z} \cdot 2\left(1 + \frac{O(\ell r)}{n}\right) \cdot \left(\frac{\ell}{n}\right)^{t - z} \cdot (\delta n)^{-z} \\
&\leq  \left(1 + \frac{O(\ell r)}{n}\right) \frac{8}{D_t}  \left(\frac{\ell}{n}\right)^{t} \cdot \sum_{z = 0}^{t}  \cdot \left(\frac{t}{\delta \ell}\right)^{z} \\
&\leq \left(1 + \frac{O(\ell r)}{n}\right) \frac{8}{D_t}  \left(\frac{\ell}{n}\right)^{t} \cdot \sum_{z = 0}^{r}  \cdot \left(\frac{r}{\delta \ell}\right)^{z} \\
&\leq \left(1 + \frac{O(\ell r)}{n}\right) \frac{16}{D_t}  \left(\frac{\ell}{n}\right)^{t} \mcom
\end{flalign*}
where we use that $\ell \geq 2r/\delta$.

Finally, we need to compute $D_t/N$. We have
\begin{flalign*}
&\frac{D_t}{N} = \frac{ {n - 2 \choose \ell - 1}^{t} \cdot {n \choose \ell}^{r + 1 - t} }{ {n \choose \ell}^{r+1}} = \left(\frac{ {n - 2 \choose \ell - 1} }{ {n \choose \ell} }\right)^t \\
& \left(\frac{\ell(n - \ell) }{n(n-1) }\right)^t \geq \left(\frac{\ell}{n}\right)^t \left(1 - \frac{O(\ell r)}{n} \right) \mper
\end{flalign*}
Thus, we have
\begin{flalign*}
&\E_{\vec{S} \sim \cD_{C}}[\deg_{i}(\vec{S})] \leq \left(1 + \frac{O(\ell r)}{n}\right) \frac{16}{D_t}  \left(\frac{\ell}{n}\right)^{t} \\
&\leq \left(1 + \frac{O(\ell r)}{n}\right) \frac{16}{N} \mcom
\end{flalign*}
which finishes the proof.
\end{proof}

\section{Smooth Partitions of Chains}
\label{sec:decomp}
In this section, we begin the proof of the second equation in \cref{lem:chainxorref}.

For notation, we let $\cH^{(t)}$ be the union, over $u$, of $\cH_u^{(t)}$, and $\wt_{\cH^{(t)}}(\cdot) = \sum_{u \in [n]} \wt_{\cH_u^{(t)}}(\cdot)$.

\begin{lemma}
Let $t \geq 1$ and $d \geq 1$ be integers. There is a subset $P_t \subseteq [n]^{t+1}$ and disjoint sets $\cT^{(Q)} \subseteq [n]^{3t + 1}$ for $Q \in P_t$ such that \begin{inparaenum}[(1)]\item  $Q \subseteq C$ for each $C \in \cT^{(Q)}$, and \item $\wt(Q) \defeq \sum_{C \in \cT^{(Q)}} \wt_{\cH^{(t)}}(C) \geq n d^t \cdot (\delta n)^{-t - 1}$.\end{inparaenum}

We say $Q$ is \emph{heavy} if $Q \in P_t$. Note that if $Q$ is heavy then $Q$ is contiguous and complete by definition.

Finally, as a trivial case, we let $P_0 = [n]$ and for $Q = (v) \in P_0$, we let $\cT^{(Q)} = (v)$. Here, we let $\wt(Q) = 1$.
\end{lemma}
\begin{proof}
The proof follows by a simple greedy algorithm. Let $S = [n]^{3t+1}$. If there exists $Q$ such that $\sum_{C \in S : Q \subseteq C} \wt_{\cH^{(t)}}(C) \geq n d^t \cdot (\delta n)^{-t - 1}$, then we remove all such $C$ from $S$ and add them to $\cT^{(Q)}$. We repeat until there is no such $Q$ remaining. We note that $Q$ cannot be used twice in this sequence, as when we pick a $Q$ we remove all $C \in S$ containing $Q$.
\end{proof}

\begin{definition}[Partitions of the chains]
\label{def:partition}
Let $r \geq 1$ be an integer. For each $1 \leq t \leq r$ and heavy $Q \in P_t$, we let $\cH^{(r,Q)}$ denote the set of tuples $C \in [n]^{3r + 1}$ where:
\begin{enumerate}
\item $C$ is extends a tuple in $\cT^{(Q)}$ ``backwards'', i.e., $(C_{3(r-t) + 1}, \dots, C_{3r + 1}) \in \cT^{(Q)}$;
\item $Q$ is maximal: for any $t' > t$ and $Q' \in P_{t'}$, $(C_{3(r - t') + 1}, \dots, C_{3r + 1}) \notin \cT^{(Q')}$.
\end{enumerate}
\end{definition}

\begin{observation}
We have that for each $t = 0, \dots, r$, it holds that
$\sum_{Q \in P_t}  \wt(Q) \leq n$, and so $\sum_{t = 0}^r \sum_{Q \in P_t}  \wt(Q) \leq (r+1)n$.
\end{observation}
\begin{proof}
We observe that for any $t = 0, \dots, r$, it holds that
\begin{flalign*}
\sum_{Q \in P_t} \wt(Q) = \sum_{Q \in P_t} \sum_{C \in \cT^{(Q)}}  \wt_{\cH^{(t)}}(C) \leq \sum_{C \in [n]^{3t + 1}}  \wt_{\cH^{(t)}}(C) = n \mper \qedhere
\end{flalign*}
\end{proof}

We note that \cref{def:partition} gives a partition of the $r$-chains, but the polynomial $\Psi(x)$ uses a restricted set of $(r+1)$-chains, namely those that have their head in $[k]$. In the following definition, we use the partition of the $r$-chains to induce a partition of the special $(r+1)$-chains.
\begin{definition}[Induced partition of $\cH^{(r+1)}_i$]
\label{def:inducedpartition}
Let $r \geq 1$ be an integer. For each $0 \leq t \leq r$ and each $Q \in P_t$, we let $\cH^{(r+1,Q)}_i$ denote the set of length $3r+4$ tuples of the form $(i, w_1, w_2, C)$ where $C \in \cH^{(r,Q)}$.
\end{definition}

\begin{definition}[Bipartite XOR formulas from a contiguously regular partition] \label{def:bip-Psi}
Fix integers $r, d \geq 1$. For each $1 \leq t \leq r$ and $Q \in P_t$, we define $\Psi_{i,Q}$ as the following XOR formula with terms corresponding to $(r+1)$-chains in $\cH^{(r+1,Q)}$ with $x_Q$ ``modded out'' from the corresponding monomial.
\begin{flalign*}
\Psi_{i,Q}(x) =   \sum_{C = (i, v_1, v_2, u_1, \dots, u_{r+1}) \in \cH_i^{(r+1,Q)}}  \wt_{\cH_{i}^{(r+1)}}(C) \cdot x_{v_1} x_{v_2} \prod_{h = 1}^{r} x_{\{v_{2h + 1}, v_{2h + 2}\} \setminus Q_h} \enspace.
\end{flalign*}
Here, we use the convention that if $Q_h = \star$, then $\{v, v'\} \setminus Q_h \coloneqq \{v,v'\}$.

For each $0 \leq t \leq r$,  let $\Psi^{(t)}(x,y) = \sum_{i=1}^k \sum_{Q \in P_t} b_i y_{Q}  \Psi_{i,Q}(x)$. Finally, we let $\Psi(x,y) = \sum_{0 \leq t \leq r} \Psi^{(t)}(x,y)$; here, for every heavy $Q \in P_t$ for some $0 \leq t \leq r$ used in the contiguously regular partition, we introduce a new variable $y_{Q}$.
\end{definition}

We next observe that $\Psi(x,y)$ is a relaxation of the polynomial $\Psi(x)$. Indeed, we have abused notation and labeled them both as ``$\Psi$'' for this reason. This follows from the observation is that $\Psi(x,y)$ is produced by simply replacing the monomial $x_{Q}$ in $\Psi(x)$ with a new variable $y_{Q}$ for each heavy $Q$. More formally, the following holds.
\begin{lemma}
\label{lem:val-lb-labeled}
Fix $x \in \Fits^n$. Then, there is a $y \in \Fits^{\sum_{t = 0}^r \abs{P_t}}$ such that $\Psi(x,y) = \Psi(x)$.
\end{lemma}
\begin{proof}
For each $0 \leq t \leq r$, set $y_{Q} = x_Q$ for every $Q \in P_t$, where $x_Q \coloneqq \prod_{h : Q_h \ne \star} x_{Q_h}$. 
\end{proof}

We finish this section by proving the following statement, which intuitively shows that the partitions of the chains are smooth.
\begin{lemma}[Smoothness of partitioned chains]
\label{lem:chainsmoothness}
Fix $i \in [k]$ and $t \in \{0, \dots, r\}$. Let $Z \in ([n] \cup \{\star\})^{r+1} \times [n]$. Then, $\sum_{C \in \cH_i^{(r+1)} : Z \subseteq C} \wt_{\cH_i^{(r+1)}}(C) \leq (\delta n)^{-\abs{Z}}$.

Let $Q \in P_t$ and $\cH_i^{(r+1,Q)}$ be as defined in \cref{def:inducedpartition}. Let $Z \in ([n] \cup \{\star\})^{r+1} \times [n]$ be such that $Z$ extends $Q$, i.e., $Z_{r - t + h} = Q_h$ for all $1 \leq h \leq t + 1$. Then, $\sum_{C \in \cH_i^{(r+1,Q)} : Z \subseteq C} \wt_{\cH_i^{(r+1)}}(C)$ is at most $\wt(Q) d^{\abs{Z} - \abs{Q}} (\delta n)^{-\abs{Z} - 1 + \abs{Q}}$ if $\abs{Z} \leq r+1$, and at most $(\delta n)^{-r - 1}$ if $\abs{Z} = r+2$. Furthermore, if $d^{r+1} \geq n$, then $(\delta n)^{-r - 1} \leq \wt(Q) d^{\abs{Z} - \abs{Q}} (\delta n)^{-\abs{Z} - 1 + \abs{Q}}$.  
\end{lemma}
\begin{remark}
We remark that this is place where we need the assumption that $d^{r+1} \geq n$.
\end{remark}
 \begin{proof}
The first statement follows immediately by $\delta$-smoothness of the original hypergraphs. Indeed, for any $u \in [n]$ and $v \in [n]$, we have that $\sum_{C \in [n]^3 : v \in C} \wt_{H_u}(C) \leq 1/\delta n$. We now have that
\begin{flalign*}
&\sum_{C \in \cH_i^{(r+1)} : Z \subseteq C} \wt_{\cH_i^{(r+1)}}(C) \\
&\leq  \sum_{\substack{(v_{1}, v_{2}, u_1) \\ Z_1 \in \{v_{1}, v_{2}\}}} \wt_{H_{i}}(v_{1}, v_{2}, u_{1}) \cdot \Big( \sum_{\substack{(v_{3}, v_{4}, u_2) \\ Z_2 \in \{v_{3}, v_{4}\}}} \wt_{H_{u_1}}(v_{3}, v_{4}, u_{2})  \Big( \cdots \Big(  \sum_{\substack{(v_{2r+1}, v_{2r+2}, u_{r+1}) \\ Z_r \in \{v_{2r+1}, v_{2r+2}\}}} \wt_{H_{u_{r}}}(v_{2r+1}, v_{2r+2}, u_{r+1}) \Big) \cdots \Big) \Big) \mper
\end{flalign*}
We notice that the $h$-th term is at most $1/\delta n$ if $Z_h \ne \star$, and otherwise it is at most $1$. So, in total, we get a bound of $(\delta n)^{-\abs{Z}}$.

We now prove the second part of the statement.
 Let $\abs{Q} = t+1$. We have two cases. 
 
 \parhead{Case 1: $Z$ does not contain a $\star$ entry.} This means that $\abs{Z} = r + 2$. Let $Z' \in [n]^{r+1} \times \{\star\}$ be $Z$ with the last entry replaced by a $\star$, i.e., $Z'_h = Z_h$ for all $1 \leq h \leq r + 1$, and $Z'_{r+2} = \star$. We observe that 
 \begin{flalign*}
&\sum_{C \in \cH_i^{(r+1,Q)} : Z \subseteq C} \wt_{\cH_i^{(r+1)}}(C) \leq \sum_{C \in \cH_i^{(r+1)} : Z \subseteq C} \wt_{\cH_i^{(r+1)}}(C) \leq \sum_{C \in \cH_i^{(r+1)} : Z' \subseteq C} \wt_{\cH_i^{(r+1)}}(C) \leq (\delta n)^{-\abs{Z'}} = (\delta n)^{-r - 1} \mcom
 \end{flalign*}
 where we use the first statement that we have already shown. To finish the argument in this case, we need to argue that $\wt(Q) d^{\abs{Z} - \abs{Q}} (\delta n)^{-\abs{Z} - 1 + \abs{Q}} \geq (\delta n)^{-r - 1}$.
 Indeed, we have by definition that $\wt(Q) \geq n d^{\abs{Q} - 1} (\delta n)^{-\abs{Q}}$, and so
 \begin{flalign*}
\wt(Q) d^{\abs{Z} - \abs{Q}} (\delta n)^{-\abs{Z} - 1 + \abs{Q}} \geq \frac{1}{\delta} d^{\abs{Z} - 1} (\delta n)^{-\abs{Z}} = (\delta n)^{-r - 1} \cdot \frac{1}{\delta^2 n} d^{r+1} \mper
 \end{flalign*}
 Thus, the desired inequality holds if $d^{r+1} \geq n$.

 \parhead{Case 2: $Z$ contains a $\star$ entry.} This means that $\abs{Z} \leq r + 1$. Then, we have that $Z = (Z^{(1)}, \star, Z^{(2)}, Q)$, where $Z^{(2)}$ contains no $\star$ entries. 
 
 We observe that each $C \in \cH_i^{(r+1,Q)}$ with $Z \subseteq C$ can be split into $3$ parts: $C = (i, C^{(1)}, C^{(2)}, C^{(3)})$, where $C^{(3)} \in \cT^{(Q)}$ is a length $t$ chain, $(i, C^{(1)})$ is a length $\abs{Z^{(1)}}$ chain with head $i$, and $C^{(2)}$ is a length $r - t - \abs{Z^{(1)}}$ chain whose head is the tail of $C^{(1)}$ and whose tail is the head of $C^{(3)}$. By $\delta$-smoothness, $\sum_{C^{(1)} : Z^{(1)} \subseteq C^{(1)}} \wt_{\cH_i^{(\abs{Z^{(1)}})}}(i, C^{(1)}) \leq (\delta n)^{-\abs{Z^{(1)}}}$.
 
We either have that $(Z^{(2)}, Q)$ is $Q$, i.e., $Z^{(2)}$ is empty, or that $(Z^{(2)}, Q)$ is not $Q$. In the first case, $\sum_{C^{(3)} \in \cT^{(Q)}} \wt_{\cH^{(t)}}(C^{(3)}) = \wt(Q)$ by definition (note that if $t = 0$, then $C^{(3)}$ is just the single vertex $v$ where $Q = (v)$, and we have defined $\wt(Q) = 1$). In the second case, we observe that by \cref{def:partition,def:inducedpartition}, $(Z^{(2)}, Q)$ cannot be heavy. Indeed, if it was, then either $C^{(3)} \in \cT^{(Z^{(2)}, Q)}$, and so $C \in \cH_i^{(r+1, (Z^{(2)}, Q))}$, or else there is some other $Q'$ with $\abs{Q'} = \abs{Z^{(2)}} + t + 1$ with $C^{(3)} \in \cT^{(Q')}$, in which case we would have $C \in \cH_i^{(r+1, Q')}$. We note that here we must use that $Z$ contains at least one $\star$, so that $\abs{Z^{(2)}} + \abs{Q} \leq r + 1$. This is because all heavy $Q'$ have $\abs{Q'} \leq r + 1$, as they are defined for the length $r$-chains.

Thus, $(Z^{(2)}, Q)$ cannot be heavy. It then follows that $\sum_{C^{(3)} : C^{(3)} \notin \cT^{(Q')} \ \forall Q' \in P_{t + \abs{Z^{(2)}}}} \wt_{\cH^{(t)}}(C^{(3)}) \leq n d^{\abs{Z^{(2)}} + t} (\delta n)^{-\abs{Z^{(2)}} - t - 1} \leq \wt(Q) d^{\abs{Z^{(2)}}} (\delta n)^{- \abs{Z^{(2)}}}$. We note that any $C \in \cH_i^{(r+1,Q)}$ must have $C^{(3)} \notin \cT^{(Q')} \ \forall Q' \in P_{t + \abs{Z^{(2)}}}$, as otherwise we would violate Item (2) in \cref{def:partition} since $\abs{Z^{(2)}} \geq 1$.

 To finish the proof, we observe that once $C^{(1)}$ and $C^{(3)}$ are chosen, the total weight of all ``valid'' $C^{(2)}$, i.e., $C^{(2)}$'s that could complete the chain to form $C \in \cH_i^{(r+1,Q)}$, is at most $1/\delta n$. Indeed, this is because the head of $C^{(2)}$ is the tail of $C^{(1)}$ and its tail is the head of $C^{(3)}$, and the total weight of all length $h$ chains, for any $h$, with a fixed head $u$ and fixed tail $v$ is at most $1/\delta n$ by $\delta$-smoothness. Thus, in total, we have shown that $\sum_{C \in \cH_i^{(r+1,Q)}} \wt_{\cH_i^{(r+1)}}(C) \leq (\delta n)^{-\abs{Z^{(2)}}} \cdot (\delta n)^{-1} \cdot \wt(Q) d^{\abs{Z^{(2)}}} (\delta n)^{- \abs{Z^{(2)}}} = \wt(Q) \cdot d^{\abs{Z} - \abs{Q}} (\delta n)^{-\abs{Z} - 1 + \abs{Q}}$.
 \end{proof}

\section{Spectral Refutation via Kikuchi Matrices}
\label{sec:kikuchi}
In \cref{sec:decomp}, we defined polynomials $\Psi^{(t)}(x,y)$ and a map from $x \mapsto y$ such that $\Psi(x) = \sum_{t = 0}^r \Psi^{(t)}(x,y)$ when $y$ is the image of $x$ under this map. Thus, to prove \cref{lem:chainxorref}, we need to upper bound $\E_b[\val (\sum_{t = 0}^r \Psi^{(t)}(x,y))]$. In this section, we will use the Kikuchi matrix method to bound this quantity, proving the second half of \cref{lem:chainxorref}.
\subsection{Step 1: the Cauchy--Schwarz trick} First, we show that we can relate $\sum_{t = 0}^r \Psi^{(t)}(x,y)$ to a certain ``cross-term'' polynomial obtained via applying the Cauchy--Schwarz inequality.

\begin{lemma}[Cauchy--Schwarz trick] \label{lem:cauchy-schwarz}
Let $M$ be a maximum directed matching\footnote{A directed matching is a matching, only the edges are additionally directed}\footnote{This is a perfect matching if $k$ is even, and will leave one element of $[k]$ unmatched if $k$ is odd.} of $[k]$ and let $f_M$ be the cross-term polynomial defined as
\begin{flalign*}
&f_M^{(t)} = \sum_{\{i ,  j\} \in M} b_i b_j \sum_{Q \in P_t}\frac{1}{\wt(Q)}\Psi_{i,Q}(x) \Psi_{j,Q} (x) \mcom \\
&f_M = \sum_{t = 0}^r f^{(t)}_M \mper
\end{flalign*}
Then for every $x,y$ with $\pm 1$ values, it holds that
\begin{flalign*}
&\left(\sum_{t = 0}^r \Psi^{(t)}(x,y)\right)^2 \leq n(r+1)\Paren{\frac{k(r+1)}{\delta^2 n}  + 2k\E_M[f_M]} \mcom\end{flalign*}
where the expectation $\E_{M}$ is over a uniformly random maximum directed matching $M$.
\end{lemma}
\begin{proof}
We will first apply the Cauchy--Schwarz inequality to eliminate the $y$ variables:
\begin{align*}
&\left(\sum_{t = 0}^r \Psi^{(t)}(x,y)\right)^2 = \Paren{\sum_{t = 0}^r \sum_{Q \in P_t} y_{Q} \cdot \sqrt{\wt(Q)} \Paren{\sum_{i=1}^k b_i \frac{\Psi_{i,Q}}{\sqrt{\wt(Q)}}}}^2\\
&\leq \Paren{\sum_{t = 0}^r \sum_{Q \in P_t} y_Q^2 \wt(Q)} \Paren{\sum_{t = 0}^r \sum_{Q \in P_t}\left(\sum_{i=1}^k b_i \frac{\Psi_{i,Q}}{\sqrt{\wt(Q)}}\right)^2}\\
&\leq \Paren{(r+1)n}\Paren{\sum_{t = 0}^r \sum_{Q \in P_t}\left(\sum_{i=1}^k b_i \frac{\Psi_{i,Q}}{\sqrt{\wt(Q)}}\right)^2}\\
&\leq n(r + 1)\Paren{\sum_{t = 0}^r \sum_{Q \in P_t} \frac{1}{\wt(Q)} \sum_{i, j = 1}^k b_i b_j \Psi_{i,Q} \Psi_{j,Q}} \\
&\leq n(r+1)\Paren{\sum_{t = 0}^r \sum_{Q \in P_t} \frac{1}{\wt(Q)} \sum_{i = 1}^k \Psi_{i,Q}^2 + \sum_{t = 0}^r \sum_{Q \in P_t} \frac{1}{\wt(Q)} \sum_{i \ne j \in [k]} b_i b_j \Psi_{i,Q} \Psi_{j,Q}} \mper
\end{align*}
By \cref{lem:chainsmoothness}, we have that
\begin{flalign*}
&\abs{\Psi_{i,Q}(x)} \leq  \sum_{C \in \cH_i^{(r+1,Q)}} \wt_{\cH_i^{(r+1)}}(C) \leq \wt(Q) \cdot (\delta n)^{-1} \mcom
\end{flalign*}
Hence, $\sum_{t = 0}^r \sum_{Q \in P_t} \frac{1}{\wt(Q)} \sum_{i = 1}^k \Psi_{i,Q}^2 \leq \frac{k}{\delta^2 n^2} \sum_{t = 0}^r \sum_{Q \in P_t} \wt(Q) \leq \frac{k(r+1)}{\delta^2 n}$.
\later{\peter{note: we can remove the factor of $r$ here if we count a bit more carefully}}

To finish the proof, we observe that the probability that a pair $(i,j)$ is contained in a directed matching $M$ is at least $\frac{1}{2k}$.
\end{proof}

\subsection{Step 2: defining the Kikuchi matrices} It thus remains to bound $\E_b[\val(f_M)]$ for an arbitrary directed maximum matching $M$. 

We define the Kikuchi matrices that we consider below.
\begin{definition}
\label{def:kikuchi}
Let $i,j \in [k]$ and $t \in \{0, \dots, r\}$. Let $Q \in P_t$.

Let $C = (i, v_1, v_2, u_1, v_3, v_4, \dots, u_{r+1}) \in \cH^{(r+1, Q)}_{i}$ and $C' = (j, v'_1, v'_2, u_1, v'_3, v'_4, \dots, u_{r+1}) \in \cH^{(r+1,Q)}_{j}$. We let $A_{i,j}^{(C, C', Q)} \in \Bits^{{{[n]} \choose \ell}^{2r+2}}$ be the matrix with rows and columns by indexed by $(2r+2)$-tuples of sets $(S_0, \dots, S_{r}, S'_0, \dots, S'_{r})$  of size exactly $\ell$ defined as follows.

We set $A_{i,j}^{({C}, {C'}, Q)}((S_0, \dots, S_{r}, S'_0, \dots, S'_{r}), (T_0,  \dots, T_{r}, T'_0, \dots T'_{r}))$ equal to $1$ if the following holds, and otherwise we set this entry to be $0$. In what follows, we let $C_h = \{v_{2h+1}, v_{2h+2}\}$, and we note that $\abs{C_h} = 2$ for any chain with nonzero weight, by \cref{def:hypergraph}.
\begin{enumerate}
\item For $h = 0, \dots, r - t$, we have $S_h \oplus T_h = C_h$ and $v_{2h+1} \in S_h$, $v_{2h+2} \in T_h$.
\item For $h = 0, \dots, r - t$, we have $S'_h \oplus T'_h = C'_h$ and $v'_{2h+1} \in S'_h$, $v'_{2h+2} \in T'_h$,
\item For $h = 1, \dots, t$, the following holds. Let $w_h = C_{r - t + h} \setminus Q_h$, and $w'_h = C'_{r - t + h} \setminus Q_h$. We have $S_{r - t + h} = R \cup \{w_h\}$, $T_{r - t + h} = R \cup \{w'_h\}$, and $S'_{r - t + h} = T'_{r - t + h}$.\footnote{It is possible that one could have $w_h = w'_h$ here. In that case, we pick a canonical extra vertex $v$, and require that $v \notin R$ as well. This is to ensure that the number of choices here for $S_{r - t + h}$ and $S'_{r - t + h}$ is \emph{exactly} ${n - 2 \choose \ell - 1} {n \choose \ell}$; otherwise it would be ${n - 1 \choose \ell -1} {n \choose \ell}$. The difference in the two cases is immaterial but it is convenient to have an exact count.}
\end{enumerate}
We let $A^{(t)}_{i,j} = \sum_{Q \in P_t} \frac{1}{\wt(Q)} \sum_{C \in \cH^{(r+1,Q)}_{i} ,{C'} \in \cH^{(r+1,Q)}_{j}} \wt_{\cH^{(r+1)}_i}(C)\wt_{\cH^{(r+1)}_j}(C') \cdot A_{i,j}^{({C}, {C'}, Q)}$ and $A_{i,j} = \sum_{t = 0}^r \frac{1}{D_t} A_{i,j}^{(t)}$, where $D_t = {n - 2 \choose \ell - 1}^{2r + 2 - t} \cdot {n \choose \ell}^t$. For any matching $M$ on $[k]$, let $ A_M = \sum_{(i,j) \in M} b_i b_j A_{i,j}$. We will abuse notation and let $A \defeq A_M$.
\end{definition}

The following lemma shows that we can express $f_M(x)$ as a (scaling of a) quadratic form on the matrix $A^{(t)}$.
\begin{lemma}
\label{lem:completeness}
Let $x \in \Fits^n$, and let $x' \in \Fits^{N}$, where $N = {n \choose \ell}^{2r + 2}$, denote the vector where the $(S_0, S_1, \dots, S_{r}, S'_0, S'_1, \dots, S'_{r})$-th entry of $x'$ is $\prod_{h = 0}^{r} x_{S_h} x_{S'_h}$.
Let $i,j \in [k]$ and $t \in \{0, \dots, r\}$. Let $Q \in P_t$, and let
let $C = (i, v_1, v_2, u_1, v_3, v_4, \dots, u_{r+1}) \in \cH^{(r+1, Q)}_{i}$ and $C' = (j, v'_1, v'_2, u_1, v'_3, v'_4, \dots, u_{r+1}) \in \cH^{(r+1,Q)}_{j}$. Then,  
\begin{flalign*}
{x'}^{\top} A_{i,j}^{(C, C', Q)} x' = D_t x_{v_1} x_{v_2} \prod_{h = 1}^{r} x_{\{v_{2h + 1}, v_{2h + 2}\} \setminus Q_h} \cdot x_{v'_1} x_{v'_2} \prod_{h = 1}^{r} x_{\{v'_{2h + 1}, v'_{2h + 2}\} \setminus Q_h} \mcom
\end{flalign*}
 i.e., the product of the monomials associated to $C$ and $C'$, modded out by $Q_h$, where $D_t = {n - 2 \choose \ell - 1}^{2r + 2 - t} \cdot {n \choose \ell}^t$.
Moreover, for any matrix $B_{i,j}^{(C, C', Q)}$ obtained by ``zeroing out'' exactly $\alpha D_t$ entries of $A_{i, j}^{(C, C', Q)}$, the equality holds with a factor of $1 - \alpha$ on the right.

 In particular, ${x'}^{\top} A x' = f_M(x)$.
\end{lemma}
\begin{proof}
Let  $\vec{S} = (S_0, S_1, \dots, S_{r}, S'_0, S'_1, \dots, S'_{r})$ and $\vec{T} = (T_0,  \dots, T_{r}, T'_0, \dots T'_{r})$ be such that $A_{i,j}^{(\vec{C},\vec{C'}, Q)}(\vec{S},\vec{T}) = 1$. 
Then, we have that
\begin{flalign*}
x'_{\vec{S}} x'_{\vec{T}} &= \prod_{h = 0}^{r} x_{S_h} x_{T_h} x_{S'_h} x_{T'_h}  =  \prod_{h = 0}^{r-t} x_{S_h \oplus T_h} x_{S'_h \oplus T'_h} \prod_{h = 1}^{t} x_{S_{r - t + h} \oplus T_{r - t + h}} x_{S'_{r - t + h} \oplus T'_{r - t + h}} \\
&= \prod_{h = 0}^{r-t} x_{C_h} x_{C'_h} \prod_{h = 1}^{t} x_{C_{r - t+h} \setminus Q_h} x_{C'_{r - t+h} \setminus Q_h} \enspace,
\end{flalign*}
which is equal to the product of monomials on the right-hand side of the equation we wish to show.

It thus remains to argue that $A_{i,j}^{(\vec{C},\vec{C'}, Q)}$ has exactly $D_t$ nonzero entries. We observe that, for each $h = 0, \dots, r - t$, there are exactly ${n - 2 \choose \ell - 1}$ pairs $(S_h, T_h)$ such that $S_h \oplus T_h = C_h$ with $v_{2h + 1} \in S_h$ and $v_{2h + 2} \in T_h$. Indeed, this is because we must simply choose a set of size $\ell - 1$ that does not contain either of $v_{2h + 1}$ and $v_{2h + 2}$, and then this determines $S_h$ and $T_h$.

For $h = 1, \dots, t$, there are exactly ${n - 2 \choose \ell - 1}$ choices of $(S_{r - t + h}, T_{r - t + h})$. Indeed, this is because $S_{r - t + h}$ must contain $w_h$ and $T_{r - t + h}$ must contain $w'_h$. Note that if $w_h = w'_h$, then there are actually ${n - 1 \choose \ell - 1}$ choices! However, using the slightly modified definition of the matrix in the footnote in \cref{def:kikuchi}, we can again force there to be exactly ${n - 2 \choose \ell - 1}$ choices. Finally, there are ${n \choose \ell}$ choices for $(S'_{r - t + h}, T'_{r - t + h})$, as we must have $S'_{r - t + h} = T'_{r - t + h}$.

Combining, we see that $D_t =  {n - 2 \choose \ell - 1}^{2(r - t + 1)} \cdot ({n - 2 \choose \ell - 1}{n \choose \ell})^t = {n - 2 \choose \ell - 1}^{2r + 2 - t} {n \choose \ell}^t$, as required.
\end{proof}

\subsection{Step 3: finding a regular submatrix of the Kikuchi matrix}
\label{sec:regularsubmatrix}
By \cref{lem:completeness}, in order to upper bound $\E_b[\val(f_M)]$, it suffices to bound $\E_b[\boolnorm{A}] \leq N \E_b[\norm{A}_2]$, where $N = {n \choose \ell}^{2r + 2}$; here, we use that $\boolnorm{A} \leq N \norm{A}_2$ always holds.

To bound $\norm{A}_2$, we will write $A = \sum_{(i,j) \in M} b_i b_j A_{i,j}$ and apply \cref{fact:matrixkhintchine}. To do this, we need to bound $\norm{A_{i,j}}_2$, which we shall do by upper bounding the maximum $\ell_1$-norm of any row/column of the matrix. In turns out there are some rows that indeed have a large $\ell_1$-norm. To handle this issue, we shall zero out the ``bad rows'', as follows. To do this, we will need to use the following technical lemma, proven in \cref{sec:rowpruning}, that bounds the expected $\ell_1$-norm of a row and the conditional expectation given that the row has a nonzero entry in a specific matrix $A_{i,j}^{(C,C', Q)}$.

\begin{restatable}[First and conditional moment bounds]{lemma}{RowPruning}
\label{lem:rowpruning}
Fix $r \geq 1$, $i,j \in [k]$, and let $\cH_i^{(r+1)}$ and $\cH_j^{(r+1)}$ denote the $(r+1)$-chain hypergraph with heads in $i$ and $j$ respectively. Let $\cup_{t = 0}^r \cup_{Q \in P_t} \cH_i^{(r+1,Q)}$ be a smooth partition of $\cH_i^{(r+1)}$, as defined in \cref{def:partition,def:inducedpartition}. Let $A_{i,j}$ be the Kikuchi matrix defined in \cref{def:kikuchi}, which depends on $r$, $i$, $j$, and the pieces $\cup_{Q \in P_t} \cH^{(r+1,Q)}$ of the refinement, and the matching $M$.

Let $\vec{S} = (S_0, \dots, S_r, S'_0, \dots, S'_r) \in {[n] \choose \ell}^{2r + 2}$ be a row of the matrix, and let $\deg_{i,j}(\vec{S})$ denote the $\ell_1$-norm of the $\vec{S}$-th row of $A_{i,j}$. Then, 
\begin{equation*}
\E_{\vec{S}}[\deg_{i,j}(\vec{S})] \leq \frac{1}{N \cdot \delta n} \mcom
\end{equation*}
where $N = {n \choose \ell}^{2r + 2}$.

Furthermore, let $t \in \{0, \dots, r\}$, $Q \in P_t$, and $C \in \cH_i^{(r+1, Q)}$ and $C' \in \cH_j^{(r+1, Q)}$. Let $\cD_{C,C',Q}$ denote the uniform distribution over rows of $A_{i,j}^{(C,C', Q)}$ that contain a nonzero entry. Then, if $d^{r+1} \geq n$ and $\ell \geq 2d(r+1)/\delta$, it holds that 
\begin{equation*}
\E_{\vec{S} \sim \cD_{C,C',Q}}[\deg_{i,j}(\vec{S})] \leq \left(1 + \frac{O(\ell r)}{n} \right) \cdot \frac{4}{N \delta n} \mper
\end{equation*}
\end{restatable}

Let us now use \cref{lem:rowpruning} to argue the following. For a sufficiently large constant $\Gamma$, there exist submatrices $B_{i,j}^{(C, C', Q)}$, i.e., a $\Bits$-matrix where $B_{i,j}^{(C, C', Q)}(\vec{S}, \vec{T}) = 1$ implies $A_{i,j}^{(C, C', Q)}(\vec{S}, \vec{T}) = 1$, such that \begin{inparaenum}[(1)] \item each $B_{i,j}^{(C, C', Q)}$ contains exactly $D_t / 2$ nonzero entries, and \item the $\ell_1$-norm of any row/column of $B_{i,j}$ (defined analogously to $A_{i,j}$) is at most $ \frac{\Gamma}{N \cdot \delta n}$.\end{inparaenum}

We do this as follows. First, we observe that $A_{i,j}^{(C,C', Q)}(\vec{S}, \vec{T}) = A_{j,i}^{(C', C, Q)}(\vec{R}, \vec{W})$, where $\vec{R} = (S'_0, \dots, S'_r, S_0, \dots, S_r)$ and $\vec{W} = (T'_0, \dots, T'_r, T_0, \dots, T_r)$. In particular, this symmetry implies that the bounds on the moments for rows in \cref{lem:rowpruning} hold for columns as well.

Let $\cB_1 = \{\vec{S} : \deg_{i,j} (\vec{S}) \geq \frac{\Gamma}{N \cdot \delta n}\}$ denote the set of bad rows with $\ell_1$-norm at least $\frac{\Gamma}{N \cdot \delta n}$, and similarly let $\cB_2$ be the same but for the columns. Applying Markov's inequality and the conditional degree bound, we see that $\cB_1$ contains at most $O(1/\Gamma)$-fraction of the rows where $A_{i,j}^{(C,C', Q)}$ is nonzero, and similarly $\cB_2$ contains at most $O(1/\Gamma)$-fraction of the columns where $A_{i,j}^{(C,C', Q)}$ is nonzero. Thus, after removing these rows, we still have at least $(1 - O(1/\Gamma)) D_t$ nonzero entries in $A_{i,j}^{(C,C', Q)}$. When $\Gamma$ is a sufficiently large constant, this is at least $1/2$, and so we can choose an arbitrary subset of \emph{exactly} $D_t/2$ nonzero entries. We let $B_{i,j}^{(C,C', Q)}$ be the matrix with those nonzero entries. 

The first property is clearly satisfied by construction. The second property is satisfied because the $\ell_1$-norm of any row/column of $B_{i,j}$ is clearly at most $\frac{\Gamma}{N \cdot \delta n}$, again by construction.

\subsection{Step 4: finishing the proof}
Let $B_{i,j}^{(C,C',Q)}$ be the matrix produced in \cref{sec:regularsubmatrix}. 

We let $B^{(t)}_{i,j} = \sum_{Q \in P_t} \frac{1}{\wt(Q)} \sum_{C \in \cH^{(r+1,Q)}_{i} ,{C'} \in \cH^{(r+1,Q)}_{j}} \wt_{\cH^{(r+1)}_i}(C)\wt_{\cH^{(r+1)}_j}(C') \cdot B_{i,j}^{({C}, {C'}, Q)}$ and $B_{i,j} = \sum_{t = 0}^r \frac{1}{D_t} B_{i,j}^{(t)}$. For any matching $M$ on $[k]$, let $ B_M = \sum_{(i,j) \in M} b_i b_j B_{i,j}$. We will abuse notation and let $B \defeq B_M$.

By \cref{lem:completeness} and the fact that $B_{i,j}^{(C,C', Q)}$ has exactly $D_t/2$ nonzero entries of $A_{i,j}^{(C,C',Q)}$ in it, we see that for every $x \in \Fits^n$, there exists $x' \in \Fits^N$ such that ${x'}^{\top} B x' = \frac{1}{2} f_M(x)$. We also have that $\norm{B_{i,j}}_2 \leq \frac{\Gamma}{N \cdot \delta n}$, by construction in \cref{sec:regularsubmatrix}.

By \cref{fact:matrixkhintchine}, it therefore follows that
\begin{flalign*}
\E_b[\val(f_M(x))] \leq 2 \E_b[N \norm{B}_2] \leq N \cdot \frac{\Gamma}{N \cdot \delta n} \cdot O( \sqrt{k \log N}) = O(\sqrt{k \ell r \log n}) \cdot \frac{1}{\delta n}
\end{flalign*}
Hence,
\begin{flalign*}
&\E_b[\val(\Psi(x,y))]^2 \leq \E_b[\val(\Psi(x,y)^2)] \leq n(r+1)\Paren{\frac{k(r+1)}{\delta^2 n}  + 2k\E_{b,M}[\val(f_M)]} \\
&\leq n(r+1)\Paren{\frac{k(r+1)}{\delta^2 n}  + 2k O(\sqrt{k \ell r \log n}) \cdot \frac{1}{\delta n}} = \frac{k(r+1)}{\delta} \Paren{\frac{r+1}{\delta}  + 2 O(\sqrt{k \ell r \log n})} \\
&\leq \frac{k(r+1)}{\delta}  O(\sqrt{k \ell r \log n}) \mcom
\end{flalign*}
as $\ell \geq O(r/\delta)$ and we can assume that $k \geq 1/\delta$ (as otherwise we are already done).

\section{Row Pruning: Proof of \cref{lem:rowpruning}}
\label{sec:rowpruning}
In this section, we prove \cref{lem:rowpruning}, restated below.

\RowPruning*

\begin{proof}
We begin by estimating the first moment, i.e., $\E_{\vec{S}}[\deg_{i,j}(\vec{S})]$. By definition, we have that 
\begin{flalign*}
&\E_{\vec{S}}[\deg_{i,j}(\vec{S})] = \frac{1}{N} \sum_{t = 0}^r \frac{1}{D_t} \sum_{Q \in P_t} \frac{1}{\wt(Q)} \sum_{C \in \cH^{(r+1,Q)}_{i} ,{C'} \in \cH^{(r+1,Q)}_{j}} \wt_{\cH^{(r+1)}_i}(C)\wt_{\cH^{(r+1)}_j}(C') \cdot D_t \\
&= \frac{1}{N} \sum_{t = 0}^r  \sum_{Q \in P_t} \frac{1}{\wt(Q)} \sum_{C \in \cH^{(r+1,Q)}_{i} ,{C'} \in \cH^{(r+1,Q)}_{j}} \wt_{\cH^{(r+1)}_i}(C)\wt_{\cH^{(r+1)}_j}(C') \mper
\end{flalign*}
We note that the latter quantity is simply equal to $\frac{1}{N} \sum_{C \in \cH_i^{(r+1)}} \wt_{\cH^{(r+1)}_i}(C) \sum_{C' \in \cH_j^{(r+1,Q)} : C \in \cH_i^{(r+1, Q)}} \frac{1}{\wt(Q)} \cdot \wt_{\cH^{(r+1)}_j}(C')$, where the second sum is over $C' \in \cH_j^{(r+1,Q)}$ where $Q$ is determined by the choice of $C$. We note that for any $Q$, $\sum_{C' \in \cH_j^{(r+1,Q)}}  \wt_{\cH^{(r+1)}_j}(C') \leq \frac{\wt(Q)}{\delta n}$, and hence we conclude that
\begin{equation*}
\E_{\vec{S}}[\deg_{i,j}(\vec{S})] \leq \frac{1}{N} \sum_{C \in \cH_i^{(r+1)}} \wt_{\cH^{(r+1)}_i}(C) \frac{1}{\delta n} \leq \frac{1}{N \cdot \delta n} \mper
\end{equation*}

\later{\peter{could add a note that we could use the simple bounds from the inequalities we have established but this loses an additional factor of $r$}}

Next, we estimate the conditional first moment. Fix a $Q \in P_t$ for some $0 \leq t \leq r$, and let $C \in \cH^{(r+1,Q)}_{i} ,{C'} \in \cH^{(r+1,Q)}_{j}$. We now bound $\E_{\vec{S} \sim \cD_{C, C', Q}}[\deg_{i,j}(\vec{S})]$, where $\cD_{C,C',Q}$ is the uniform distribution over all rows $\vec{S}$ such that $A_{i,j}^{(C,C',Q)}$ has a nonzero entry. We note that there are exactly $D_t$ such rows.

We shall proceed in two steps. First, we consider a fixed $(D, D', Q')$ with $D \in \cH^{(r+1,Q')}_{i} ,{D'} \in \cH^{(r+1,Q')}_{j}$. Let $\abs{Q'} = t' + 1$. We will upper bound the number of rows $\vec{S}$ where $A_{i,j}^{(C,C', Q)}$ and $A_{i,j}^{(D,D',Q')}$, normalized by the factor of $1/D_{t'}$. This will depend on the number of shared vertices $z$ between these two pairs of chains, for an appropriate definition of shared vertices. Then, we will, for each choice of $z$, bound the total weight of the number of chains $(D, D', Q')$ have ``intersection $z$'' with $(C,C',Q)$, which will conclude the argument.

\parhead{Step 1: bounding the normalized number of entries for a fixed $(D,D',Q')$.} To begin, we will define the number of ``shared vertices'' between two pairs of chains $(C,C',Q)$ and $(D,D',Q')$.

\begin{definition}[Left vertices]
Let $(C,C',Q)$ be such that $Q \in P_t$ and $C \in \cH^{(r+1,Q)}_{i} ,{C'} \in \cH^{(r+1,Q)}_{j}$. Let $C = (i, v_1, v_2, u_1, \dots, u_{r+1})$ and $C' = (j, v'_1, v'_2, u'_1, \dots, u'_{r+1})$. The tuple of \emph{left vertices} of $(C,C',Q)$ is the sequence $(v_1, v_3, v_5, \dots, v_{2(r - t) + 1}, w_1, \dots, w_t, v'_1, v'_3, \dots, v'_{2(r - t) + 1})$, where $C_h = \{v_{2h + 1}, v_{2h + 2}\} = \{w_h, Q_h\}$. We denote this sequence by $L(C,C', Q)$. 
\end{definition}
\begin{remark}
\label{rem:intersectionmotivation}
The reason for the above definition is the following. If $\vec{S}$ is a row where the matrix $A_{i,j}^{(C,C',Q)}$ has a nonzero entry, then the entries of $L(C,C',Q)$ (in order) are contained in the sets $(S_0, \dots, S_{r - t}, S_{r - t + 1}, \dots, S_{r}, S'_0, \dots, S'_{r - t})$, e.g., $v_1 \in S_0$, $v_3 \in S_1$, $w_1 \in S_{r - t + 1}$, etc.
\end{remark}
\begin{definition}[Intersection patterns]
Let $(C,C',Q)$ and $(D,D', Q')$ be such that $C \in \cH^{(r+1,Q)}_{i},{C'} \in \cH^{(r+1,Q)}_{j}$ and $D \in \cH^{(r+1,Q')}_{i} ,{D'} \in \cH^{(r+1,Q')}_{j}$. 

The \emph{intersection pattern} of $(C,C',Q)$ and $(D,D', Q')$, given by $Z \in \Bits^{2r + 2 - t}$, is defined as $Z_h = 1$ if $L(C,C',Q)_h = L(D,D',Q')_h$, and it is $0$ otherwise. Note that the sequences $L(C,C',Q)$ and $L(D,D',Q')$ may not have the same length; if $h$ is ``out of bounds'' for $L(D,D',Q')$, then we set $Z_h = 0$.
\end{definition}

We now fix $(D,D', Q')$ and count the number of rows as a function of the intersection pattern $Z$. Let $t' = \abs{Q'} - 1$. We have two cases. In the first case, $t \geq t'$, which implies that $\abs{L(C,C',Q)} \leq \abs{L(D,D',Q)}$.
We observe that in order for a row $\vec{S}$ to have a nonzero entry for both pairs of chains, the following must hold:
\begin{enumerate}
\item for $h = 1, \dots, r + 2$ (the first $r+1$ sets), we have $\{L(C,C',Q)_h, L(D,D',Q)_h\} \subseteq S_h$,
\item for $h = r+2, \dots, 2r + 3 - t$ (the next $r + 1 - t$ sets), we have $\{L(C,C',Q)_h, L(D,D',Q)_h\} \subseteq S'_{h - (r+2)}$,
\item for $h = 2r + 3 - t, \dots, 2r + 2 - t'$ (the next $t - t'$ sets), we have $L(D,D',Q)_h \in S'_{h - (r+2)}$,
\item for $h = 2r + 2 - t' + 1, \dots, 2r + 2$ (the final $t'$ sets), we have $S'_{h - (r + 2)}$ is arbitrary.
\end{enumerate}
We observe that for each intersection point, i.e., an $h$ such that $L(C,C',Q)_h =  L(D,D',Q)_h$, there are ${n \choose \ell - 1}$ choices for the corresponding set, as it needs to only contain one vertex. For each nonintersection point, i.e., an $h \in \{1, \dots, 2r + 2 - t\}$ where $L(C,C',Q)_h \ne  L(D,D',Q)_h$, we have ${n \choose \ell - 2}$ choices, because the set needs to contain both vertices. Finally, we have ${n \choose \ell - 1}$ choices for each of the $t - t'$ sets in the third case, and ${n \choose \ell}$ choices for the last $t$ sets in the final case. In total, we have ${n \choose \ell - 1}^z {n \choose \ell - 2}^{2r + 2 - t - z} {n \choose \ell - 1}^{t - t'} {n \choose \ell}^{t'}$.

In the second case, $t \leq t'$. We observe that by swapping the roles of $t$ and $t'$ above, we get a bound of ${n \choose \ell-1}^z {n \choose \ell - 2}^{2r + 2 - t' - z}{n \choose \ell - 1}^{t' - t} {n \choose \ell}^{t}$.

Now, although the above counts are different, we observe that they are within constant factors of each other. Indeed, we have 
\begin{flalign*}
&\frac{{n \choose \ell - 2}^{-t'}{n \choose \ell - 1}^{t' - t} {n \choose \ell}^{t}}{{n \choose \ell - 2}^{-t}{n \choose \ell - 1}^{t - t'} {n \choose \ell}^{t'}} = \left({n \choose \ell - 2}^{-1}{n \choose \ell - 1}^{2} {n \choose \ell}^{-1}\right)^{t' - t} \\
&= \left(\frac{\ell(n - \ell + 2)}{(\ell - 1)(n - \ell +1)} \right)^{t' - t} = \left( 1 + \frac{n - 1}{(\ell - 1)(n - \ell + 1)} \right)^{t' - t} \mcom
\end{flalign*}
and this ratio is between $\frac{1}{2}$ and $2$ since $\abs{t' - t} \leq r$ and $\frac{n - 1}{(\ell - 1)(n - \ell + 1)} \geq \frac{2}{\ell} \geq \frac{1}{\Gamma r}$ for a sufficiently large constant $\Gamma$.

Next, we observe that while we have an upper bound of $2 \cdot {n \choose \ell-1}^z {n \choose \ell - 2}^{2r + 2 - t}{n \choose \ell - 1}^{t - t'} {n \choose \ell}^{t'}$ on the number of rows, which depends on $t'$, each entry has a scaling factor of $\frac{1}{D_{t'}}$. We now give an upper bound on the \emph{normalized} number of entries that does not depend on $t'$. We have
\begin{flalign*}
&2\frac{{n \choose \ell-1}^z {n \choose \ell - 2}^{2r + 2 - t - z}{n \choose \ell - 1}^{t - t'} {n \choose \ell}^{t'}}{D_{t'}} = 2\frac{{n \choose \ell-1}^z {n \choose \ell - 2}^{2r + 2 - t - z}{n \choose \ell - 1}^{t - t'} {n \choose \ell}^{t'}}{{n - 2 \choose \ell - 1}^{2r + 2 - t'} \cdot {n \choose \ell}^{t'}} =   2\left( \frac{ {n \choose \ell - 2} }{ {n \choose \ell - 1} } \right)^{ 2r + 2 - t - z}  \cdot \left( \frac{ {n \choose \ell - 1} }{ {n - 2 \choose \ell - 1} }\right)^{2r + 2 - t'} \\
&=  2\left( \frac{\ell - 1}{n - \ell + 2} \right)^{2r + 2 - t - z} \cdot \left( \frac{n(n-1)}{(n - \ell + 1)(n - \ell)}\right)^{2r + 2 - t'} \\
&\leq 2\left( \frac{\ell}{n} \right)^{2r + 2 - t - z} \cdot \left(1 + \frac{O(\ell r)}{n}\right) \mper
\end{flalign*}

\parhead{Step 2: bounding the weight of $(D,D',Q')$ with a fixed intersection pattern $Z$.}
Let us fix the intersection pattern $Z$ and then determine the total weight of all $(D, D', Q')$ with $D \in \cH^{(r+1,Q')}_{i} ,{D'} \in \cH^{(r+1,Q')}_{j}$ with these intersection points. To do this, we will apply \cref{lem:chainsmoothness}. 

First, we observe that fixing an intersection pattern induces a $Z^{(1)} \in \{[n]\cup \{\star\}\}^{r+1} \times \{\star\}$, simply by filling in $Z^{(1)}$'s non-$\star$ entries with the appropriate vertices of $L(C,C',Q)$. We note that such a $Z^{(1)}$ never has the tail filled in, as the tail is not a potential intersection point. By \cref{lem:chainsmoothness}, this implies that the total weight of $D$ that contain $Z^{(1)}$ is at most $(\delta n)^{-\abs{Z^{(1)}}}$.

Next, we bound the total weight of all $D'$ that are valid for a fixed $D$. We observe that $D \in \cH_i^{(r+1,Q')}$ for some $i$, and hence $D'$ must be in $\cH_j^{(r+1,Q')}$. We note that $Z$ induces an intersection pattern $Z^{(2)}$ on $D'$, and moreover $Z^{(2)}$ does not intersect with the ``$Q'$-part'' of the chain $D'$, namely the links that contain vertices from $Q'$. So, it follows that $D'$ contains $(Z^{(2)}, Q')$.

By \cref{lem:chainsmoothness}, we have that the total weight of all $D'$ is at most $\wt(Q) d^{\abs{Z^{(2)}}} (\delta n)^{-\abs{Z^{(2)}} - 1}$. As each entry in $A_{i,j}^{(D,D', Q')}$ is scaled down by a factor of $\wt(Q')$, the normalized weight is therefore at most $d^{\abs{Z^{(2)}}} (\delta n)^{-\abs{Z^{(2)}} - 1}$.

In total, we get a bound of $(\delta n)^{-\abs{Z^{(1)}}} \cdot d^{\abs{Z^{(2)}}} (\delta n)^{-\abs{Z^{(2)}} - 1}$, which is at most $d^{\abs{Z}} (\delta n)^{-\abs{Z} - 1}$. Here, we use that $\abs{Z} = \abs{Z^{(1)}} + \abs{Z^{(2)}}$.

\parhead{Putting it all together.} By combining steps (1) and (2) (and paying an additional ${2r + 2 - t \choose z}$ factor to choose the nonzero entries of $Z$), we thus obtain the final bound of
\begin{flalign*}
&\E_{\vec{S} \sim \cD_{C, C', Q}}[\deg_{i,j}(\vec{S})] \leq \frac{1}{D_t} \sum_{z = 0}^{2r + 2 - t} {2r + 2 - t \choose z} \cdot \left(1 + \frac{O(\ell r)}{n}\right) \cdot 2\left(\frac{\ell}{n}\right)^{2r + 2 - t - z} \cdot d^{z} (\delta n)^{-z - 1} \\
&\leq \left(1 + \frac{O(\ell r)}{n}\right)\frac{2}{D_t} \left(\frac{\ell}{n}\right)^{2r + 2 - t} \cdot  \sum_{z = 0}^{2r + 2 - t} (2r + 2 - t)^z \cdot \left(\frac{\ell}{n}\right)^{- z} \cdot d^{z} (\delta n)^{-z - 1} \\
&= \left(1 + \frac{O(\ell r)}{n}\right)\frac{2}{D_t \cdot \delta n} \left(\frac{\ell}{n}\right)^{2r + 2 - t}  \cdot \sum_{z = 0}^{2r + 2 - t}  \left( \frac{(2r + 2 - t) \cdot d}{\delta \ell} \right)^{z} \\
&\leq \left(1 + \frac{O(\ell r)}{n}\right) \frac{4}{D_t \cdot \delta n} \left(\frac{\ell}{n}\right)^{2r + 2 - t} \mcom
\end{flalign*}
where we use that $\ell \geq 2d(2r + 2)/\delta$.

To finish the proof, we need to compute $\frac{D_t}{N}$. We have that
\begin{flalign*}
&\frac{D_t}{N} = \frac{{n - 2 \choose \ell - 1}^{2r + 2 - t} \cdot {n \choose \ell}^t}{{n \choose \ell}^{2r + 2}} = \left( \frac{{n - 2 \choose \ell - 1}}{{n \choose \ell}}\right)^{2r + 2 - t} = \left( \frac{\ell (n - \ell)}{n(n-1)}\right)^{2r + 2 - t} \\
&\geq \left(\frac{\ell}{n}\right)^{2r + 2 - t} \cdot \left( 1 - \frac{\ell - 1}{n - 1}\right)^{2r + 2 - t} \geq \left(\frac{\ell}{n}\right)^{2r + 2 - t} \left(1 - \frac{(\ell - 1)(2r + 2)}{n-1}\right) = \left(\frac{\ell}{n}\right)^{2r + 2 - t} \left(1 - \frac{O(\ell r)}{n}\right) \mcom
\end{flalign*}
Thus, 
\begin{flalign*}
&\E_{\vec{S} \sim \cD_{C, C', Q}}[\deg_{i,j}(\vec{S})] \leq \left(1 + \frac{O(\ell r)}{n}\right) \frac{4}{D_t \cdot \delta n} \left(\frac{\ell}{n}\right)^{2r + 2 - t} \leq \left(1 + \frac{O(\ell r)}{n}\right) \frac{4}{N \cdot \delta n} \mcom
\end{flalign*}
which finishes the proof.
\end{proof}


\bibliographystyle{alpha}
{\small
\bibliography{references}
}

\appendix

\ignore[imperfect completeness appendix]{
\section{The Case of Imperfect Completeness in \cref{mthm:nonlin}}
\label{app:imperfectcompleteness}
In this appendix, we prove \cref{mthm:nonlin} when the code has completeness $1 - \eps$. The proof is essentially identical to the proof in the perfect completeness case presented in \cref{sec:chainpolys,sec:chainxor,sec:graphref,sec:decomp,sec:kikuchi,sec:rowpruning}, with only minor changes that we describe here.

First, we observe that the reduction in \cref{sec:adaptivesmooth} holds with the following minor change: the decoding function $f_{v_1, a_1, v_2, a_2}(a_3)$ for $C = (v_1, a_1, v_2, a_2, v_3)$ might not be deterministic. This means that the function $f_{v_1, a_1, v_2, a_2}(a_3)$ is a convex combination of the deterministic functions specified in \cref{sec:adaptivesmooth}, and so we may need to add multiple copies of $C = (v_1, a_1, v_2, a_2, v_3)$ with different weights to handle the different deterministic functions in the convex combination. This only introduces some minor issues with notation.

The key change that we need to make lies in \cref{claim:chaincompleteness}. We no longer have that the chain polynomials correctly decode $x_u$ for every $x \in \Code$. In fact, we can see that, by the ``chain decoder'' interpretation of the adaptive chains given in \cref{sec:techsnonlin}, the chain polynomial computes the advantage of the chain decoder $\Dec_{r+1}^{x}(u)$ when decoding $x_u$, namely $\E[x_u \Dec_{r+1}^{x}(u)]$, where the expectation is over the internal randomness of the chain decoder. In this case, by union bound, the chain decoder is correct with probability at least $1 - (r+1) \eps$, and so $\E[x_u \Dec_{r+1}^{x}(u)] \geq 1 - 2(r+1) \eps$.

Now, when we use \cref{lem:chainpolyref} to refute the chain polynomial instances, we set parameters as follows. Let $\eta > 0$ to be chosen later, and set $r_0$ be such that $r_0 +1 = \floor{\frac{1}{2 \eps} - \eta}$ and $r_1$ be such that $r_1 = \gamma \sqrt{\log _2 n}$ for some constant $\gamma$ to be chosen later. We then let $r = \min(r_0, r_1)$. Note that by choice of $r$, $\frac{1}{2 \eps} - \eta \geq r+1$, and so $1 - 2(r+1)\eps \geq 2\eta \eps$.

Now, we set $d$ to be such that $d^{r+1} \geq n$, so we have to set $d = n^{1/{r+1}}$. Finally, we set $\ell = d r / \delta$. Following the calculations, we thus get that either
\begin{flalign*}
\eta^2 \eps^2 k \leq r^2 2^{O(r)} O(\ell r \log n) = \frac{O(1)}{\delta} r^6 2^{O(r)} d \log n \mcom
\end{flalign*}
or
\begin{flalign*}
\eta^4 eps^4 k \leq \frac{r^6}{\delta^2} 2^{O(r)} O(\ell r \log n) = \frac{r^8}{\delta^3}2^{O(r)} d \log n \mper
\end{flalign*}
The second equation is always the dominant term. If $\eps \leq \gamma/\sqrt{\log_2 n}$, then we observe that we are simply in the same parameter regime as in the perfect completeness, and we get the same bound. If $\eps \geq 2\gamma /\sqrt{\log_2 n}$, then we have that 
\begin{flalign*}
 k \leq \frac{r^8}{\delta^3 \eta^4 \eps^4}2^{O(r)} n^{\frac{1}{r+1}} \log n \mper
\end{flalign*}
Taking $\gamma$ large enough and $\eta = O(1/\log n)$ implies that $(\delta^3 \eps^4 k) \leq \tilde{O}(n^{\frac{1}{r+1}})$. This finishes the proof, as $r+1 =  \floor{\frac{1}{2 \eps} - \eta}$. Note that the final $\log(1/\delta)$ loss comes from \cref{fact:bgt17}.
}

\section{Design $3$-LCCs over $\F_2$ from Reed--Muller Codes}
\label{app:reedmuller}
In this appendix, we give a simple folklore construction of a design $3$-LCCs (\cref{def:designs}) using Reed--Muller codes. 

\begin{lemma}[Design $3$-LCCs over $\F_2$ from Reed--Muller Codes]
\label{lem:designreedmuller}
Let $t$ be an integer, and let $k = 1 + t + {t \choose 2}$. Then, there is a design $3$-LCC with blocklength $n = 4^t$ of dimension $k$. In particular, $n \leq 2^{2 \sqrt{2k}}$.
\end{lemma}
To prove this lemma, we will need the following fact about polynomials over $\F_4$.
\begin{fact}
\label{fact:polyinterpolation}
Let $f(x) = \alpha_0 + \alpha_1 x + \alpha_2 x^2$ be a degree-$2$ polynomial over $\F_4$. Then, $\sum_{\beta \in \F_4} f(\beta) = 0$.
\end{fact}
\begin{proof}
Recall that the field $\F_4$ is equivalent to the polynomial ring $\F_2[\beta]$ modulo the equation $\beta^2 + \beta + 1 = 0$. We have 
\begin{flalign*}
&f(0) = \alpha_0 \\
&f(1) = \alpha_0 + \alpha_1 + \alpha_2 \\
&f(\beta) = \alpha_0 + \alpha_1 \beta + \alpha_2 \beta^2 \\
&f(1 + \beta) = \alpha_0 + \alpha_1 (1 + \beta) + \alpha_2 (1 + \beta)^2 \\
&\implies f(0) + f(1) + f(\beta) + f(1 + \beta) = \alpha_0 \cdot 4 + \alpha_1 \cdot 2(1 + \beta) + \alpha_2 (1 + \beta^2 + (1 + 2 \beta + \beta^2)) \\
&= 0 \mcom
\end{flalign*}
as $2 = 0$ in $\F_4$. 
\end{proof}

\begin{proof}[Proof of \cref{lem:designreedmuller}]
We will define the code in two stages. First, we will define, via an encoding map, a code over $\F_4$ with the desired dimension argue that it is a design $3$-LCC. Then, we will use this code to construct a code over $\F_2$.

Let $\cV$ denote the vector space of degree $\leq2$ polynomials over $\F_4$ in $t$ variables $x_1, \dots, x_t$. We note that $\cV$ has dimension $k$. 

For each $b \in \F_4^k$, we encode $b$ by (1) letting $f_b(x_1, \dots, x_t)$ be the degree-$2$ polynomial with coefficients given by $b$, and (2) evaluating $f_b$ over all $x \in \F_4^t$; this yields an output $Z \in {\F_4}^{4^t} = \F_4^n$, which is the encoding $\Enc(b)$. We note that $\Enc$ is clearly an $\F_4$ linear map.

We now argue that this encoding map is a design $3$-LCC. Indeed, we need to define a system of constraints such that for every pair $x^{(0)},x^{(1)} \in \F_4^t$, there is a unique constraint containing $x^{(0)},x^{(1)}$. Let $x^{(\beta)} = x^{(0)} + \beta(x^{(1)} - x^{(0)})$ and $x^{(1 + \beta)} = x^{(0)} + (1 + \beta)(x^{(1)} - x^{(0)})$. We note that $x^{(0)}, x^{(1)}, x^{(\beta)}$ and $x^{(1 + \beta)}$ is the line $L(t) = x^{(0)} + \lambda(x^{(1)} - x^{(0)})$ containing $x^{(0)}, x^{(1)}$. Fix $b \in \F_4^k$, and let $f_b$ be the corresponding polynomial. We know that $g(\lambda) = f_b(L(\lambda))$ is a degree-$2$ univariate polynomial in $\lambda$. Hence, by \cref{fact:polyinterpolation}, it follows that $f_b(x^{(0)}) + f_b(x^{(1)}) + f_b(x^{(\beta)}) + f_b(x^{(1 + \beta)}) = 0$. Hence, for each pair $x^{(0)}, x^{(1)} \in \F_4^t$, there exists a constraint containing this pair, and moreover, because two points determine a line, any constraint containing this pair must be exactly this line. Thus, the code given by $\Enc$ is a design $3$-LCC.

We now use the above code to construct a binary code. Let $\Tr \colon \F_4 \to \F_2$ be the trace map. We let $\cV'$ be the image of $\cV$ under $\Tr$ (applied element-wise to each vector in $\cV$). We note that because $\cV$ has dimension $k$ over $\F_4$ is a linear code, it is systematic, meaning that there is a subset $S \subseteq \F_4^t$ such that $\cV \vert_S = \F_4^k$. Therefore, because the trace map is identity on $\F_2$, it follows that $\cV' \vert_S = \F_2^k$, i.e., that $\cV'$ has dimension $k$ also.

To finish the proof, we need to argue that $\cV'$ is a design $3$-LCC. Let $g \in \cV'$. We will show that for each line $x^{(0)},x^{(1)}, x^{(\beta)}, x^{(1 + \beta)}$ in $\F_4^t$ as defined earlier, it holds that $g(x^{(0)}) + g(x^{(1)}) + g(x^{(\beta)}) + g(x^{(1 + \beta)}) = 0$. Indeed, we have that $g = \Tr(f)$ for some $f \in \cV$, and that $f(x^{(0)}) + f(x^{(1)}) + f(x^{(\beta)}) + f(x^{(1+\beta)}) = 0$. Because all the coefficients in the linear constraint are $1$, i.e., they are in $\F_2$, the constraint still holds after applying $\Tr(\cdot)$, as this is an $\F_2$-linear map. Thus, the constraint holds, which finishes the proof.
\end{proof}

\end{document}